\documentclass{article}

\usepackage{PRIMEarxiv}

\usepackage[utf8]{inputenc} 
\usepackage[T1]{fontenc}    
\usepackage{url}            
\usepackage{booktabs}       
\usepackage{amsfonts}       
\usepackage{nicefrac}       
\usepackage{microtype}      
\usepackage{lipsum}
\usepackage{fancyhdr}       
\usepackage{graphicx}       
\graphicspath{{media/}}     

\usepackage{cite}
\usepackage{amsmath}
\usepackage{amssymb,amsfonts}
\usepackage{algorithmic}
\usepackage{graphicx}
\usepackage{textcomp}

\usepackage[utf8]{inputenc} 
\usepackage[T1]{fontenc}    
\usepackage{hyperref}       
\usepackage{url}            
\usepackage{booktabs}       
\usepackage{nicefrac}       
\usepackage{microtype}      
\usepackage[dvipsnames]{xcolor}    
\usepackage{xr}
\usepackage{mathtools}
\usepackage{amsthm}
\usepackage[ruled,vlined]{algorithm2e}
\usepackage{mathrsfs}
\usepackage{footnote}
\usepackage{wrapfig}
\usepackage[caption=false, font=footnotesize]{subfig}
\usepackage{float}
\usepackage{cancel}

\usepackage{hyperref}       
\hypersetup{
    colorlinks = true,
    citecolor = {red},
    linkcolor = blue
}
\usepackage{authblk}

\newcommand{\dimMatCom}[0]{K}
\newcommand{\dimMatF}[0]{n}
\newcommand{\dimMatC}[0]{m} 
\newcommand{\matF}[0]{\ensuremath{\mathbf{A}}}
\newcommand{\matC}[0]{\ensuremath{\wt{\matF}}}
\newcommand{\matLinCoarse}[0]{\ensuremath{\mathbf{W}}}
\newcommand{\vecLinCoarse}[0]{\ensuremath{\mathbf{w}}}
\newcommand{\graphG}[0]{\ensuremath{\mathcal{G}}}
\newcommand{\graphF}[0]{\ensuremath{\mathcal{G}_\mathrm{fine}}}
\newcommand{\graphC}[0]{\ensuremath{\mathcal{G}_\mathrm{coarse}}}

\newcommand{\matCom}[0]{\ensuremath{\mathbf{Q}}}

\newcommand{\comFine}[0]{\ensuremath{\mathbf{P}}}

\newcommand{\syncComAs}[0]{\ensuremath{\mathbf{\Phi}}}

\newcommand{\protoCovSize}[0]{\ensuremath{r}}

\newcommand{\stateVecF}[0]{\ensuremath{\mathbf{x}}}
\newcommand{\stateVecC}[0]{\ensuremath{\wt{\stateVecF}}}
\newcommand{\controlVec}[0]{\ensuremath{\mathbf{u}}}

\newcommand{\CtrlMatF}[0]{\ensuremath{\mathbf{B}}}

\newcommand{\genMat}[0]{\ensuremath{\mathbf{Z}}}
\newcommand{\genVec}[0]{\ensuremath{\mathbf{z}}}
\newcommand{\genMatOth}[0]{\ensuremath{\mathbf{Y}}}
\newcommand{\genVecOth}[0]{\ensuremath{\mathbf{y}}}
\newcommand{\scalingSBM}[0]{\ensuremath{\rho_{\dimMatF}}}
\newcommand{\specRadius}[0]{\ensuremath{\rho}}
\newcommand{\Oscale}[0]{\ensuremath{\mathcal{O}}}
\newcommand{\constFine}[0]{\ensuremath{\epsilon}}
\newcommand{\constCoarse}[0]{\ensuremath{\constFine}} 
\newcommand{\diagComSize}[0]{\ensuremath{\mathbf{D}}}
\newcommand{\comSizes}[0]{\ensuremath{\mathbf{d}}}

\newcommand{\matFnom}[0]{\ensuremath{\matF_\mathrm{nom}}}
\newcommand{\matFBarNom}[0]{\ensuremath{\bar{\matF}_\mathrm{nom}}}
\newcommand{\matCnom}[0]{\ensuremath{\matC_\mathrm{nom}}}
\newcommand{\matCBarNom}[0]{\ensuremath{\bar{\matC}_\mathrm{nom}}}
\newcommand{\overalConstTemp}[0]{\ensuremath{\iota}}
\newcommand{\overalConstTempCoarse}[0]{\ensuremath{\tilde{\overalConstTemp}}}
\newcommand{\dirichletPar}[0]{\ensuremath{{\beta}}}
\newcommand{\indexRandomSet}[0]{\ensuremath{\mathcal{I}}}
\newcommand{\hoeffProbConst}[0]{\ensuremath{\delta}}
\newcommand{\hoeffCoeffConst}[0]{\ensuremath{\eta}}
\newcommand{\constUnifSampl}[0]{\ensuremath{\eta_\circ}}
\newcommand{\probUnifSampl}[0]{\ensuremath{\delta_\circ}}

\newcommand{\maxFunctionInn}[0]{\ensuremath{{f}}}

\newcommand{\errorMat}[0]{\ensuremath{\mbf{L}}}
\newcommand{\inCtrlMat}[0]{\ensuremath{\mbf{\Upsilon}}}
\DeclareMathSymbol{\shortminus}{\mathbin}{AMSa}{"39}

\newcommand{\genDim}{\ensuremath{\dimMatF}}

\newcommand{\transpose}{\mathsf{T}} 
\newcommand{\bs}{\boldsymbol}
\newcommand{\wt}{\widetilde}
\newcommand{\mbf}{\mathbf}

\newcommand{\trace}{\mathrm{tr}}
\newcommand{\Cgram}{\bs{\mathcal{C}}}

\newtheorem{thm}{Theorem}

\newtheorem{prop}[thm]{Proposition}

\newtheorem{exmp}{Example}
\newtheorem{assump}{Assumption}

\newtheorem{remark}{Remark}

\newtheorem{lemma}{Lemma}

\newtheorem{definition}{Definition}

\title{Controllability of Coarsely Measured Networked Linear Dynamical Systems (Extended Version)\thanks{This paper is based upon
    work supported by by the Natural Sciences and Engineering Research Council (NSERC) of Canada, including through a Discovery Research Grant; the National Science Foundation under awards OAC-1934766,  CCF-2029044, and CCF-2048223; and the National Institutes of Health under the award 1R01GM140468-01.}
}

\author[$\P, \ast$]{Nafiseh Ghoroghchian}
\author[$\dagger$]{Rajasekhar Anguluri}
\author[$\dagger$]{Gautam Dasarathy}
\author[$\P$]{Stark C. Draper}
\affil[$\P$]{ Edward S. Rogers Sr. Dept. of Electrical and Computer Engineering, University of Toronto, Toronto, Canada}
\affil[$\dagger$]{School of Electrical, Computer and Energy Engineering, Arizona State University, Tempe, AZ, USA}
\affil[$\ast$]{Vector Institute for Artificial Intelligence, Toronto, Canada}
\date{\today}

\begin{document}
\maketitle
\begin{abstract}
  We consider the controllability of large-scale linear networked dynamical systems when complete knowledge of network structure is unavailable and knowledge is limited to coarse summaries. 
  We provide conditions under which average controllability 
  of the fine-scale system can be well approximated by average controllability of the (synthesized, reduced-order) coarse-scale system. To this end, we require knowledge of some inherent parametric structure of the fine-scale network that makes this type of approximation possible. Therefore, we assume that the underlying fine-scale network is generated by the stochastic block model (SBM)---often studied in community detection. We then provide an algorithm that directly estimates the average controllability of the fine-scale system using a coarse summary of SBM. Our analysis indicates the necessity of underlying structure (e.g., in-built communities) to be able to quantify accurately the controllability from coarsely characterized networked dynamics. We also compare our method to that of the reduced-order method and highlight the regimes where both can outperform each other. Finally, we provide simulations to confirm our theoretical results for different scalings of network size and density, and the parameter that captures how much community-structure is retained in the coarse summary. 
  \end{abstract}

	\section{Introduction}\label{sec:intro}

	 In this paper we study controllability of networks with discrete-time linear dynamics (hereafter, LTI systems) when our knowledge of system
        structure is limited to coarse summaries. We are motivated by
        myriad real-world settings, ranging from power and water systems to brain networks, where system identification must be
        performed based upon data collected by low-resolution
        instruments unable to probe fine-scale structure.  Our
        motivating real-world example is the human brain.  While efforts are
        under way to produce a map of a canonical human brain, our
        knowledge of the brain as an interconnected, network system is
        not yet to the level of the individual neurons of the entire brain
        \cite{betzel2017multi}.  And yet, modern
        medical technologies aim to control the brain using coarse MRI based network connectivity { \cite{gu2015controllability,karrer2020practical}}.  For example, novel brain implants
        designed for epilepsy patients \emph{steer} the brain away
        from states that correspond to seizures
        \cite{heck2014two}. 
        The inaccessibility of fine-scale network information may result in wrong control decisions that diverge greatly from those derived from the fine-scale network.
        
       Our goal is to quantify the controllability of a fine-scale networked LTI system using its coarse summaries and a few of the fine-scale network properties such as community structure (see below). In the presence of full knowledge of the fine-scale system, existing works address this goal by first obtaining a lower-dimensional (coarse-scale) system using various methods under the umbrella term of model order reduction (MOR). Then, they relate the controllability of this coarse-scale system to that of the original fine-scale system \cite{benner15}. But this approach is inapplicable in settings 
       wherein one does not have access to the fine-scale network. Nonetheless, we may fit a dynamical system for coarse network data and study the controllability of this system as a proxy for the controllability of the fine-scale system. However, it is not clear when such an approximation is sufficiently accurate. In Section \ref{sec: group-vs-coarse}, we study this approximation in detail. 
       The second contrast between classical MOR and our setting is best understood in terms of \textit{active} versus \textit{passive} MOR.
       Classical MOR is active in that it determines how to coarsen the system to yield the best reduction.  But for us,
        our knowledge is restricted by the precision of our
        observations (e.g., owing to the resolution of measurement devices or to anatomical considerations), thereby limiting our ability to devise best reduction schemes. So, passively collected network data is our starting point.
        
        To realize our goal, we consider a suitable controllability Gramian based metric. It is well known that the Gramian based metrics help compare the control effort exerted either by two individual nodes or groups of nodes \cite{pasqualetti2014controllability,yuan2013exact, baggio2022,lindmark2020}. We focus on the latter to address the following question. For certain groups of nodes that we can control globally (i.e., not actuating each node individually) suppose that we 
        have access only to their \emph{coarse summary}; that is, certain combinations of weights of network edges in and across the groups of nodes. Then, we ask which groups influence (or require minimum energy) in controlling the fine-scale system? This joint measuring and actuating is required in epilepsy applications to determine the device (instrument) position to collapse the unstable brain-state oscillations that can lead to seizures {\cite{o2018nurip,kassiri2017closed}}. 
    
To leverage fine-scale network properties (e.g., community structure, sparsity, and node degree) along with coarse summaries, we assume that the fine-scale network is generated by a stochastic block model (SBM). These models find applications in studying communities in social networks and recently a few groups have considered SBMs for network dynamics \cite{schaub2020blind, xing2020community}. { It is known that SBMs naturally allow for generating networks with different community structures (e.g., assortative and dis-assortative) \cite{abbe2017community,betzel2018diversity}. Interestingly, we show that they also naturally allow to study the the relationship between controllability and the \emph{synchronization}}---a measure that quantifies how well the coarse network data are synchronized (or overlapped) with communities. 
The higher the synchronization, the less information the coarse summary contains regarding the fine-scale network (see Section \ref{sec:preliminary}).


To quantify the controllability of the fine-scale LTI system, we consider the \emph{average controllability} metric given by the trace of the controllability Gramian of a discrete-time system with non-negative system matrices \cite{farina2000positive}. Specifically, we characterize the average controllability vector of a set of systems. Each entry of the vector
 corresponds to the average controllability of a system with specific group of input (control) nodes as described above. We then devise two competing schemes (see below) to estimate this average controllability vector and  compare the true and estimated vectors.
Our main contributions are as follows.

		\begin{enumerate}
			\item In Section \ref{sec: group-vs-coarse}, we consider an auxiliary, fictitious, linear dynamical system based on the coarse data, which we term the  \textit{passively reduced-order model (PROM)}.\footnote{The prefix "passive" in  PROM emphasizes that the auxiliary system is not a reduction of the fine-scale model. Rather the auxiliary system is a low-dimensional fictitious model that we fit to the coarse data.} Each  node in the fictitious network corresponding to PROM is associated with a group of  nodes in the fine-scale network that map to the same coarse measurement. We use the average controllability vector of the PROM to approximate that of the fine-scale system. Notice that this implicitly couples measurement and actuation. 
We derive a tight upper bound on the \textit{approximation-error} which goes to zero as the coarse and fine network sizes grow; the network becomes dense; the coverage of the fine network by coarse measurements expands; and the synchronization between the coarse summary data and the underlying community structure increases.  If synchronization is not sufficiently high, the error may not approach zero even as the network size increases.
			\item In Section \ref{sec:est_solution}, we learn the average controllability vector of the fine-scale system directly from the coarse data without considering an auxiliary system. We develop a novel learning-based algorithm with its roots in a mixed-membership method proposed by \cite{mao2020estimating} for unsupervised learning of the parameters and the community structure in an SBM. We derive a tight upper bound on the estimation error and characterize its convergence in terms of network parameters. Using this bound, we show that the estimation error approaches zero as long as the community structure parameter estimates are accurate and the coverage of the fine network by coarse measurements is sufficiently rich. Although the learning-based technique's estimate quality is implicitly dependent on synchronization, it has more resilience than the PROM method in low-synchronized scenarios.
			
			\item In Section \ref{sec:simulations}, we present multiple numerical simulations to show that the estimation errors and their qualitative behavior (in terms of the various network and coarsening parameters) are comparable to the tight error bounds and their qualitative nature we derived in Sections \ref{sec: group-vs-coarse} and \ref{sec:est_solution}. Importantly, consistent with the theory, simulations show that the synchronization has less impact on the estimation error of the learning-based method. Finally, our empirical investigation reveals that both the approaches studied can outperform each other in terms of estimation error. Specifically, the learning-based approach outperforms its counterpart in regions for which the network density and the coarse network size are not too small, and the coarse measurements are not fully synchronized with communities.
		\end{enumerate}

 
\subsection{Related work}
\textit{Controllability of Complex Networks}: Controllability of dynamical systems has a rich, long history since its inception in Kalman's work \cite{kalman1960b}. Several articles and textbook entries were written about its characterizations and manifestations in several applications \cite{klamka2018controllability}. In studies concerning network controllability, the two areas that received great attention are determining (a) the minimal set of control nodes to steer a network to a desired state \cite{pequito2016minimum, olshevsky2014minimal, tzoumas2015minimal, dilip2019, chanekar2021}; and (b) the mimimum control energy associated with a set of (driver) nodes to steer a network to a desired state \cite{vosughi2019target, lindmark2020, baggio2022, pasqualetti2014controllability, klickstein2018control}. The control energy is often quantified using certain scalar metrics of the spectrum of controllability Gramian matrix. Our work falls into the second category and we focus on the average controllability metric.
Capitalizing on the sub-modularity property of several control energy metrics, including average energy, a few studies
developed greedy algorithms to design network topology that has minimum control energy (see \cite{cortesi2014submodularity,  srighakollapu2021optimizing}). Further, in \cite{cortesi14submodular, gu2015controllability}, the authors studied which group of nodes, when actuated as
inputs, can be used to steer the network to an arbitrary target state, and at what cost. Different from the aforementioned works, which focus only on a metric for a group or full set of nodes, we consider a vector of such scalar metrics (specifically, average energy) to study the
comparative influence of different groups of nodes. Ours is the first work to characterize this type of error bounds for the controllability of coarse networks, and we contribute to the nascent field of analyzing controllability/obseravability metrics for random linear systems \cite{preciado2016controllability, bopardikar2021randomized}. 

\textit{MOR for Complex Networks}: There is a matured theory for MOR methods for dynamical systems \cite{benner15, antoulas2005, moore1981principal}. However, MOR methods for network dynamics are limited. We briefly review some of the recent developments and recommend \cite{cheng2021model} for further reading. Methods based on reducing the dimension of individual subsystems in a network were described in \cite{chengrobust2018, MONSHIZADEH20131}, where as methods based on reducing the number of nodes in a network using clustering and projection were described in \cite{jongsmamodel2018, Benner2021, chengreduction2017, Monshizadeh2014}. 
Specifically, the latter methods typically rely on Petrov-Galerkin approximation, where projection matrices are constructed to retain certain structural aspects (e.g., clusters) in the reduced network. Although our coarsening operation relies on a given projection matrix, our goal is not to obtain a reduced-order system that well-approximate, or retains some spatial structure of, the high-dimensional system. Rather we quantify the average controllability of the fine-scale network using a coarse network only. Finally, in the literature, the model reduction error is quantified using the system $H_2$ and $H_\infty$ norms \cite{jongsmamodel2018}. Instead, we use the $\ell_1$-norm to compare the average controllability of a group of nodes in fine-scale and their counterpart in the coarse network. 

\textit{Coarsened SBM as a generative process}:  The
	SBM and its variants provide a powerful modeling
	framework to facilitate fundamental understanding of graph community
	organization and have found several applications, including social and
	power networks (see \cite{dulac2020mixed,funke2019stochastic,abbe2017community} and references therein). Several studies assume the knowledge of fine-scale network, which is a reasonable assumption for networks with a few thousand nodes. Nonetheless, there are several networks (eg., brain networks) with millions of nodes, and not all the network connections are available for a  analyst to do inference (e.g., control or estimation). However, analyst can still have access to summaries of connections (for eg., restrictions in brain signal acquisition from macro-contacts, and imaging data \cite{ghoroghchian2018hierarchical,ghoroghchian2020node}).
While the study of recovering graph structures from coarsened measurements has received some attention in the past \cite{gilbert2004compressing, ahn2012graph, dasarathy2015sketching}, the study of extracting community structure from coarse summaries is
recent. In \cite{ghoroghchian2021graph}, the authors extended the
SBM for community detection \cite{abbe2017community}, to lay out a framework for a coarsened, weighted variant of the SBM. We build off those results
in this paper.

	
  	\noindent \textbf{Notation}: We denote vectors and matrices using bold faced small and upper case letters. For $\genMat=[\genMat_{uv}]\in \mathbb{R}^{\genDim\times \genDim}$, define $\|\genMat\|_\infty=\max_{1\leq u\leq n}\sum_{v=1}^\genDim|\genMat_{uv}|$;  $\|\genMat\|_\mathrm{max}=\max_{u,v}|\genMat_{uv}|$; $\|\genMat\|_{1,1}=\sum_{u=1}^\genDim \sum_{v=1}^\genDim |\genMat_{uv}|$; and $\|\genMat\|_2=\sqrt{\lambda_\mathrm{max}(\genMat^\transpose\genMat)}$, where $\lambda_\mathrm{max}$ is the maximum eigenvalue. Denote the spectral radius of $\mbf{Z}$ by $\rho(\genMat)=\max_i\{|\lambda_i(\genMat)|\}$, where $\lambda_i$ is the $i$-th eigenvalue; $\mathrm{diag}(\genMat)=[\genMat_{11}, \ldots, \genMat_{\genDim\genDim}]^\transpose\in \mathbb{R}^{\genDim}$; and $\mathrm{Diag}(\genMat)$ sets the off-diagonal entries of $\genMat$ to zero. 
  	 The matrix inequality $\genMat_1\leq \genMat_2$ implies element wise inequality. 
  	 Denote $\genDim$ dimensional all ones and all zero vectors by $\mbf{1}_\genDim$ and $\mbf{0}_\genDim$. 
  	   For $\genVec=[\genVec_u]\in\mathbb{R}^\genDim$, define $\mathrm{supp}(\genVec)$ as the set of indices $i$ for which $\genVec_i\ne 0$;  $\|\genVec\|_1=\sum_{u=1}^\genDim |\genVec_u|$.
  	   For a positive integer $\genDim$, we denote $[\genDim]\triangleq\{1,\ldots,\genDim\}$. The
  	 cardinality of a set $\mathcal{V}$ is denoted by $|\mathcal{V}|$.  Vector $\mbf{e}_i$ is a standard basis vector containing all $0$'s except for its $i$-th element that is $1$.
  	   The Hadamard (element-wise) product is denoted by $\odot$. 
  	   We write $f(\genDim)=\Oscale(h(\genDim))$ iff there exist positive reals $c_0$ and $\genDim_0$ such that $|f(\genDim)|\leq c_0 h(\genDim)$ for all $\genDim\geq \genDim_0$. Conversely, $f(\genDim)=\Omega(h(\genDim))$ iff there exist $c_0,\dimMatF_0>0$ such that $\vert f(\genDim)\vert \geq c_0 h(\genDim)$ for all $\genDim\geq \genDim_0$.
  	   $[a,b]$ shows the interval (set of real numbers) between scalar values $a$ and $b$.

	\section{Preliminaries and Problem Statement}\label{sec:preliminary}
	 In this section, we introduce a class of stable LTI systems that evolve on networks with randomly generated edge weights. We then describe how we use the SBM to generate random \emph{fine} networks with community structure, and
how we model the coarse measurements of these networks. We then formulate our problem statement.
	\subsection{Dynamical systems and controllability}\label{sec: linear systems}
 A network is an un-directed graph $\graphG\triangleq(\mathcal{V}, \mathcal{E})$, with node set $\mathcal{V}\triangleq \{1,\ldots,n\}$ and edge set $\mathcal{E}\subseteq \mathcal{V}\times \mathcal{V}$. 
	Each edge $(u,v)\in \mathcal{E}$ has a weight $\mathbf{A}_{uv}=\mathbf{A}_{vu}\in \mathbb{R}$. The \textit{weighted symmetric adjacency matrix} of $\graphG$ is
	$\mathbf{A}\triangleq[\mathbf{A}_{uv}]$, where $\mathbf{A}_{uv}=\mathbf{A}_{uv}=0$ whenever $\mathbf{A}_{uv}\notin \mathcal{E}$. A random network is an un-directed graph where the $[\mathbf{A}_{uv}]$ are random variables. 
	For a network $\graphG$ with adjacency matrix $\matF$, associate state $x_i[t] \in \mathbb{R}$ to the $i$-th node at time $t$. The states evolve according to LTI dynamics as 
	\begin{align}\label{generalLTI}
\left.
\stateVecF[t+1]  = \matF_\mathrm{nom} \stateVecF[t] + {\mathbf{B}} \controlVec[t]\right., \quad \quad \forall\,\,\, t\in \mathbb{N}. 
	\end{align}
	The state $\stateVecF[t]=[x_1[t],\ldots,x_n[t]]^\transpose$ is steered by the input $\mbf{u}[t]\in \mathbb{R}^{n}$. The normalized state adjacency matrix $\matF_\mathrm{nom}\triangleq \matF/{(\constFine +\specRadius(\matF))}$ is such that 
	the spectral radius $\rho(\matF_\mathrm{nom})<1, \forall \constFine>0$; see \cite{karrer2020practical}.
The input matrix $\mbf{B}=\mathrm{Diag}(\mbf{b})$ with $\mbf{b} \in  \{0,1\}^{\dimMatF}$, where we call $\mathcal{K}=\mathrm{supp}(\mbf{b})\subset [\dimMatF]$ the \textit{control node set}, determines which nodes receive control input.
 Note that $\mbf{B}\mbf{u}[t]$ is equal to ${\genMat}\mbf{u}_{\mathcal{K}}[t]$, where ${\genMat}$ is a rectangular matrix composed of columns of $\mbf{B}$ indexed by $\mathcal{K}$, and $\mbf{u}_{\mathcal{K}}$ is a ${\mathcal{K}}$-length vector. We stick with the former notation in \eqref{generalLTI} to enhance readability of notations in the next sections.
For instance, for $\mbf{B}=\mathrm{Diag}(\mbf{1}_{n_1},\mbf{0}_{n-n_1})$, the input enters the network through control nodes set $\mathcal{K} =\{1,\ldots,n_1\}$. 
    
  \emph{Average controllability} is a widely studied metric of controllability {that quantifies the difficulty of controlling a network \cite{muller1972analysis,pasqualetti2014controllability,gu2015controllability}
    defined as }$\trace\left[\Cgram_T(\matF_\mathrm{nom}, \mbf{B})\right]$, where
	\begin{align}\label{gramianFine_def}
		 \Cgram_T(\matF_\mathrm{nom}, \mbf{B}) = \sum_{\ell=0}^{T-1} (\matF_\mathrm{nom})^\ell \CtrlMatF\CtrlMatF^{\transpose} (\matF_\mathrm{nom})^\ell
	\end{align}
	is the $T$-step controllability Gramian of \eqref{generalLTI}. By construction $\Cgram_T(\matF_\mathrm{nom}, \mbf{B})\succeq 0$, and hence, $\trace\left[\Cgram_T(\matF_\mathrm{nom}, \mbf{B})\right]\geq 0$.	When the context is clear, we drop the arguments $\matF_\mathrm{nom}$ and $\mbf{B}$ from $\Cgram_T(\matF_\mathrm{nom}, \mbf{B})$ and write $\Cgram_T$. For ease of exposition, we work with the average controllability for the infinite time horizon Gramian $\Cgram= \lim_{T \to \infty} \Cgram_T$,
	which exists when $\rho(\matF_\mathrm{nom})<1$.
	
	The inverse of the average controllability is a lower bound on
	the minimum energy (in expectation) required to drive the system in \eqref{generalLTI} from the origin to any target state on the unit sphere. This energy would be $\int_{S}\stateVecF^\transpose\Cgram^{-1} \stateVecF\,d\stateVecF/\int_{S} \,d\stateVecF=n^{-1}\trace(\Cgram^{-1})$, where $S=\{\stateVecF:\|\stateVecF\|_2=1\}$ \cite{muller1972analysis,baggio2022}. It follows that $\trace(\Cgram^{-1}) \geq 1/\trace(\Cgram)$. The tightness of the bound is numerically explored in \cite{yan2012controlling, karrer2020practical}. The higher $\trace(\Cgram)$ is for a given set of control nodes (indexed by the non-zero columns in $\mbf{B}$), the smaller the average energy, and the lower the influence that set has on the network. In this light, $\trace(\Cgram)$ can be interpreted as the overall network controllability proxy in all directions in the state space \cite{summers2015submodularity}. Owing to the above discussion and given the ill-posedness of $\Cgram$ for high-dimensional networks, the average controllability provides an easily computable proxy for the average energy \cite{ikeda2018sparsity, pirani2020controllability, srighakollapu2021optimizing}. Furthermore, because $\trace(\Cgram)$ equals the energy of the impulse response (or equivalently, the $H_2$-norm of the network system), the average controllability quantifies the capability of grouped input nodes to drive the network to easy-to-reach state with minimum effort \cite{mcgowan2021controllability, fang2022effects, deng2020controllability}. 
	
	We can interpret the average controllability using the nuclear norm which is the sum of singular values, i.e., the trace of a positive semi-definite matrix \cite{recht2010guaranteed}. The nuclear norm is the convex envelope of the rank function (that is, the number of rank-one matrices needed to obtain a given matrix \cite[page 71]{chi2017convex}). The higher the nuclear norm, meaning the large average controllability, the harder will be to approximate the Gramian or metrics of the Gramian by those of its reduced-order model. This is one reason why our learning-based approach (see Section \ref{sec:est_solution}) that directly estimates the average controllability can perform better.

	\begin{figure*}
		\centering
		\includegraphics[width=0.9\textwidth]{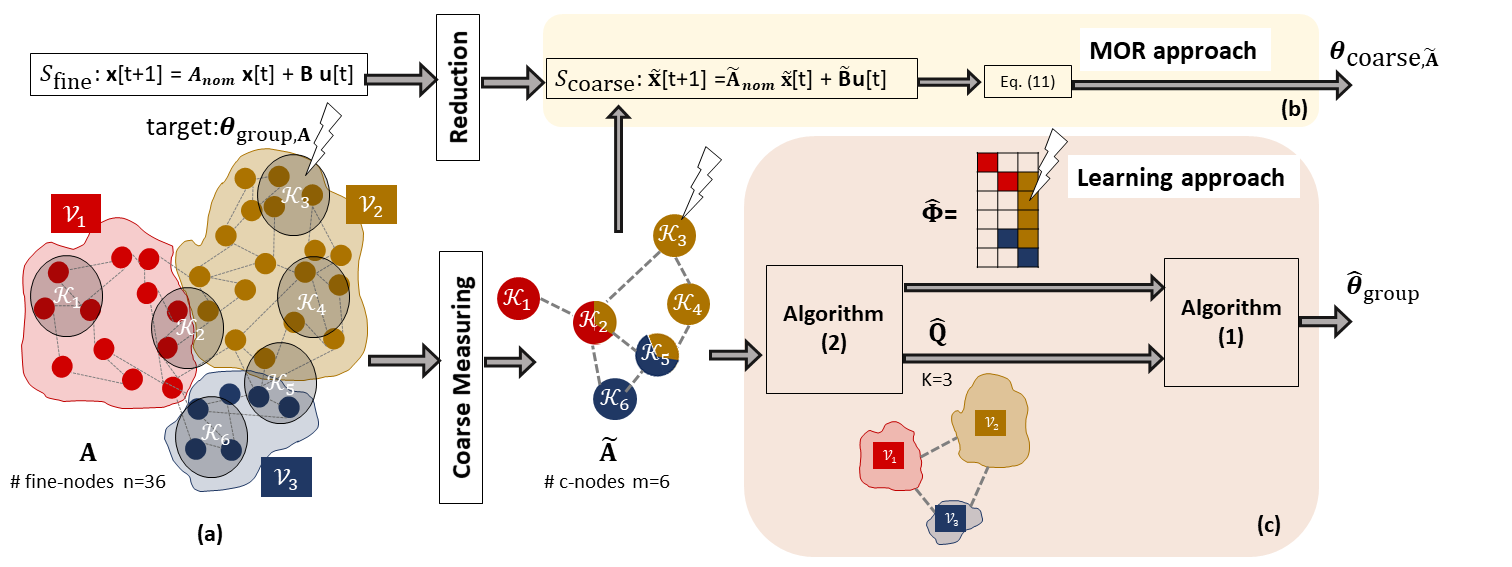}
		\caption{\small{Schematic of PROM- and learning-based approaches for estimating $\bs{\theta}_{\mathrm{group},\matF}$. (a) For $\graphF$ consisting of $\dimMatF=36$ nodes, we have $\dimMatCom=3$ communities ($\mathcal{V}_1$, $\mathcal{V}_2$, $\mathcal{V}_3$) with $\dimMatC=6$ coarse (c)-nodes ($\mathcal{K}_1$,\ldots,$\mathcal{K}_6$), each having coverage size $\protoCovSize=3$. The control node set is $\mathcal{K}_3$. The lightening bolts show an example of where external input can be exerted. {(b) In the PROM-based approach, we infer $\bs{\theta}_{\mathrm{group},\matF}$ via the reduced-order dynamical system $\mathcal{S}_\mathrm{coarse}$. (c) The learning-based approach capitalizes on the mixed-membership (MM) Algorithm \ref{alg:MMcommunityMao} to estimate $\bs{\theta}_{\mathrm{group},\matF}$ directly from $\matC$. This  avoids the need to consider $\mathcal{S}_\mathrm{coarse}$.}}}
		\label{fig:systemModel}	
	\end{figure*} 
	
\subsection{The stochastic block model}\label{subsec:SBMmodel}
	The stochastic block model (SBMs) is a  probabilistic model that produces random graphs with communities.
	We use this model, as hinted in the introduction, to reveal the role of community structure in the quality of controllability estimation from coarse measured graphs. 
	Formally, let $\graphG\triangleq(\mathcal{V},\mathcal{E})$ be the un-directed graph 
	(i.e., the fine graph) 
	with $\dimMatF$ nodes and random edge weights generated according to the general $\mathrm{GSBM}(\dimMatF,\matCom, \comSizes)$: 
	
	
\begin{definition} (\textbf{General SBM}) \label{def:generalSBM}
Given a natural number $\dimMatCom$, vector $\comSizes\in [0,1]^{\dimMatCom}$, and a matrix $\matCom \in [0,1]^{\dimMatCom\times \dimMatCom}$, the $\mathrm{GSBM}(\dimMatF,\matCom, \comSizes)$ is a distribution on a graph $\graphG$ of size $\dimMatF$
such that 1) the $\graphG$'s node set, denoted by $\mathcal{V}$ where $\lvert \mathcal{V} \rvert = \dimMatF$, is
	 partitioned into $\dimMatCom$ disjoint communities $\mathcal{V}=\displaystyle\cup_{k=1}^{\dimMatCom}\mathcal{V}_k$. The communities have relative sizes $\comSizes=[\comSizes_1,\ldots$  $\comSizes_\dimMatCom]$ where $\lvert \mathcal{V}_k\rvert = \dimMatF \comSizes_k$; and 2) 
	two nodes $u\in \mathcal{V}_k$ and $v\in \mathcal{V}_{k'}$ are joined by an edge with weight $\matF_{uv}\in\{0,1\}$ drawn with probability $\matCom_{kk'}$ independently from other edges, for all $k,k' \in [K]$. \qed 
\end{definition}

	The $\mathrm{GSBM}(\dimMatF,\matCom, \comSizes)$ generates a random un-directed graph with $\dimMatCom$ communities with non-identical intra- and cross-community connection probabilities specified by $\matCom\in [0,1]^{\dimMatCom\times\dimMatCom}$. An operational interpretation  of the General SBM is provided by the \emph{community membership matrix} $\comFine \in\mathbb{R}^{\dimMatCom\times \dimMatF}$: 
	\begin{align}\label{def_scaled_comMemberMat}
		\comFine_{kv} = \left\{\begin{array}{lr}
			1 & \mathrm{if }  v\in\mathcal{V}_k \\
			0 & \mathrm{otherwise}.  
		\end{array}\right.
	\end{align}	
For any $ k,k'\in[\dimMatCom]$, it follows that 
	\begin{align}\label{eq: SBM}
		\matF_{uv}\sim
		\mathrm{Bernoulli}(\matCom_{k,k'}), \quad \mathrm{if} \quad  \comFine_{ku}=\comFine_{k'v}=1. 
	\end{align} 
It is immediate that $\comFine{\comFine}^\transpose=\mathrm{Diag}\left(\vert \mathcal{V}_1\vert ,\ldots,\vert \mathcal{V}_\dimMatCom \vert \right)$, where $\lvert\mathcal{V}_k\rvert$ is the size of the $\dimMatCom$-th community. Let $\diagComSize\triangleq\frac{1}{\dimMatF}\comFine\comFine^\transpose=\mathrm{Diag}\left(\comSizes\right)$ be the diagonal matrix of \textit{relative} community sizes.

To facilitate analysis and to harness the power of SBM generated networks, we work with high-dimensional networks. We make some assumptions on the scaling of the parameters of the underlying SBM as $\dimMatF\to \infty$.
\begin{assump}\label{assump: graph sparsity}{(\textbf{SBM scaling} \cite{abbe2017community})} 
	There exists a sequence $\scalingSBM \in (0,1)$ and a constant $K \times K$ matrix $\matCom_{\circ}\geq0$ such that 
	\begin{align}\label{eq: scalingSBM}
		\matCom = \scalingSBM \matCom_{\circ}, 
	\end{align} 
	where $0\leq [\matCom_{\circ}]_{kk'}\leq 1$, and $\scalingSBM\sqrt{\dimMatF}\to\infty$ as $\dimMatF\rightarrow \infty$.
\end{assump}
The parameter $\scalingSBM$ in \eqref{eq: scalingSBM} allows us to regulate the density (or sparsity) of connections in and across communities. For several real-world networks, $\scalingSBM$ typically decreases with $\dimMatF$ \cite{abbe2017community}. 

\begin{remark}
The relative community sizes of the GSBM in Def. \ref{def:generalSBM} are set to $\comSizes$. Another common model is to choose community memberships in an i.i.d. manner from some distribution $\mbf{p}$. In this model, relative community sizes become random with  $\mathbb{E}[\comSizes]=\mbf{p}$. The analysis of this paper will hold for such a model except that, when evoking the expectation of the fine matrix $\matF$, an outer expectation should be added to account for the randomness in community sizes. 
\end{remark} 

\vspace{-3.0mm}
\subsection{Fine Graph and Coarse Measurements}\label{subsec: coarsening}
In this part, we describe the modeling choice and notations for the coarse measuring process.
We assume that the fine graph 
$\graphF$ with adjacency matrix $\matF$ is drawn from the General SBM distribution in Definition \ref{def:generalSBM}. Moreover, \eqref{generalLTI} is the \textit{fine-scale} system dynamics, and we refer to it as $\mathcal{S}_\mathrm{fine}$, i.e.
\begin{align}\label{fine_LTI}
\mathcal{S}_\mathrm{fine}:\,\,&\left.
\stateVecF[t+1]  = \matF_\mathrm{nom} \stateVecF[t] + {\mathbf{B}} \controlVec[t]\right., \quad \quad \forall\,\,\, t\in \mathbb{N}. 
\end{align}
We define the coarse-scale summary of $\matF$ as
\begin{align}\label{linear_coarsening_model}
	\matC \triangleq \matLinCoarse {\matF} \matLinCoarse^\transpose \quad \in \mathbb{R}^{m \times m},
\end{align}
where $\matLinCoarse\in\mathbb{R}^{\dimMatC\times\dimMatF}$ is the \emph{coarse-measurement matrix}  whose rows indicate which parts of the graph are measured. In this paper we assume (a) The $i$-th row of the coarse-measurement matrix $\matLinCoarse$ has $\protoCovSize_i$ equal non-zero terms which sum to $1$. These non-zero indices are the fine nodes that are measured as a group; and (b) $\matLinCoarse\matLinCoarse^\transpose=\mathrm{Diag}([1/\protoCovSize_1,\cdots,1/\protoCovSize_\dimMatC])$, meaning the sets of fine nodes corresponding to each coarse measurement are disjoint. 
The coarsened matrix $\matC$ is the starting point of our analysis. 
	
	The matrix $\matC$ can be interpreted as the adjacency matrix of an un-directed \textit{coarse} graph $\mathcal{G}_\mathrm{coarse}$ with $\dimMatC$ \textit{fictitious} nodes---referred to as the (coarse or) \textit{c-nodes}. This interpretation is helpful as we discuss the network dynamics of $\matC$ in Section \ref{sec: group-vs-coarse}. The regime where the \textit{total coverage size} $\sum_{i\in[\dimMatC]}\protoCovSize_i\ll \dimMatF$ indicates that c-nodes cover the fine graph only sparsely. In other words, there may exist (several) fine nodes that do not contribute to $\matC$ (see Fig.~\ref{fig:systemModel}). Similar coarsening models have been used in the literature (see \cite{ghoroghchian2021graph,dasarathy2015sketching}).

\subsection{Problem Statement}\label{subsec:probstatement}
	Our objective is to quantify the average controllability of $\graphF$ using only $\mathcal{G}_\mathrm{coarse}$ (see \eqref{linear_coarsening_model}), and the knowledge of $\graphF$'s community structure in Definition~\ref{def:generalSBM}. Furthermore, 
	we do not have access to the way the coarse graph is
    acquired at the time of decision making; that is,  the coarse-measurement matrix $\matLinCoarse$ is unknown, meaning we do not know which group of fine nodes are mapped to each c-node. 
	

Let $\matF_\mathrm{nom}=\matF/{(\constFine +\specRadius(\matF))}$ in \eqref{fine_LTI} be such that 
$\matF$ is drawn from General SBM given by \eqref{eq: SBM}. The input matrix $\mathbf{B}_{\mathcal{K}}\triangleq \mathbf{B}$ is diagonal by definition.
Let
$\mathcal{K}_i=\mathrm{supp}(\vecLinCoarse_i)$, for $i \in [\dimMatC]$, be the set (hereafter, group) of nodes coarsened by the $i$-th row $\vecLinCoarse_i$ of $\matLinCoarse$ in \eqref{linear_coarsening_model}.  
As a result, the control node groups $\{\mathcal{K}_1,\ldots,\mathcal{K}_m\}$ match the measured groups.

	Recall that $\vecLinCoarse_i$ is the coarse measuring vector mapping to the $i$-th c-node. For the $\mathcal{K}_i$-th group with $\mathbf{B}_{\mathcal{K}_i}=\mathrm{Diag}(\vecLinCoarse_i^\transpose)$, $i \in [\dimMatC]$, consider the average controllabity metric $\bs{\theta}^{(i)}_{\mathrm{group},\matF}\triangleq\trace[\Cgram(\matF_\mathrm{nom}, \mathbf{B}_{\mathcal{K}_i})]$. Define the $m$-dimensional group average controllability vector as 
	\begin{align}\label{eq: group avg-controllability}
		\bs{\theta}_{\mathrm{group},\matF} &\triangleq\left[\bs{\theta}^{(1)}_{\mathrm{group},\matF},\ldots,\bs{\theta}^{(m)}_{\mathrm{group},\matF}\right]^\transpose.
	\end{align}
	Note that $\bs{\theta}^{(i)}_{\mathrm{group},\matF}\geq 0$, due to the positive semi-definiteness of the Gramian $\Cgram(\cdot)$. This vector summarizes the average controllability metrics 
	for all control nodes sets $\{\mathcal{K}_1,\ldots,\mathcal{K}_m\}$. Thus, $\bs{\theta}_{\mathrm{group},\matF}$ helps us deduce several quantitative properties that are key to being able to control \eqref{fine_LTI} via $\mathbf{B}_{\mathcal{K}_i}$. For example, $\bs{\theta}_{\mathrm{group}}$ tells us which $\mathcal{K}_is$ can drive the network to the desired target state with least control effort. 

	Using only the knowledge of $\matC$ in \eqref{linear_coarsening_model}, our goal is to 
	estimate the vector $\bs{\theta}_{\mathrm{group},\matF}$ in \eqref{eq: group avg-controllability}. We will develop two approaches. The first, \textit{Passively Reduced-Order Model (PROM)}, is based on the traditional \textit{model order reduction} (MOR) approach. Here we rely upon a passively reduced-order auxiliary system  $\mathcal{S}_\mathrm{coarse}$ in \eqref{eq: coarse system}, which is governed by $\matC$, to infer $\bs{\theta}_{\mathrm{group},\matF}$. Our second approach is learning-based. We directly estimate $\bs{\theta}_{\mathrm{group},\matF}$ using a clustering-based mixed-membership community-learning algorithm (see Section \ref{sec:est_solution}). In the latter case, we bypass the intermediate step of dealing with a reduced-order system. Fig.~\ref{fig:systemModel} illustrates the proposed approaches. We use the $\ell_1$-norm difference between the true- and the estimated controllability vectors to compare the performance of these competing methods (see \eqref{eq: error metric MOR}). Our main results in Sections \ref{sec: group-vs-coarse} and \ref{sec:est_solution} are probabilistic because $\mathbf{A}$ is a random matrix.

	We introduce the expected quantities $\bar{\matF}\triangleq \mathbb{E}[\matF]$ and $\bar{\matC}\triangleq\mathbb{E}[\matC]$. We later use these to place some mild regulations on $\mathcal{G}_\mathrm{coarse}$. From \eqref{eq: SBM} and \eqref{linear_coarsening_model}, the expected quantities evaluate to 
	\begin{align}\label{def_coarseMembershipMat}
		\bar{\matF}&= \comFine^\transpose \matCom \comFine \quad \mathrm{and} \quad
		\bar{\matC}= \syncComAs \matCom \syncComAs^\transpose,
	\end{align}
	where $\syncComAs \triangleq\matLinCoarse \comFine^\transpose \in \mathbb{R}^{\dimMatC\times \dimMatCom}$ is the \textit{coarse community membership matrix}, and $\syncComAs_{ik}$ captures how much the $i$-th c-node overlaps with the $\dimMatCom$-th community. 
	To make the analysis less cumbersome, in what follows,  we let the coarsening matrix has fixed row support ($r$).
	The example below illustrates the definition of the matrices $\comFine$, $\diagComSize$, $\syncComAs$, and $\matLinCoarse$. 
	\begin{figure}[htp]
		\centering
		\includegraphics[width=0.25\textwidth]{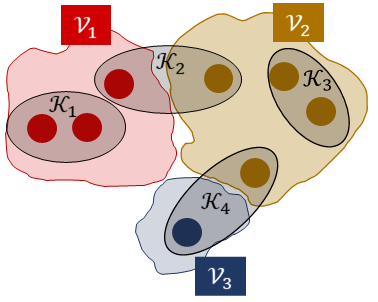}
		\caption{\small{{A simple case of a community-structured network $\graphF$ and its coarse measurements for Example~\ref{example:graph_coarse}. This graph contains $\dimMatF=8$ nodes, $\dimMatCom=3$ communities ($\mathcal{V}_1$, $\mathcal{V}_2$, $\mathcal{V}_3$) and $\dimMatC=4$ coarse (c)-nodes ($\mathcal{K}_1$,\ldots,$\mathcal{K}_4$), each having coverage size $\protoCovSize=2$. Note that c-nodes $1,3$ are perfectly synchronized with one community, while c-nodes $2,4$  overlap with two communities.}}}
		\label{fig.exampleSys}	
	\end{figure} 
	\begin{exmp}\label{example:graph_coarse}
		For the network shown in Fig.~\ref{fig.exampleSys}-(a), the fine community membership matrix is
		\begin{align*}
		    \comFine=\left[\begin{array}{cccccccc}
			    1 & 1 & 1 & 0 & 0 & 0 & 0 & 0\\
			    0 & 0 & 0 & 1 & 1 & 1 & 1 & 0\\
			    0 & 0 & 0 & 0 & 0 & 0 & 0 & 1
			\end{array}\right] \in \lbrace 0,1\rbrace^{3\times 8}.
		\end{align*}
		The relative fine-scale community sizes are
		\begin{align*}
		    \comSizes = [\frac{3}{8}, \frac{1}{2}, \frac{1}{8}], \quad \diagComSize = \mathrm{Diag}(\comSizes)=\left[\begin{array}{ccc}
			    \frac{3}{8} & 0 & 0\\
			    0 & \frac{1}{2} & 0\\
			    0 & 0 & \frac{1}{8}
			\end{array}\right] \in \mathbb{R}^{3\times 3}.
		\end{align*}
		The coarse-measurement matrix where $\protoCovSize=2$ is
		\begin{align*}
		    \matLinCoarse = \left[\begin{array}{cccccccc}
			    \frac{1}{2} & \frac{1}{2} & 0 & 0 & 0 & 0 & 0 & 0\\
			    0 & 0 & \frac{1}{2} & \frac{1}{2} & 0 & 0 & 0 & 0\\
			    0 & 0 & 0 & 0 & \frac{1}{2} & \frac{1}{2} & 0 & 0 \\
			    0 & 0 & 0 & 0 & 0 & 0 & \frac{1}{2} & \frac{1}{2} 
			\end{array}\right] \in \mathbb{R}^{4\times 8}.
		\end{align*}
		The coarse community membership matrix tells us what fractions of each coarse node's group of fine nodes belong to each underlying community,
		\begin{align*}
		    \syncComAs =\matLinCoarse\comFine^\transpose= \left[\begin{array}{ccc}
			    1 & 0 & 0\\
			    \frac{1}{2} & \frac{1}{2} & 0\\
			    0 & 1 & 0 \\
			    0 & \frac{1}{2} & \frac{1}{2}
			\end{array}\right] \in \mathbb{R}^{4\times 3}.
		\end{align*}
		Finally, we note that the rows of $\syncComAs^\transpose$ determine that each community contains what portion of the coarse nodes' groups of fine nodes. Hence,
		\begin{align*}
		    \frac{\syncComAs^\transpose\syncComAs}{\dimMatC}= \left[\begin{array}{ccc}
			    \frac{5}{16} & \frac{1}{16} & 0\\
			    \frac{1}{16} & \frac{6}{16} & \frac{1}{16}\\
			    0 & \frac{1}{16} & \frac{1}{16} 
			\end{array}\right] \in \mathbb{R}^{3\times 3}
		\end{align*} 
		is a measure of the correlation of (and cross-correlation between) the (coarse membership) representations of communities (i.e., the rows of $\syncComAs^\transpose$). Matrix  $\frac{\syncComAs^\transpose\syncComAs}{\dimMatC}$ plays a role in the error bound of Theorem~\ref{thm: group-coarse-error2} and is discussed in Section \ref{sec: group-vs-coarse}. 
Each c-node in $\{1,3\}$ covers one community; each c-nodes in $\{2,4\}$ overlap with two. 
\qed 
\end{exmp}

	\begin{remark}\label{rmk: sync}
		\textbf{(Synchronization of community and coarsening)}: The coarsening operation in Section \ref{subsec: coarsening} is oblivious to the community structure of $\graphF$. Thus, the $i$-th c-node may contain information about multiple communities if $\lvert\mathrm{supp}(\syncComAs_{i}) \rvert>1$ (the subscript denotes the $i$-th row). \textbf{Perfect synchronization} is the situation in which $\lvert\mathrm{supp}(\syncComAs_{i}) \rvert=1$. This occurs when each c-node covers nodes from only one community.
	\end{remark}

	\section{Approximating Group Average controllability: PROM-based Approach}\label{sec: group-vs-coarse} 
	In this section, we define PROM (Passively Reduced-Order Model) as an auxiliary or fictitious network linear dynamics associated with $\graphC$. We then analyze how well the average controllability vector of this reduced system approximates $\bs{\theta}_{\mathrm{group},{\matF}}$ in \eqref{eq: group avg-controllability}. Formally, consider the coarse system $\mathcal{S}_\mathrm{coarse}$ defined by the dynamics
	\begin{align}\label{eq: coarse system}
		\mathcal{S}_\mathrm{coarse}:\,\,&
		\left.
		\wt{\stateVecF}[t+1]  = \wt{\matF}_\mathrm{nom} \wt{\stateVecF}[t] + \wt{\mathbf{B}}\wt{\controlVec}[t]\right. , 
	\end{align} 
	where $\wt{\matF}_\mathrm{nom}=\wt{\matF}/{(\constCoarse+\specRadius(\matC))}$, $\matC$ is given by \eqref{linear_coarsening_model}.
	Note that $\stateVecC[t]\in \mathbb{R}^{m}$ is not a compressed version of the state $\mbf{x}[t]$ in $\mathcal{S}_\mathrm{fine}$. Rather, $\stateVecC[t]$ is a fictitious state controlled by the dynamics of the (scaled) matrix $\wt{\matF}$. This fictitious state is controlled by the input $\wt{\mathbf{B}}\wt{\controlVec}[t] \in \mathbb{R}^{\dimMatC}$. We refer to this  coarse-scale system modeling from the fine-scale system as PROM because the dimension of the coarse system can be much less than that of the fine system, i.e. $ \dimMatC \ll \dimMatF$.
	
	Since each group of nodes that receive control input maps to one c-node,
	$\wt{\mbf{B}}_i=\mbf{e}_i$
	is the coarsened input matrix.
	Also let $\bs{\theta}^{(i)}_{\mathrm{coarse},\matC}\!\triangleq\!\trace[\Cgram(\matC_\mathrm{nom}, \wt{\mbf{B}}_i)]\!\geq\!1$ be the average controllability for the $i$-th node in $\graphC$, which is, by definition, lower-bound by $1$. Define the coarse average controllability vector:
	\begin{align}\label{eq: coarse avg-controllability}
		\bs{\theta}_{\mathrm{coarse},\matC} &\triangleq\left[\bs{\theta}^{(1)}_{\mathrm{coarse},\matC},\ldots,\bs{\theta}^{(m)}_{\mathrm{coarse},\matC}\right]^\transpose\in \mathbb{R}^{m}. 
	\end{align}

	We show in Theorem \ref{thm: group-coarse-error2} that the $\bs{\theta}_{\mathrm{coarse},\matC}$ well approximates $\bs{\theta}_{\mathrm{group},{\matF}}$ under some conditions. To quantify the approximation, we define the $\ell_1$-error metric: 
	\begin{align}\label{eq: error metric MOR}
	\hspace{-2.0mm}	{\Delta(\matF,\matC)}\!\triangleq\! 
	\left\|\frac{\protoCovSize\bs{\theta}_{\mathrm{group},{\matF}}-\mbf{1}_\dimMatC}{\|\protoCovSize\bs{\theta}_{\mathrm{group},{\matF}}-\mbf{1}_\dimMatC\|_1}\!-\!\frac{\bs{\theta}_{\mathrm{coarse},\matC}-\mbf{1}_\dimMatC}{\|\bs{\theta}_{\mathrm{coarse},\matC}-\mbf{1}_\dimMatC\|_1}\right\|_1.
	\end{align}
	We prefer ${\Delta(\matF,\matC)}$ to the usual metric $\|\bs{\theta}_{\mathrm{group},\matF}-\bs{\theta}_{\mathrm{coarse},\matC}\|_1$ for the following reasons. First, the factor of $\protoCovSize$ in the first term of ${\Delta(\matF,\matC)}$ in \eqref{eq: error metric MOR} accounts for the additional scaling by $1/\protoCovSize$ in the fine system input matrix $\mathbf{B}_{\mathcal{K}_i}=\mathrm{Diag}(\vecLinCoarse_i^\transpose)$ (see Lemma~\ref{lemma:groupTheta} in the appendix). Second, the shift by $-1$ discounts the inherent $+1$ shift in the average controllability definition (see \eqref{gramian_to_diag} in the appendix). Third, the denominators in \eqref{eq: error metric MOR} ensure that the error is {scale-invariant}, i.e. the error is unaffected if all elements of either of the controllability vectors are multiplied by a common scalar factor. Moreover we chose $\|.\|_1$ over other vector norms (like $\|.\|_\mathrm{max}$) as it aggregates the errors when choosing a singular or any combinations of c-nodes to stimulate sequentially.

The following matrix will appear in Theorems~\ref{thm: group-coarse-error2} and contributes to upper bounding the controllability estimation errors
\begin{align}\label{upsilon_def}
	\inCtrlMat\triangleq 
	  \diagComSize^{-\frac{1}{2}} \left(\left[\mbf{I}-(\overalConstTemp\diagComSize^{\frac{1}{2}} \matCom_{\circ}\diagComSize^{\frac{1}{2}})^2\right]^{-1} -\mbf{I}\right) \diagComSize^{-\frac{1}{2}} \in \mathbb{R}^{\dimMatCom\times\dimMatCom},
\end{align}
where 
$\overalConstTemp=1/(\frac{\constFine}{\scalingSBM\dimMatF} +\specRadius(\matCom_{\circ} \diagComSize))$ is the normalization factor $1/{(\constFine +\specRadius(\bar{\matF}))}$ in \eqref{fine_LTI} after substituting $\bar{\matF}$ with $\matF$, with $\specRadius(\cdot)$ is the spectral radius, $\scalingSBM$ and $\matCom_{\circ}$ are in Assumption \ref{assump: graph sparsity}.
{
We set the normalized matrix $\genMat=\overalConstTemp\diagComSize^{\frac{1}{2}} \matCom_{\circ}\diagComSize^{\frac{1}{2}}$. Hence, the term $\left[\mbf{I}-(\overalConstTemp\diagComSize^{\frac{1}{2}} \matCom_{\circ}\diagComSize^{\frac{1}{2}})^2\right]^{-1}$ in \eqref{upsilon_def} becomes the resolvent of $\genMat$, i.e. $(z\mbf{I}-\genMat^2)^{-1}$, evaluated at $z=1$. Also, note that matrix $\diagComSize^{\frac{1}{2}} \matCom_{\circ}\diagComSize^{\frac{1}{2}}$ is the intra- and cross-community probability matrix that is re-weighted according to the relative community sizes. Recall that the resolvent captures the direct and indirect effects (via the sum of the powers) of $\diagComSize^{\frac{1}{2}} \matCom_{\circ}\diagComSize^{\frac{1}{2}}$. 
As a result, matrix $\inCtrlMat$ in \eqref{upsilon_def} encapsulates the size-weighted community-connection probability matrix's direct and indirect impacts. For more information on the significance and the properties of resolvents see \cite{dunford1988linear}.}
Later, in Lemma~\ref{lemma:approximate_gAC}, we see how $[\inCtrlMat]_{kk}$ moves in tandem (i.e. increase and decrease together) with both $\bs{\theta}^{(i)}_{\mathrm{group},\bar{\matF}}$ and $\bs{\theta}^{(i)}_{\mathrm{coarse},\bar{\matC}}$ for $i$-th c-nodes that are perfectly synchronized with community $\dimMatCom$.
The aformentioned definitions coupled with Lemma~\ref{lemma:approximate_gAC} in Section~\ref{sec:est_solution}, contribute to the proof of the next Theorem.

	\begin{thm}\label{thm: group-coarse-error2}(\textbf{$\ell_1$ bound on} $\bs{\theta}_{\mathrm{group},\matF}-\bs{\theta}_{\mathrm{coarse},\matC}$): 
		Let ${\Delta(\matF,\matC)}$ be defined as in \eqref{eq: error metric MOR} for $\matF\in \lbrace0,1\rbrace^{\dimMatF\times\dimMatF}$ and $\matC\in \mathbb{R}^{\dimMatC\times\dimMatC}$, and $0<\hoeffProbConst<1$ be a constant. Then, under Assumption~\ref{assump: graph sparsity}
		$\exists \dimMatF_0,\dimMatC_0\in \mathbb{N}$ such that for $\dimMatF>\dimMatF_0,\dimMatC>\dimMatC_0$
		\begin{align*}
		         {\Delta\!(\matF,\!\matC\!)}
		         \!=\! \Oscale \! \! \left( \!\! \frac{\hoeffCoeffConst\!+\!\tilde{\hoeffCoeffConst}}{\scalingSBM}\!\!  +\! \sqrt{\!\frac{\dimMatF}{\dimMatC\protoCovSize}\!}\sqrt{\!\frac{\log(\!\frac{1}{\hoeffProbConst}\!)}{2}} 
		         \!+\! \frac{\|\!\errorMat_{\syncComAs\!,\!\inCtrlMat}\!\|_{\!1}}{\dimMatC} \!+\!\left\|\!\diagComSize\!\shortminus\!\frac{\syncComAs^{\!\transpose}\! \syncComAs}{\dimMatC}\!\right\|_{\!\mathrm{max}\!}\!\!\right)
		\end{align*}
		with probability at least $1-3\hoeffProbConst$. Here 
		$
		    \hoeffCoeffConst = \frac{\sqrt{{\log(\frac{\dimMatF^2}{\hoeffProbConst})}/{(2\dimMatF)}}}{\left\|\matCom \diagComSize \matCom\right\|_{\mathrm{min}}}$; $\tilde{\hoeffCoeffConst} = \frac{\sqrt{{\log(\frac{\dimMatC^2}{\hoeffProbConst})}/{(2\dimMatC)}}}{ \|\matCom\frac{\syncComAs^{\!\transpose}\! \syncComAs}{\dimMatC} \matCom\|_{\mathrm{min}}}$;  $\errorMat_{\syncComAs,\inCtrlMat}=\syncComAs {\mathrm{diag}(\inCtrlMat)}-\mathrm{diag}(\syncComAs\inCtrlMat\syncComAs^\transpose)$; $\diagComSize=(1/\dimMatF)\comFine\comFine^\transpose$; 
		$\inCtrlMat$ is defined in \eqref{upsilon_def}; $\syncComAs$ is defined in \eqref{def_coarseMembershipMat} and $\matCom$ in Def.~\ref{def:generalSBM}. \qed
	\end{thm}
	Theorem \ref{thm: group-coarse-error2} 
	indicates that $\Delta(\matF,\matC)$ approaches zero as: 1) the total coverage proportion, i.e. ${\dimMatC\protoCovSize}/{\dimMatF}$, increases; 2) $\hoeffCoeffConst$ and $\tilde{\hoeffCoeffConst}$,
    that guarantee the bound on $\Delta(\matF,\matC)$ exists with large probability, 
	become negligible; 3)  
	$\frac{1}{\dimMatC}\|\syncComAs {\mathrm{diag}(\inCtrlMat)}-\mathrm{diag}(\syncComAs\inCtrlMat\syncComAs^\transpose)\|_1$ and $\|\diagComSize-\frac{\syncComAs^\transpose \syncComAs}{\dimMatC}\|_{\mathrm{max}}$ decrease. The latter two terms have to do with \textit{synchronization (or overlap)}, i.e. the extent to which coarse measurements are synchronized (or overlapped) with communities (see Remark \ref{rmk: sync}). The first term $\frac{1}{\dimMatC}\|\syncComAs {\mathrm{diag}(\inCtrlMat)}-\mathrm{diag}(\syncComAs\inCtrlMat\syncComAs^\transpose)\|_1$ is equal to zero if perfect synchronization holds.
	The relationship between the second term, $\|\diagComSize-\frac{\syncComAs^\transpose \syncComAs}{\dimMatC}\|_{\mathrm{max}}$, and community overlap is less straightforward. To understand this, note that matrix $\frac{\syncComAs^\transpose \syncComAs}{\dimMatC}$ (see Example~\ref{example:graph_coarse}) quantifies a measure of correlation between (coarse membership) representations of communities. This matrix reveals two natural tradeoffs in the coarsening process. The first is \textit{overlap}: the off-diagonal elements, $[\frac{\syncComAs^\transpose \syncComAs}{\dimMatC}]_{kk'}= \sum_{i\in[\dimMatC]} \syncComAs_{ik}\syncComAs_{ik'} /\dimMatC$ for $k\neq k'$, show the extent of synchronization between coarse measurements and communities (c.f., Remark~\ref{rmk: sync} and Example~\ref{example:graph_coarse}). The second is \textit{balancedness}: the diagonal elements, $[\frac{\syncComAs^\transpose \syncComAs}{\dimMatC}]_{kk} = \sum_{i\in[\dimMatC]} \syncComAs_{ik}^2 /\dimMatC$ for all $k\in[\dimMatCom]$, reflect the relative community sizes \textit{after} coarse measuring. Separately considering diagonal and off-diagonal elements, helps demystify the role of the error term $\|\diagComSize-\frac{\syncComAs^\transpose \syncComAs}{\dimMatC}\|_{\mathrm{max}}$ in Theorem~\ref{thm: group-coarse-error2}:
	\begin{align}\label{DMinusPhiPhiTterms}
        \left\lvert [{\diagComSize} - \frac{\syncComAs^\transpose \syncComAs}{\dimMatC}]_{k,k'} \right\rvert  = \left \lbrace\begin{array}{ll}
             \sum_{i\in[\dimMatC]} \syncComAs_{ik}\syncComAs_{ik'} /\dimMatC & \mathrm{if}\quad k\neq k' \\
             \lvert \comSizes_k-\sum_{i\in[\dimMatC]} \syncComAs_{ik}^2 /\dimMatC\rvert & \mathrm{if}\quad  k=k'.
            \end{array}\right.
    \end{align}
	From \eqref{DMinusPhiPhiTterms} we infer that $\|\diagComSize-\frac{\syncComAs^\transpose \syncComAs}{\dimMatC}\|_{\mathrm{max}}$ approaches zero when c-nodes overlap with fewer communities, \textit{and} c-nodes are more evenly spread over communities. The latter means that larger communities get more coarse measurements, and vice versa. 
	
	Overall, we conclude that the PROM-based estimate $\bs{\theta}_{\mathrm{coarse},\matC}$ will well approximate $\bs{\theta}_{\mathrm{group},\matF}$ when the graph density $\scalingSBM$ and number of fine- and c-nodes are sufficiently large, when coarse measurements are sufficiently synchronized with the communities, and when coarse measurements are distributed across communities in a balanced way.
	In many settings such conditions will be (roughly) satisfied in practice, in particular when communities and coarse measurements are spatially (geographically) localized \cite{sporns2016modular,galhotra2018geometric}. This is because nearby nodes tend to be connected and form communities (e.g., social networks of friends and neighbours, regions with distinct functional and structural features in the brain). Moreover, coarse-measuring is likely local (e.g. parcellation of brain imaging data, electrophysiological recording of the brain, aggregation of geographically adjacent nodes). Also, coarse measurements tend to be uniformly distributed across a network and, hence, larger communities are typically measured more frequently.

	\section{Approximating Group Average Controllability: Learning-based  Approach}\label{sec:est_solution}
	We now present our learning-based approach to estimate $\bs{\theta}_{\mathrm{group},\matF}$.
	Unlike the PROM-based approach that relies on $\mathcal{S}_\mathrm{coarse}$, we directly estimate elements in $\bs{\theta}_{\mathrm{group},\matF}$ using a mixed-membership (MM) community-learning algorithm \cite{mao2020estimating,rossetti2019cdlib,aicher2015learning}. 
	Specifically, we work with the {mixed-membership} algorithm of \cite{mao2020estimating} which is not only numerically efficient but comes with strong theoretical guarantees. 

	As mentioned in Section~\ref{sec: group-vs-coarse}, the following lemma is instrumental in the proof of both of our primary results, Theorems \ref{thm: group-coarse-error2} and \ref{thm:full_est_cntrl_error}. This lemma also provides motivation for our candidate estimator in \eqref{eq: candidate estimator}. For these reasons, we delayed the presentation of Lemma~\ref{lemma:approximate_gAC}  until now. In essence, this lemma characterizes the group-average controllability vector associated with the expected matrix $\bar{\matF}$ in \eqref{def_coarseMembershipMat}. 

    \begin{lemma}\label{lemma:approximate_gAC} 
    \textbf{(Group Controllability of Expected Dynamics):}
    Define ${\bs{\theta}}_{\mathrm{group},\bar{\matF}}$ by replacing $\matF$ in \eqref{eq: group avg-controllability} with the expected matrix $\bar{\matF}$ given by \eqref{def_coarseMembershipMat}. Let $\mbf{\Phi}$ be as in \eqref{def_coarseMembershipMat}, $\comFine$ be in \eqref{def_scaled_comMemberMat}, and $\protoCovSize$ be the coverage size of $\matLinCoarse$ in \eqref{linear_coarsening_model}. Under Assumption \ref{assump: graph sparsity}, we have
		{
			\begin{align}\label{true_groupCntrl}
			    {\bs{\theta}}_{\mathrm{group},\bar{\matF}} = \frac{1}{\protoCovSize}\left(\mbf{1}_{\dimMatC} +  \frac{1}{\dimMatF}\syncComAs {\mathrm{diag}(\inCtrlMat)}\right), 
			\end{align}
			where $\dimMatF$ and $\dimMatC$ are the dimensions of $\matF$ and $\matC$; 
			moreover $\inCtrlMat$ is defined in \eqref{upsilon_def}.
		}\qed
	\end{lemma}
    Lemma \ref{lemma:approximate_gAC} gives us a formula to compute the group average controllability vector associated with the expected matrix $\bar{\matF}$.
    Recall that $\syncComAs_{i,k}=\lvert \lbrace v: v\in  \mathcal{V}_k\cap \mathrm{supp}(\vecLinCoarse_i) \rbrace\rvert /\lvert \mathrm{supp}(\vecLinCoarse_i)\rvert$ ($k \in [K]$, $i\in [\dimMatC]$) is the fraction of the $i$-th c-node's overlap with community $\dimMatCom$. From Lemma \ref{lemma:approximate_gAC}, it follows that $\bs{\theta}_{\mathrm{group},\bar{\matF}}^{(i)} \propto\sum_{k\in[\dimMatCom]} \syncComAs_{i,k}\inCtrlMat_{kk}$. In view of this observation, we infer that c-nodes that have the most overlap with communities of largest $\inCtrlMat_{kk}$ are the most controllable.
	Moreover, Theorem \ref{thm:est_cntrl_error} (see Appendix) shows that $\Delta(\matF, \bar{\matF})$ can be arbitrarily small with high probability. This fact coupled with Lemma~\ref{lemma:approximate_gAC}, leads us to propose our candidate estimator
	\begin{align}\label{eq: candidate estimator}
		\hat{\bs{\theta}}_{\mathrm{group}}\triangleq \mbf{1}_{\dimMatC} + \hat{\syncComAs} \mathrm{diag}(\hat{\inCtrlMat}), 
	\end{align}
	where 
	\begin{align}\label{hat_defs}
	       \nonumber  \hat{\inCtrlMat}&\triangleq
	    \hat{\diagComSize}^{-\frac{1}{2}} \left(\left[\mbf{I}-(\hat{\overalConstTemp}\hat{\diagComSize}^{\frac{1}{2}} \hat{\matCom}_{\circ}\hat{\diagComSize}^{\frac{1}{2}})^2\right]^{-1} -\mbf{I}\right) \hat{\diagComSize}^{-\frac{1}{2}},
	    \\
	    \hat{\overalConstTemp} & ={\constFine}/{\dimMatF} +\specRadius(\hat{\matCom} \hat{\diagComSize}).
	\end{align}
	Note that the candidate estimator in \eqref{eq: candidate estimator} is similar to the true group controllability vector in \eqref{true_groupCntrl}, but there are two important differences. First, the true terms $\syncComAs,\matCom,$ and $\diagComSize$ are replaced with  their estimated (hatted) counterparts. Second, the constant coefficients $1/\protoCovSize$ and $1/\dimMatF$ in \eqref{true_groupCntrl} are omitted since they affect neither the error metric (see \eqref{errorHat} wherein $1/\protoCovSize$ is compensated and $1/\dimMatF$ cancels out in both the numerator and denominator) nor the order of nodes when assessed based on their controllability value.
	The hatted quantities $\hat{\syncComAs}$ and $\hat{\matCom}$ are obtained from Algorithm \ref{alg:MMcommunityMao}, which takes as input $\matC$ and the number of communities $\dimMatCom$\footnote{Algorithm \ref{alg:MMcommunityMao}, a mixed-membership algorithm, is adapted from \cite{mao2020estimating}. It is a type of spectral clustering method that first performs eigen decomposition of $\matC$ to find the overlapping membership ($\bs{\Phi}_{ik}$) of the fine nodes. A pruning step is included (see steps 4 and 5 in Algorithm \ref{alg:MMcommunityMao}) that improves algorithm performance.}. We obtain $\hat{\diagComSize}$ from Algorithm \ref{alg:modular_coarse_ctrl}.  Importantly, $\hat{\bs{\theta}}_{\mathrm{group}}$ in \eqref{eq: candidate estimator} is obtained from coarsened matrix $\matC$ and not from the fine scale matrix $\matF$. 
	The following assumption is required by \cite{mao2020estimating} for the recovery of $\hat{\syncComAs}, \hat{\matCom}$ using their proposed Algorithm~\ref{alg:MMcommunityMao}.
	\begin{assump}\label{assump: pureNode}
    	\textbf{(Existence of a perfectly-synchronized c-node):}  For all $k\in[K]$, there exist at least one coarse node $i\in[\dimMatC]$ such that $ \syncComAs_{ik}=1$.
    \end{assump}
    This means that $\mathcal{G}_\mathrm{coarse}$ has at least one perfectly-synchronized {(or equivalently ``pure'')} c-node per community $k \in [K]$.
    { In real-world networks  communities and coarse measurements are often  spatially or geographically localized (see  Section~\ref{sec: group-vs-coarse} for a more detailed explanation). Hence, it is reasonable to expect that there is  at least one coarse measurement per community that does not overlap with multiple communities.}
	The next theorem bounds the error of our proposed learning-based estimation in \eqref{eq: candidate estimator} using the following metric:
	\begin{align}\label{errorHat}
	    \widehat{\Delta}(\matC) = 
	    \left\|\frac{\hat{\bs{\theta}}_{\mathrm{group}}-\mbf{1}_\dimMatC}{\|\hat{\bs{\theta}}_{\mathrm{group}}-\mbf{1}_{\dimMatC}\|_1}-\frac{\protoCovSize\bs{\theta}_{\mathrm{group},{\matF}}-\mbf{1}_\dimMatC}{\|\protoCovSize\bs{\theta}_{\mathrm{group},{\matF}}-\mbf{1}_{\dimMatC}\|_1} \right\|_1.
	\end{align}
	\begin{thm}\label{thm:full_est_cntrl_error}
		(\textbf{$\ell_1$ error bound on} $\hat{\bs{\theta}}_{\mathrm{group}}-\bs{\theta}_{\mathrm{group},\matF}$):  Let $\widehat{\Delta}(\matC)$ be defined as in \eqref{errorHat} for $\matF\in \lbrace0,1\rbrace^{\dimMatF\times\dimMatF}$ and $\matC\in \mathbb{R}^{\dimMatC\times\dimMatC}$, and $0<\hoeffProbConst<1$ be a constant. Then, under Assumptions~\ref{assump: graph sparsity},\ref{assump: pureNode},  $\exists \dimMatF_0\in \mathbb{N}$ such that for $\dimMatF>\dimMatF_0$
		{\small
		\begin{align*}
			& \widehat{\Delta}(\matC) 
			  \!=\! \Oscale \!  \left( \! \frac{\hoeffCoeffConst}{\scalingSBM}\!\! +\! \sqrt{\!\frac{\dimMatF}{\dimMatC\protoCovSize}\!}\sqrt{\!\frac{\log(\!\frac{1}{\hoeffProbConst}\!)}{2}\!} \! + \! \frac{\|\!\errorMat_{\syncComAs}\!\|_{1,1}}{\dimMatC}
			 \! + \!\| \!\errorMat_{\matCom}\!\|_{\mathrm{max}} \! +\! \|\! \errorMat_{\diagComSize}\!\|_{\mathrm{max}} \! \right) \! 
		\end{align*}
		}
		holds with probability at least $1-2\hoeffProbConst$, where $\errorMat_{\syncComAs} = \hat{\syncComAs}- \syncComAs $, $\errorMat_{\matCom}=\hat{\matCom}- \matCom$, and  $\errorMat_{\diagComSize}=\hat{\diagComSize}-\diagComSize$; $\hoeffCoeffConst = \frac{\sqrt{{\log(\frac{\dimMatF^2}{\hoeffProbConst})}/{(2\dimMatF)}}}{\left\|\matCom \diagComSize \matCom\right\|_{\mathrm{min}}}$; $\dimMatF$ and $\dimMatC$ are the dimensions of $\matF$ and $\matC$; $\syncComAs$ is defined in \eqref{def_coarseMembershipMat} and $\matCom$ in Def.~\ref{def:generalSBM}; and $\diagComSize=(1/\dimMatF)\comFine\comFine^\transpose$.	 \qed
	\end{thm}
	Theorem \ref{thm:full_est_cntrl_error} 
	suggests that for sufficiently large $\dimMatF$ (smaller $\hoeffCoeffConst$ which yields high probabilistic guarantees for the bound on $\widehat{\Delta}(\matC)$) and large total coverage proportion $\frac{\dimMatC\protoCovSize}{\dimMatF}$, the estimate $\hat{\bs{\theta}}_{\mathrm{group}}$ approximates $\bs{\theta}_{\mathrm{group},{\matF}}$ to arbitrary precision { under the following conditions. First, the graph is dense enough (sufficiently large $\scalingSBM$). 
	Second, the coarse community membership matrix is well estimated (smaller $\| \errorMat_{\syncComAs} \|_{1,1}/\dimMatC$). 
	The third condition needed to yield small $\widehat{\Delta}(\matC)$ is that cross-community probability estimation does not suffer from high error (smaller $\| \errorMat_{\matCom}\|_{\mathrm{max}}$), and the relative community sizes estimated from the coarse graph are close to the ones in the fine graph (smaller $\| \errorMat_{\diagComSize}\|_{\mathrm{max}} $)}.
	%
	%
	%
	%
	%
	{ The authors of \cite{mao2020estimating} show that $\| \errorMat_{\syncComAs} \|_{1,1}/\dimMatC$ and $\| \errorMat_{\matCom}\|_{\mathrm{max}}$ in Theorem \ref{thm:full_est_cntrl_error} approach zero under some conditions as $\dimMatC\rightarrow \infty$. While it is unclear whether these theoretical conditions hold in our setting (given our coarsening process), our simulations show that $\widehat{\Delta}(\matC)$ decreases with an increase in $m$. The study and characterization of the behavior of $\| \errorMat_{\syncComAs} \|_{1,1}/\dimMatC$ and $\| \errorMat_{\matCom}\|_{\mathrm{max}}$ with respect to graph scaling (i.e. $\scalingSBM$) is left for future work. }
	
	The result of Theorem \ref{thm:full_est_cntrl_error} is important because one can directly infer the most influential control node groups, $\mathcal{K}_i$ (one with high average controllability) via the most influential $i$-th c-nodes, and vice versa. 
	{  In contrast to the PROM-based error characterization in Theorem~\ref{thm: group-coarse-error2}, the learning-based method's error does not suffer explicitly from non-synchronization of coarse measurements and communities, although the error is implicitly influenced by the estimation quality of $\syncComAs,\matCom,\diagComSize$. Hence, learning-based estimation of group controllability is likely to be more robust than that of the PROM-based when community overlap is present. This hypothesis is validated via simulation in Section~\ref{sec:simulations}.    
	}

	\begin{algorithm}[htp]
		\centering
		\caption{Direct Inference of the Group Average Controllability}\label{alg:modular_coarse_ctrl}
		\begin{algorithmic}[1]
			\REQUIRE estimates $\hat{\syncComAs}$ and $\hat{\matCom}$ from Algorithm ~\ref{alg:MMcommunityMao}, and the number of communities $\dimMatCom$
			\STATE compute $\hat{\diagComSize}=\mathrm{Diag}(\frac{\hat{\syncComAs}^\transpose \mbf{1}_{\dimMatC}}{\mbf{1}_{\dimMatCom}\hat{\syncComAs}^\transpose \mbf{1}_{\dimMatC}})$ and $\hat{\overalConstTemp}=\frac{\constFine}{\dimMatF} +\specRadius(\hat{\matCom} \hat{\diagComSize})$
			\RETURN $\hat{\bs{\theta}}_{\mathrm{group}}\!=\! \mbf{1}_{\dimMatC} \!+\! \hat{\syncComAs} \mathrm{diag}\!\left(\hat{\diagComSize}^{-\frac{1}{2}} \left(\left[\mbf{I}\!-\!(\hat{\overalConstTemp}\hat{\diagComSize}^{\frac{1}{2}} \hat{\matCom}_{\circ}\hat{\diagComSize}^{\frac{1}{2}})^2\right]^{-1} \!\!-\!\mbf{I}\right) \hat{\diagComSize}^{-\frac{1}{2}}) \right) $
		\end{algorithmic}
	\end{algorithm}
	
	\begin{algorithm}[htp]
		\centering
		\caption{Mixed-Membership Community Estimation Algorithm \cite{mao2020estimating}}\label{alg:MMcommunityMao}
		\begin{algorithmic}[1]
			\REQUIRE Coarse adjacency matrix $\matC$, number of communities $\dimMatCom$
			\STATE compute the highest $\dimMatCom$ spectral-decomposition of $\matC$ as $\hat{V}\hat{\Lambda}\hat{V}^\transpose$ and set $\mathcal{S}_{\mathrm{pruned}} = \mathrm{Prune}(\hat{V})$
			\STATE set $X = \hat{V}([\dimMatC]\backslash \mathcal{S}_{\mathrm{pruned}}, :)$ and compute $\mathcal{S}_{\mathrm{pure}} =$ Successive Projection Algorithm$(X^\transpose)$
			\STATE  set $X_{\mathrm{pure}} = X(\mathcal{S}_{\mathrm{pure}}, :)$ and compute un-normalized $\hat{\syncComAs}^{\mathrm{un-nom}} = \hat{V} X_{\mathrm{pure}}^{-1}$
			\STATE $\hat{\syncComAs}^{\mathrm{un-nom}}_{ik}\leftarrow 0\quad \mathrm{if }\hat{\syncComAs}^{\mathrm{un-nom}}_{ik}<e^{-12}, {\small\forall i\in[\dimMatC], k\in [\dimMatCom]}$
			\RETURN $\hat{\syncComAs}=\textrm{Diag}^{-1}(\hat{\syncComAs}^{\mathrm{un-nom}}\mathbf{1}_\dimMatCom)\hat{\syncComAs}^{\mathrm{un-nom}}$ and $\hat{\matCom} = X_{\mathrm{pure}} \hat{\Lambda} X_{\mathrm{pure}}^\transpose$
		\end{algorithmic}
	\end{algorithm} 

\section{Simulations}\label{sec:simulations}

We validate our theoretical results by plotting
the errors $\Delta(\matF,\matC)$ and $\widehat{\Delta}(\matC)$ and show that these errors are comparable to the bounds we obtained in Theorems \ref{thm: group-coarse-error2} and \ref{thm:full_est_cntrl_error} \footnote{All the results presented in this paper are reproducible. The 
theoretical findings are annotated and step-by-step elaborated in the appendix. The data generation process and the parameter values used for numerical simulations are fully explained. In addition, the Python code from which the simulation figures are generated is available upon request.
The code is also on Github and the repository will go public upon submission acceptance.}. We generate  $\matF\!\sim\!\mathrm{GSBM}(\dimMatF,\matCom, \comSizes)$ and then determine $\matC=\matLinCoarse\matF\matLinCoarse^\transpose$. 
For a realization of $\matC$, we obtain the support of each row of $\matLinCoarse$ independently, by following a 3-step process. First, we draw a vector of i.i.d elements $\varrho\in[\dimMatCom]^\dimMatC$ from the distribution with probability vector $[\comSizes_k]_{k=1}^\dimMatCom$. The $i$-th element $\varrho_i\in[\dimMatCom]$ represent the community with which the $i$-th c-node has the most overlap.
Second, we generate a set of random vectors $\mbf{\pi}\sim \mathrm{Dir}(\dirichletPar, \cdots, \underbrace{1}_{\varrho_i\mathrm{th}}, \cdots, \dirichletPar) \in [0,1]^\dimMatCom$, where $\mathrm{Dir}(.)$ is the Dirichlet multivariate probability distribution parameterized by $0<\dirichletPar<1$ \cite{lin2016dirichlet}. $\mathrm{Dir}(\dirichletPar, \cdots, \underbrace{1}_{\varrho_i\mathrm{th}}, \cdots, \dirichletPar)$ outputs a vector of size $\dimMatCom$ whose $\varrho_i$-th component is the largest, the strength of other components is proportionate to $\dirichletPar$, and its elements sum to $1$. The Dirichlet distribution is a distribution over a probability simplex, whose output allocates probabilities to $\dimMatCom$ distinct categories (here communities).
Finally, for each community $\dimMatCom$ where $\mbf{\pi}_k>0$, $\lfloor\protoCovSize\cdot\mbf{\pi}_k\rfloor$ fine nodes are randomly chosen from $\mathcal{V}_k^{\mathrm{non-sel}}$ (i.e., non-selected fine indices in community $\dimMatCom$) and set as the support of $\vecLinCoarse_i$. The chosen indices are removed from $\mathcal{V}_k^{\mathrm{non-sel}}$. This process continues until the support of all $\vecLinCoarse_i$s are selected.
The first step  ensures that the community balancedness criterion illustrated in Section~\ref{sec: group-vs-coarse} is maintained, i.e. the coverage of each community after coarse-measuring remains (in expectation) proportionate to the fine-scale community sizes. The second step controls the level of community overlap (or synchronization); the smaller $\dirichletPar$ is chosen, the more coarse nodes get synced with communities, and vice versa (c.f., Remark~\ref{rmk: sync}).

We set number of fine nodes $\dimMatF=5000$, the overlap parameter $\dirichletPar=0.05$, and the number of communities $\dimMatCom=4$. Finally, for $\matCom=\scalingSBM\matCom_{\circ}\in \mathbb{R}^{K\times K}$, we set $[\matCom_{\circ}]_{kk}=p=0.5$, $[\matCom_{\circ}]_{kk'}=q=0.1$ (for $k\ne k'$). If not specified, the number of c-nodes $\dimMatC=100$, the density scaling parameter $\scalingSBM=0.1$, and the coverage size per c-node $\protoCovSize=10$. 		
Fig.~\ref{fig.errorVSm}-\ref{fig.errorVSoverlap} illustrate the qualitative behavior of the errors with respect to changes in the degree of (non-)synchronization in coarse nodes (i.e., $\dirichletPar$), $\dimMatC$, and $\scalingSBM$.
To fairly compare these errors, we also consider a \textit{baseline} { 
random vector $\bs{\mu} \in[1,2]^{\dimMatC}$ whose elements are generated (drawn) independently from the uniform distribution on the interval $[1,2]$, independently of $\matC$. We then define the \textit{baseline error} as 
{\small $ 
\left\|\frac{\bs{\mu}-\mbf{1}_\dimMatC}{\|\bs{\mu}-\mbf{1}_\dimMatC\|_1}-\frac{\protoCovSize\bs{\theta}_{\mathrm{group},{\matF}}-\mbf{1}_\dimMatC}{\|\protoCovSize\bs{\theta}_{\mathrm{group},{\matF}}-\mbf{1}_\dimMatC\|_1}\right\|_1$.} }
We make the following observations. First, both the learning and PROM-based errors are consistently better than the random baseline. Second, the learning-based approach has consistently smaller error than that of the PROM-based approach for large parametric regimes. Third, Fig.~\ref{fig.errorVSm} shows that all errors monotonically decrease as $m$ increases. Fourth, Fig.~\ref{fig.errorVSdensity} shows that errors decrease as $\scalingSBM$ increases. This is expected because larger values of $\scalingSBM$ result in more distant intra- and cross-community edge densities. This makes community representation extraction and controllability estimation easier. Finally, Fig.~\ref{fig.errorVSoverlap} demonstrates the higher tolerance of the Learning approach to situations wherein coarse measurements are less synchronized (larger $\dirichletPar$), in comparison to the PROM method. From Fig.~\ref{fig.errorVSoverlap} we observe that in the perfectly synchronized regime on the left of the x-axis, PROM performs very well. As the overlap increases up to the point around $\dirichletPar\simeq 0.1$, measurements become less synchronized with communities and both estimation errors goes up. After this point, the overlap becomes so high that the controllability values essentially become closer and closer to one another and result in a slight decrease in the PROM estimation error. These results are consistent with our bounds in Theorems \ref{thm: group-coarse-error2} and \ref{thm:full_est_cntrl_error}.
\begin{figure}
	\centering
	\subfloat[{w.r.t. no. coarse nodes $\dimMatC$.}]{\includegraphics[width=.48\textwidth]{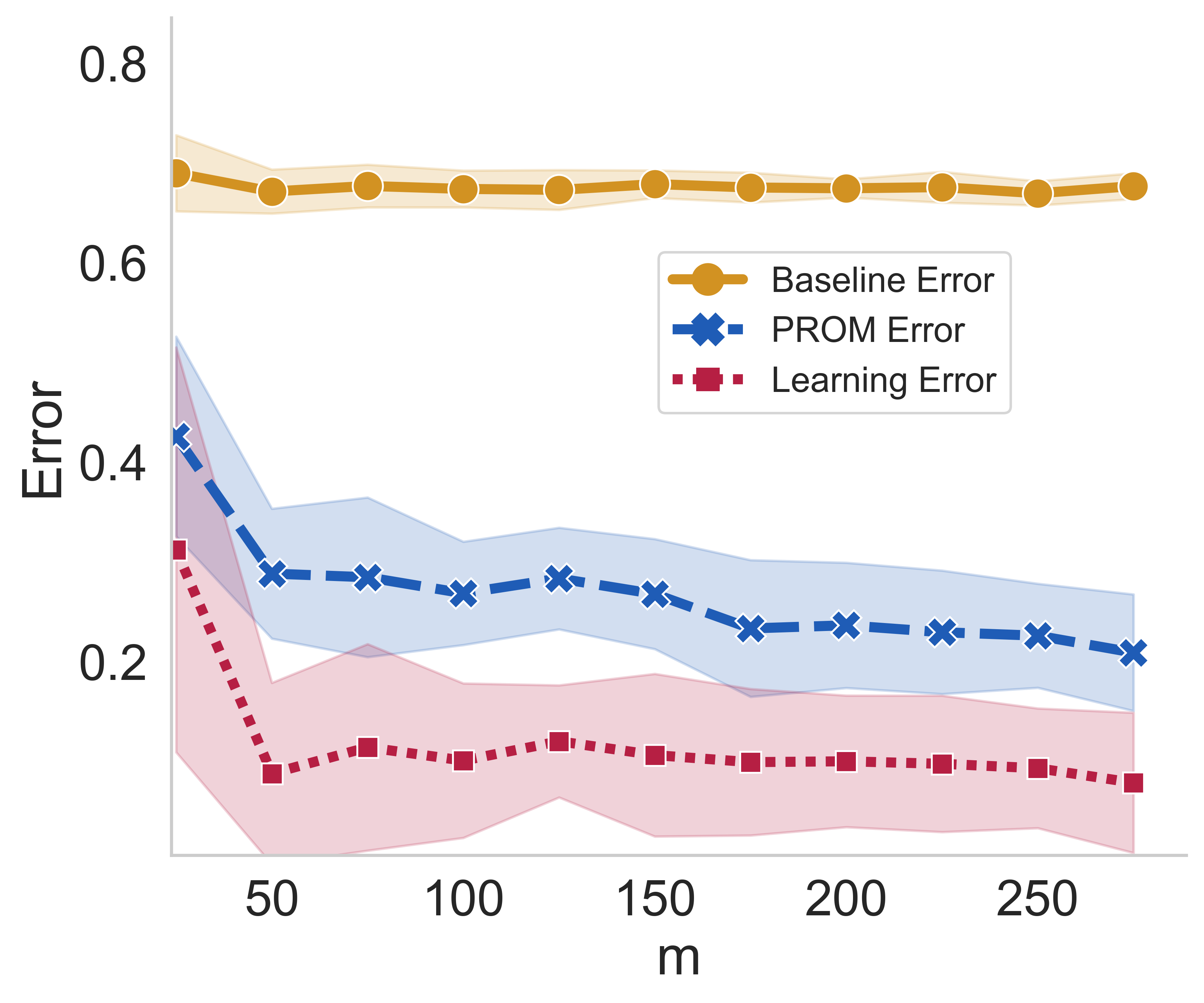}
		\label{fig.errorVSm}}
	\vfil
	\subfloat[{w.r.t. graph density scaling $\scalingSBM$.}]{\includegraphics[width=.48\textwidth]{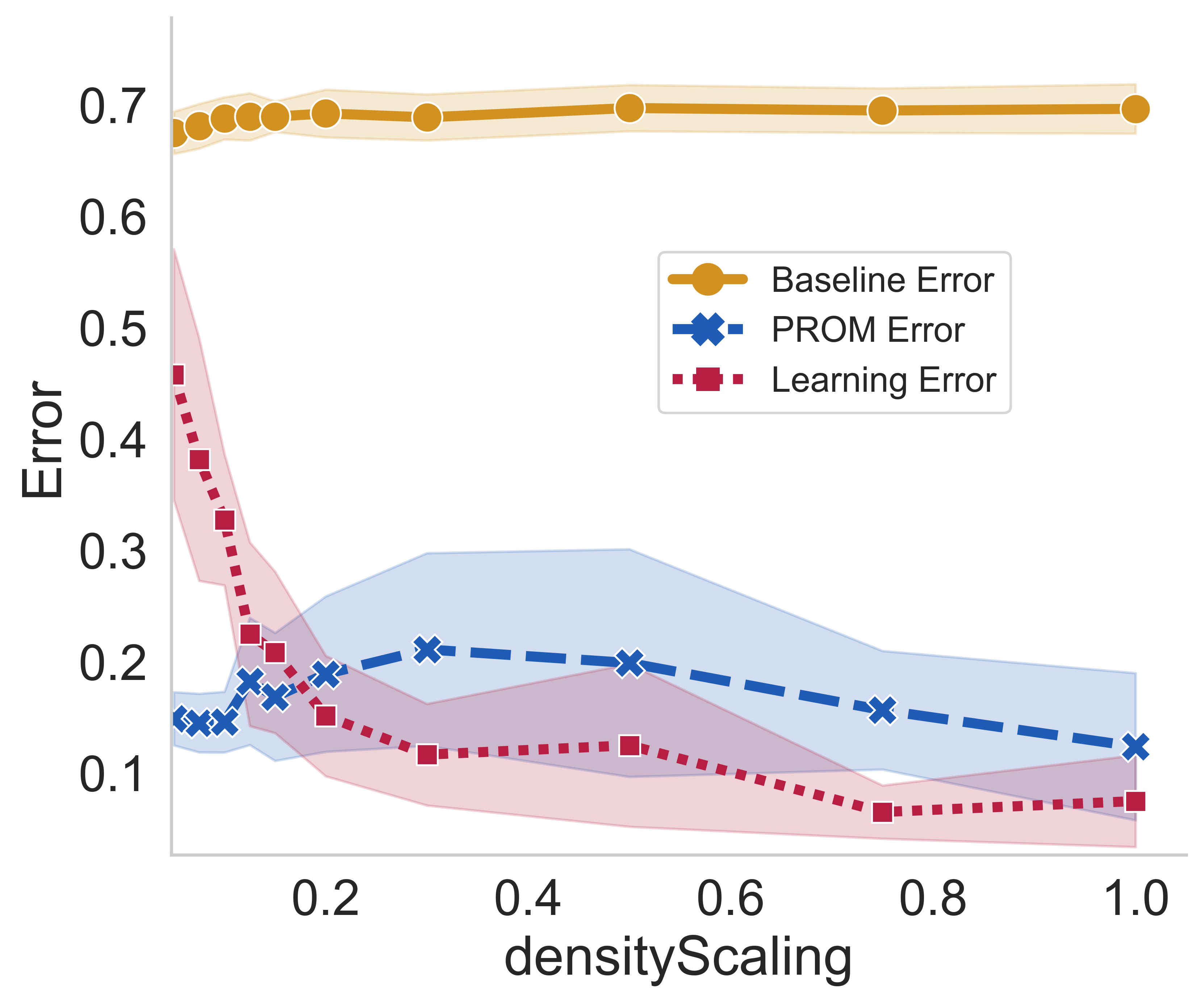}\label{fig.errorVSdensity}}
	\vfil
	\subfloat[{ w.r.t. overlap extent $\dirichletPar$.}]{\includegraphics[width=.48\textwidth]{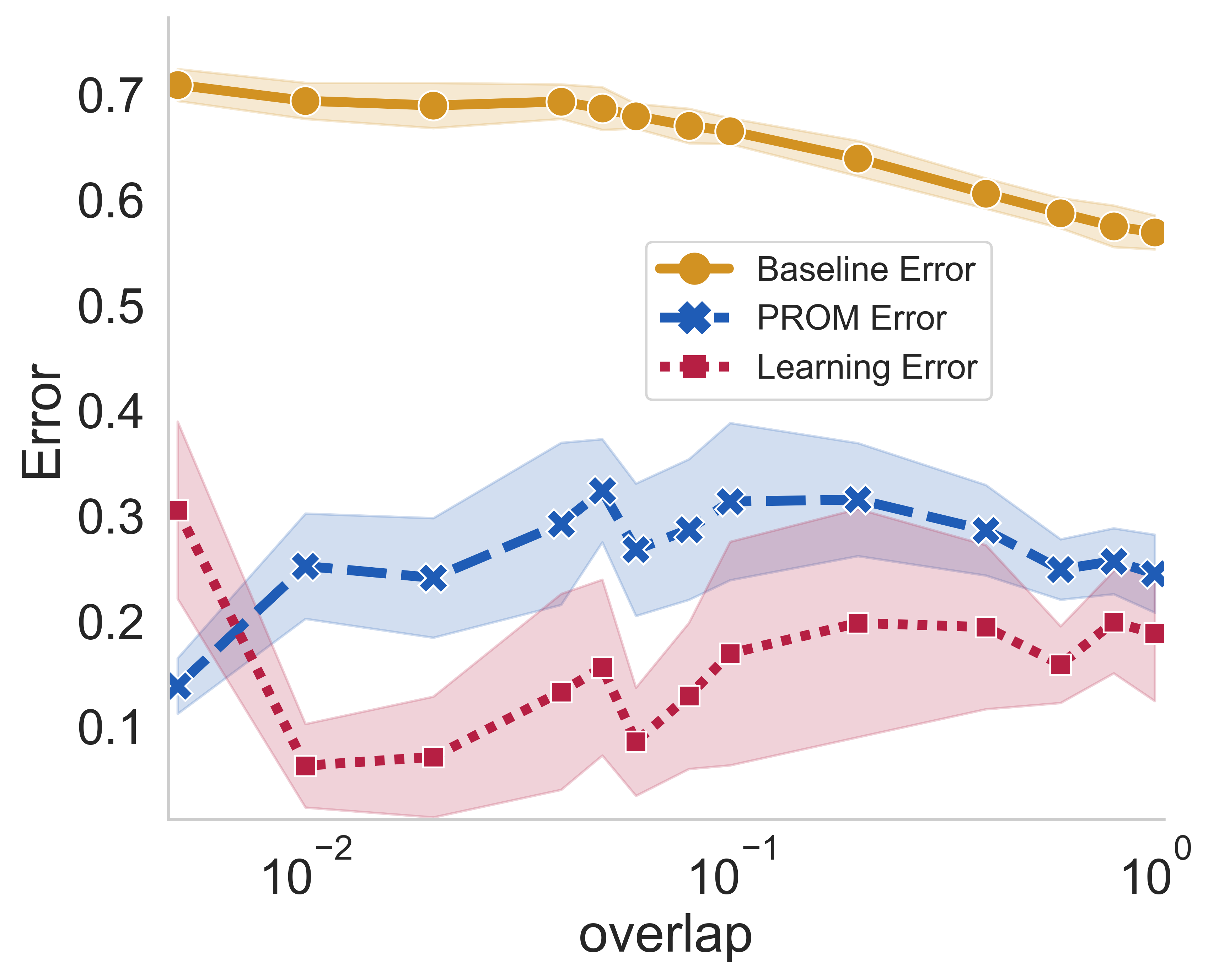}\label{fig.errorVSoverlap}}
	\caption{\small Estimation error based on PROM and Learning approaches. PROM Error=$\Delta(\matF,\matC)$ and Learning Error=$\widehat{\Delta}(\matC)$. The shaded region in the figures represent one standard deviation computed for $20$ independent realizations. }
\end{figure}

		
\section{Conclusion and Future work}\label{sec:futureWork}
{
We introduced a learning-based framework that exploits the power of community-based representation learning to infer average controllability of fine graphs from coarse summary data. We compared the performance of this approach with that of the Passively Reduced-Order Model (PROM) approach. For both these methods, we derived high probability error bounds on the deviation between the error estimate and ground truth, and validated the theory with numerical simulations. Our results highlight the role of fine- and coarse-network sizes, graph density, and measurement synchronization with communities in modulating the estimation errors. Interestingly, for the latter approach, we show that the estimation error decreases with network size albeit the synchronization bias, which is not the case with the PROM-based approach. For future, we plan to implement our theory to study the role of coarsening, community structures, and synchronization aspects on the controllability of brain networks. Extending the proposed tools and techniques to controllability metrics other than average controllability is left for future work.
}

\appendix
\section{Appendix}
{
This section contains proofs for the main results in Sections \ref{sec: group-vs-coarse} and \ref{sec:est_solution}.


	\noindent{\bf Notations}: 
	$\mbf{1}_{\dimMatCom,\dimMatCom}$ denotes an all-one matrix of dimension $\dimMatCom$ by $\dimMatCom$.

\noindent{\bf Useful Matrix Norm Equivalence \cite{liu2009trace,horn2012matrix, chapter5April26folder}}: 
For two real, square, and symmetric matrices $\genMat,\genMatOth \in \mathbb{R}^{\genDim \times \genDim}$, the following equality and inequalities hold:
\begin{enumerate}
	\item \label{item:dffInvToInv} If $\genMat,\genMatOth$ are invertible: $\genMat^{-1} - {\genMatOth}^{-1} = \genMat^{-1}({\genMatOth}-\genMat){\genMatOth}^{-1}$.
	\item \label{item:trace_UB_trace} $\trace(\genMat\genMatOth)\leq \|\genMat\|_2 \trace(\genMatOth)$.
	\item \label{item:resolvNormUB} For all norms:
    if $\| \genMat\| <1 \Rightarrow \|  (\mathbf{I}-\genMat)^{-1}\| <\frac{1}{1-\| \genMat\| }$.
    \item \label{item:maxInftyNorms} $\| \genMat\genMatOth\|_{\mathrm{max}}\!\leq\! \| \genMat\|_\infty \| \genMatOth\|_{\mathrm{max}}$. 
\item \label{item:InftyNormsqrtMax} $ \|\genMat\|_\infty \leq \genDim \| \genMat\|_{\mathrm{max}} $.
\item \label{item:twoNormMaxIneq} $\|\genMat\|_2 \leq \genDim \| \genMat\|_{\mathrm{max}}$.
\item \label{item:maxProdMax} $\|\genMat\genMatOth\|_{\mathrm{max}} \leq \genDim \|\genMat\|_{\mathrm{max}}\|\genMatOth\|_{\mathrm{max}}$
\end{enumerate}
Note that when $\genMat$ is square and symmetric, the 2-norm and the spectral radius coincide and can be used interchangeably, i.e. $\rho(\genMat)=\|\genMat\|_2$.


\begin{prop}\label{prop:hoeffdingIneq}
\textbf{(Hoeffding's Inequality) \cite{boucheron2013concentration}}: {
Let $\{\genVec_\ell\}_{\ell=1}^{|\genVec|}$ be bounded independent random variables. Then the following holds with probability at least $1-\hoeffProbConst$}, for some $\delta >0$: 
\begin{align}
     \left\lvert \sum_{\ell=1}^{\lvert\genVec\rvert} \genVec_\ell\!  -\!\mathbb{E}\!\left[\!\sum_{\ell=1}^{\lvert\genVec\rvert} \genVec_\ell\!\right]\right\rvert
     & \!\leq\! \sqrt{\frac{\lvert\genVec\rvert}{2}\log\left(\frac{1}{\hoeffProbConst}\right)}(\max(\genVec)\!-\!\min(\genVec))  \\
     &
     \leq \sqrt{\frac{\lvert\genVec\rvert}{2}\log\left(\frac{1}{\hoeffProbConst}\right)}\frac{\max(\genVec)\!-\!\min(\genVec)}{\mathbb{E}\left[\sum_{\ell=1}^{\lvert\genVec\rvert} \genVec_\ell\right]} \mathbb{E}\left[\sum_{\ell=1}^{\lvert\genVec\rvert} \genVec_\ell\right]
     \\
     &
     \leq\! \sqrt{\frac{\lvert\genVec\rvert}{2}\log(\!\frac{1}{\hoeffProbConst}\!)}\frac{\max(\genVec)\!-\!\min(\genVec)}{\min(\mathbb{E}\!\left[\sum_{\ell=1}^{\lvert\genVec\rvert} \genVec_\ell\right])} \mathbb{E}\!\left[\!\displaystyle\sum_{\ell=1}^{\lvert\genVec\rvert} \genVec_\ell\!\right]\!,
\end{align}
where $\max(\genVec)$ and $\min(\genVec)$ respectively denote the maximum and minimum possible values of $\{\genVec_\ell\}_{\ell=1}^{\lvert \genVec\rvert}$. \qed
\end{prop}

We need the following results to prove Lemma~\ref{lma: Gramian differences} and Theorem~\ref{thm:full_est_cntrl_error}. 
\begin{prop}\label{prop:matrixElemInequExtends2norms}
\textbf{(Monotone Matrix Norms \cite{tyrtyshnikov1997brief}):} Suppose the element-wise inequality $\genMat\geq\genMatOth$ holds. Then, $\|\genMat\|_p\geq \|\genMatOth\|_p$, for a positive integer $p$.\qed
\end{prop}

\begin{prop}\label{prop:uppBound_relDiff}
\textbf{(An upper bound on the $\ell_1$ norm of the difference between two $\ell_1$-normalized vectors):} Let $\genVec,\genVecOth \in \mathbf{R}^{m}$ be such that $\genVec\geq \mbf{1}_\dimMatC$ and $\genVecOth\geq \mbf{1}_\dimMatC$, where $\mbf{1}_m$ is the all-ones vector. Then, 
\begin{align*} 
	\left\|\frac{\genVec-\mbf{1}_\dimMatC}{\|\genVec\!-\!\mbf{1}_\dimMatC\|_1} \!-\! \frac{\genVecOth-\mbf{1}_\dimMatC}{\|\genVecOth-\mbf{1}_\dimMatC\|_1} \right\|_1 
	 \!\leq\! 2 \frac{\|\genVec\!-\!\genVecOth\|_1}{\max\left({\|\genVec\!-\!\mbf{1}_\dimMatC\|_1, \|\genVecOth\!-\!\mbf{1}_\dimMatC\|_1}\right)}.
\end{align*}\qed
\end{prop}
\begin{proof}
Using the triangle inequality, we have:
\begin{align}\label{diffff}
	\left\| \frac{\genVec-1}{\left\| \genVec -\mbf{1}_\dimMatC\right\|_1} - \frac{\genVecOth-1}{\left\| \genVecOth -\mbf{1}_\dimMatC\right\|_1} \right\|_1
	\nonumber & = \displaystyle\sum_{i=1}^{\dimMatC}  \left\lvert\frac{\genVec_i-1-(\genVecOth_i-1)+(\genVecOth_i-1)}{\left\| \genVec -\mbf{1}_\dimMatC\right\|_1} - \frac{\genVecOth_i-1}{\left\| \genVecOth -\mbf{1}_\dimMatC\right\|_1} \right\rvert 
	\\ 
	\nonumber& = \displaystyle\sum_{i=1}^{\dimMatC}  \left\lvert\frac{\genVec_i - \genVecOth_i}{\left\| \genVec -\mbf{1}_\dimMatC\right\|_1} - [\genVecOth_i-1]\left[\frac{1}{\left\| \genVecOth -\mbf{1}_\dimMatC\right\|_1} -\frac{1}{\left\| \genVec -\mbf{1}_\dimMatC\right\|_1} \right] \right\rvert 
	\\ 
	\nonumber& = \displaystyle\sum_{i=1}^{\dimMatC}  \left\lvert\frac{\genVec_i - \genVecOth_i}{\left\| \genVec -\mbf{1}_\dimMatC\right\|_1} - [\genVecOth_i-1]\left[\frac{\displaystyle\sum_{j=1}^{\dimMatC} (\genVec_j-\genVecOth_j)}{\left\| \genVecOth -\mbf{1}_\dimMatC\right\|_1 \left\| \genVec -\mbf{1}_\dimMatC\right\|_1} \right] \right\rvert 
	\\ 
	& \leq \frac{\displaystyle\sum_{i=1}^{\dimMatC} \left\vert\genVec_i-\genVecOth_i\right\vert}{\|\genVec-\mbf{1}_\dimMatC\|_1} \left[1+ \frac{\displaystyle\sum_{i=1}^{\dimMatC} \left\vert\genVecOth_i-1\right\vert}{\|\genVecOth-\mbf{1}_\dimMatC\|_1} \right]  
	 = 2 \frac{\|\genVec-\genVecOth\|_1}{\|\genVec-\mbf{1}_\dimMatC\|_1}. 
\end{align}
By exchanging the roles of $\genVec_i$ and $\genVecOth_i$, we find a similar bound that combined to \eqref{diffff} gives the inequality in the statement of Lemma.
\end{proof}

\subsection{Lemma
and Proof: 
Relationship Between Fine- and Group- Average Controllability}

Let $\genMat$ be an $\dimMatF\times \dimMatF$ system matrix that could be either $\mbf{A}$ or
$\bar{\matF}$ (the expected quantity). 
Recall the controllability Gramian of the LTI system in  \eqref{generalLTI} has been defined in \eqref{gramianFine_def}.
We define the $\dimMatF$-dimensional \textit{fine} average controllability vector $\bs{\theta}_{\mathrm{fine},\genMat}$ similar to \eqref{eq: group avg-controllability} as 
\begin{align}\label{cntrlblty_fine_def}
\bs{\theta}_{\mathrm{fine},\genMat}^\transpose\triangleq \begin{bmatrix}
			\trace[\Cgram(\genMat_{\mathrm{nom}}, \mbf{e}_1)]& \ldots& \trace[\Cgram(\genMat_{\mathrm{nom}}, \mbf{e}_\dimMatF)]
		\end{bmatrix}, 
\end{align} 
where $\mbf{e}_i$ is the $i$-th canonical basis vector in $\mathbb{R}^n$. 
By definition in \eqref{gramianFine_def}, the following holds: 
\begin{align}\label{gramian_to_diag}
   \bs{\theta}^{(i)}_{\mathrm{fine},\genMat} & = \trace\left[\Cgram(\genMat_{\mathrm{nom}}, \mbf{e}_i)\right]= \displaystyle\sum_{\tau=0}^{\infty} \trace( \genMat_{\mathrm{nom}}^\tau  \mathbf{e}_i \mathbf{e}_{i}^\transpose \genMat_{\mathrm{nom}}^\tau) \nonumber\\
   &=\displaystyle\sum_{\tau=0}^{\infty} [\mathrm{diag}(\genMat_{\mathrm{nom}}^{2\tau} )]_i = 1 + [\mathrm{diag}(\genMat_{\mathrm{nom}}^2)]_i +\cdots \nonumber \\
   & = [\mathrm{diag}\left((\mbf{I} - \genMat_{\mathrm{nom}}^2)^{-1} \right)]_i.
\end{align}
The series converges to the final equality because $\rho(\mbf{Z})\leq 1$. 
Putting \eqref{gramian_to_diag} into a vector form gives:
\begin{align}\label{theta_fine}
\hspace{-2.5mm}\bs{\theta}_{\mathrm{fine},\genMat}\!=\! \mathrm{diag}\left(\Cgram(\genMat_\mathrm{nom}, \mbf{I}_{\dimMatF\times \dimMatF})\right)\!=\!\mathrm{diag}\left((\mbf{I} - \genMat_{\mathrm{nom}}^2)^{-1} \right).
\end{align}
The following lemma states the group average controllability vector in \eqref{eq: group avg-controllability} is a linear mapping of the fine average controllability vector defined above. 
\begin{lemma}\label{lemma:groupTheta} With the notation defined above for $\mbf{Z}$ consider the average controllability vector in \eqref{eq: group avg-controllability}. Then
\begin{align*}
	\bs{\theta}_{\mathrm{group},\genMat} = \frac{1}{\protoCovSize}\matLinCoarse \bs{\theta}_{\mathrm{fine},\genMat}. 
\end{align*}\qed
\end{lemma}

\begin{proof}
We begin by simplifying the group average controllability using its definition in \eqref{eq: group avg-controllability} and the definition of Gramian in \eqref{gramianFine_def}, for a general matrix $\genMat$:
\begin{align}\label{elem_c_cntrl}
	\nonumber\bs{\theta}_{\mathrm{group},\genMat}^{(i)} & = \trace\left[\Cgram(\genMat_{\mathrm{nom}}, \mathrm{diag}(\vecLinCoarse_i^\transpose))\right] \\
	 \nonumber& =\trace\left[\displaystyle\sum_{\tau=0}^{\infty} \genMat_{\mathrm{nom}}^\tau \mathrm{diag}(\vecLinCoarse_i^\transpose) \mathrm{diag}(\vecLinCoarse_i^\transpose)^{\transpose} \genMat_{\mathrm{nom}}^\tau\right] \\
	\nonumber& = \trace\left[\displaystyle\sum_{\tau=0}^{\infty} \genMat_{\mathrm{nom}}^\tau \mathrm{diag}((\vecLinCoarse_i \odot \vecLinCoarse_i)^\transpose) \genMat_{\mathrm{nom}}^\tau\right]\\
		\nonumber& \stackrel{\mathrm{(a)}}{=} \frac{1}{{\protoCovSize^2}}\trace\left[\displaystyle\sum_{\tau=0}^{\infty} \genMat_{\mathrm{nom}}^\tau (\displaystyle\sum_{v\in\mathrm{supp}(\vecLinCoarse_i)} \mathbf{e}_v \mathbf{e}_{v}^\transpose) \genMat_{\mathrm{nom}}^\tau\right]\\  
		\nonumber& = \frac{1}{{\protoCovSize^2}} \displaystyle\sum_{v\in\mathrm{supp}(\vecLinCoarse_i)}\underbrace{\trace\left[\displaystyle\sum_{\tau=0}^{\infty} \genMat_{\mathrm{nom}}^\tau  \mathbf{e}_v \mathbf{e}_{v}^\transpose \genMat_{\mathrm{nom}}^\tau\right]}_{\bs{\theta}_{\mathrm{fine},\genMat}^{(v)}} \\
		& = (\vecLinCoarse_i \odot \vecLinCoarse_i)^\transpose \bs{\theta}_{\mathrm{fine},\genMat},
\end{align}
where (a) is due to the assumption of $\mathrm{supp}(\vecLinCoarse_i)=\protoCovSize$ at the end of Section \ref{subsec:probstatement}, for all $i\in[\dimMatC]$; $\odot$ denotes the Hadamard product; and $\bs{\theta}_{\mathrm{fine},\genMat}$ has been defined in \eqref{theta_fine} in which $\matF$ is substituted with $\genMat$.
Putting \eqref{elem_c_cntrl} in vector form concludes the proof.
\end{proof}
The above relationship states that the group controllability $\bs{\theta}_{\mathrm{group},\genMat}$ {  is the mean value of the fine controllability of the groups of nodes that map to each c-node, weighted by the factor $1/\protoCovSize$ (that is rooted in the definition of the coarse-measurement matrix $\matLinCoarse$).}  


The following results are instrumental 
in proving Lemma~\ref{lemma:approximate_gAC} and Lemma~\ref{lma: Gramian differences}. 
\begin{lemma}\label{lemma:spectralAbarEquivalence}
\textbf{(spectral radius of $\bar{\matF}$):} Let $\diagComSize=\frac{1}{\dimMatF}\comFine\comFine^\transpose$ with
$\scalingSBM$ be defined as in \eqref{eq: scalingSBM}. Then the spectral radius of $\bar{\matF}=\comFine^\transpose \matCom \comFine$ in \eqref{def_coarseMembershipMat} is given as 
\begin{align*}
     \specRadius(\bar{\matF})= \dimMatF\specRadius(\matCom \diagComSize)= \dimMatF\scalingSBM\specRadius(\matCom_{\circ} \diagComSize).
\end{align*}\qed
\end{lemma}
\begin{proof}
Recall that the eigenvalues of $\mbf{A}\mbf{B}$ and $\mbf{B}\mbf{A}$ coincide for any two square matrices $\mbf{A}$ and $\mbf{B}$. Thus, 
\begin{align*} 
    \specRadius(\bar{\matF})= \specRadius(\matCom \comFine\comFine^\transpose ) = \dimMatF\specRadius(\matCom \diagComSize)= \dimMatF\scalingSBM\specRadius(\matCom_{\circ} \diagComSize), 
\end{align*}
where for the last equality we use the relation in \eqref{eq: scalingSBM}.
\end{proof} 

The result below gives an equivalent expression for the spectral radius of the expected coarse-scale adjacency matrix $\bar{\matC}$. 
\begin{lemma}\label{lemma:spectralAtildebarEquivalence}
\textbf{(spectral radius of $\bar{\matC}$):} The spectral radius of $\bar{\matC}$ can be expanded as
\begin{align}
    \specRadius(\bar{\matC}) = \specRadius(\syncComAs \matCom \syncComAs^\transpose)= \dimMatC \specRadius\left(\matCom \frac{\syncComAs^\transpose \syncComAs}{\dimMatC}\right).
\end{align}\qed
\end{lemma}
\begin{proof}
Similar to Lemma~\ref{lemma:spectralAbarEquivalence}. Details are omitted. 
\end{proof}

\vspace{-1em}
\subsection{Lemma and proof: Error between the Gramians of random and expected LTI systems}

Let $\overline{\mathcal{S}}_\mathrm{fine}$ and $\overline{\mathcal{S}}_\mathrm{coarse}$ denote the expected dynamics of LTI \eqref{fine_LTI} when $\matF$ and $\matC$ are replaced with the expected quantities $\bar{\matF}$ and $\bar{\matC}$. The following result provides an error bound on the respective differences of the fine- and coarse- average controllability vectors, which is equivalent to bound the difference between the Gramians of ${\mathcal{S}}_\mathrm{fine}$ and $\overline{\mathcal{S}}_\mathrm{fine}$ and that of  ${\mathcal{S}}_\mathrm{coarse}$ and $\overline{\mathcal{S}}_\mathrm{coarse}$. 

\begin{lemma}\label{lma: Gramian differences}(\textbf{Error between the Gramians of random and expected LTI systems}): Let $\bs{\theta}_{\mathrm{fine},\matF}, \bs{\theta}_{\mathrm{fine},\bar{\matF}}$ be defined as in \eqref{theta_fine}, and $\bs{\theta}_{\mathrm{coarse},\matC}, \bs{\theta}_{\mathrm{coarse},\bar{\matC}}$ as in \eqref{eq: coarse avg-controllability}. Under Assumption~\ref{assump: graph sparsity}, the following holds with probability at least $1-\hoeffProbConst$: 
\begin{align}\label{eq: Gramian differences alpha}
     \alpha_\dimMatF & \triangleq \left\| \bs{\theta}_{\mathrm{fine},\matF}-\bs{\theta}_{\mathrm{fine},\bar{\matF}} \right\|_1 
    \\  &
    \leq
    \frac{ \hoeffCoeffConst\left[\frac{1/\scalingSBM\!+\!{c}_1}{ {c}_\circ^2}+ (\frac{1}{2}\!+\!2\hoeffCoeffConst)\frac{\left[{c}_1 \!+\! \frac{1}{\scalingSBM}\right]}{{c}_\circ^2}
      \left( 1 \!+ \!\sqrt{1+\hoeffCoeffConst} \sqrt{1\!+\! \frac{ 1/\scalingSBM}{\dimMatF  {c}_\circ^2}} \!+\! \frac{2\constFine/\scalingSBM}{\dimMatF  {c}_\circ} \right)\right]}{\left[1\!-\!\left(\frac{{c}_\circ}{\constFine/\scalingSBM\!+\!{c}_\circ}\right)^2\right] \left[1\!-\!\left(\frac{\specRadius(\matF)}{\constFine+\specRadius(\matF)}\right)^2\right]}
      \\ & 
      = \Oscale(\frac{\hoeffCoeffConst}{\scalingSBM}), 
\end{align} 
and with probability at least $1-\tilde{\hoeffProbConst}$: 
\begin{align}\label{eq: Gramian differences alphatilde}
 \tilde{\alpha}_\dimMatC & \triangleq  \left\| \bs{\theta}_{\mathrm{coarse},\matC}-\bs{\theta}_{\mathrm{coarse},\bar{\matC}} \right\|_1 
\\ 
	& \leq \frac{ \tilde{\hoeffCoeffConst}\left[\frac{1/\scalingSBM+\tilde{c}_1}{\tilde{c}_\circ^2}+ (\frac{1}{2}+2\tilde{\hoeffCoeffConst})\frac{\left[\tilde{c}_1 + \frac{1}{\scalingSBM}\right]}{\tilde{c}_\circ^2}
      \left( 1 \!+ \!\sqrt{1+\tilde{\hoeffCoeffConst}} \sqrt{1+ \frac{ 1/\scalingSBM}{\dimMatC  \tilde{c}_\circ^2}} + \frac{2\constCoarse/\scalingSBM}{\dimMatC  \tilde{c}_\circ} \right)\right]}{\left[1-\left(\frac{\tilde{c}_\circ}{\constCoarse/\scalingSBM+\tilde{c}_\circ}\right)^2\right] \left[1-\left(\frac{\specRadius(\matC)}{\constCoarse+\specRadius(\matC)}\right)^2\right]}
      \\ & 
      = \Oscale(\frac{\tilde{\hoeffCoeffConst}}{\scalingSBM}),
\end{align}
where 
\begin{align}\label{hoeffCoeffConst_def}
        \nonumber \hoeffCoeffConst = \frac{\sqrt{{\log(\frac{\dimMatF^2}{\hoeffProbConst})}/{(2\dimMatF)}}}{\left\|\matCom \diagComSize \matCom\right\|_{\mathrm{min}}}, & \quad 
    \tilde{\hoeffCoeffConst} = \frac{\sqrt{{\log(\frac{\dimMatC^2}{\tilde{\hoeffProbConst}})}/{(2\dimMatC)}}}{ \| \matCom \frac{\syncComAs^\transpose \syncComAs}{\dimMatC} \matCom\|_{\mathrm{min}}},  \\
        \nonumber {c}_\circ=\specRadius(\matCom_{\circ} \diagComSize), & \quad
         {c}_1=\trace\left((\diagComSize\matCom_{\circ})^2\right)
         \\
         \tilde{c}_\circ= \specRadius\left(\matCom_{\circ} \frac{\syncComAs^\transpose \syncComAs}{\dimMatC}\right), & \quad
         \tilde{c}_1=\trace\left((\frac{\syncComAs^\transpose \syncComAs}{\dimMatC}\matCom_{\circ})^2\right). 
\end{align}\qed
\end{lemma}

Note that $\Cgram(\genMat, \mbf{I})=(\mbf{I}-\genMat^2)^{-1}$, where $\genMat$ can take $\matFnom$, $\matFBarNom$, $\matCnom$, or  $\matCBarNom$. Lemma \ref{lma: Gramian differences} is effectively bounding the difference of resolvents $(z\mbf{I}-\genMat^2)^{-1}$ evaluated at $z=1$. For $\matF=\matC$, we have $\dimMatF=\dimMatC$ and $\frac{\syncComAs^\transpose \syncComAs}{\dimMatC}=\diagComSize$ and hence, both $\alpha_\dimMatF$ and $\wt{\alpha}_n$ coincide. 
Lemma \ref{lma: Gramian differences} is basically a concentration result for the Gramians of $\mathcal{S}_\mathrm{fine}$ (or $\mathcal{S}_\mathrm{coarse}$) and $\overline{\mathcal{S}}_\mathrm{fine}$ (or $\overline{\mathcal{S}}_\mathrm{coarse}$). However the rate at which the difference goes to zero is different for the fine- and coarse systems.

\begin{proof}
We start by simplifying $\alpha_\dimMatF$ in \eqref{eq: Gramian differences alpha} (explanations for each step succeed the equations):
\begin{align}\label{alphanSteps0}
	 \nonumber\alpha_\dimMatF &= \left\| \bs{\theta}_{\mathrm{fine},\matF}-\bs{\theta}_{\mathrm{fine},\bar{\matF}} \right\|_1
	\\ 
	\nonumber&  \stackrel{(a)}{=} \left\|\mathrm{diag}\left[\Cgram(\matFnom, \mbf{I}_{\dimMatF\times \dimMatF})-\Cgram(\matFBarNom, \mbf{I}_{\dimMatF\times \dimMatF})\right]\right\|_1
	\\ 
	 \nonumber&  \stackrel{(b)}{=} \trace\left(\left\lvert\Cgram(\matFnom, \mbf{I}_{\dimMatF\times \dimMatF})-\Cgram(\matFBarNom, \mbf{I}_{\dimMatF\times \dimMatF})\right\rvert\right) 
	\\ 
	\nonumber&  \stackrel{(c)}{=} \trace\left(\left\lvert(\mbf{I}-\matFnom^2)^{-1}-(\mbf{I}-\matFBarNom^2)^{-1}\right\rvert\right)
	 \\ 
	 \nonumber& \stackrel{(d)}{=} \trace\left((\mbf{I}\!-\!\matFnom^2)^{-1}\lvert\matFnom^2\!-\!\matFBarNom^2\lvert(\mbf{I}\!-\!\matFBarNom^2)^{-1}\right)
	\\ 
	 \nonumber&  \stackrel{(e)}{\leq} \left\|(\mbf{I}-\matFnom^2)^{-1}\|_2 \trace\left(\left\lvert\matFnom^2-\matFBarNom^2\right\rvert\right) \|(\mbf{I}-\matFBarNom^2)^{-1}\right\|_2
	  \\
	&  \stackrel{(f)}{\leq} \frac{1}{1-\|\matFnom^2\|_2} \frac{1}{1-\|\matFBarNom^2\|_2}  \trace\left(\left\lvert\matFnom^2-\matFBarNom^2\right\rvert\right),
\end{align} 
where step (a) is due to the equivalent derivation of the fine controllability vector in \eqref{theta_fine}; (b) is due to the equivalence of the of trace of the absolute value of a matrix and the $\ell_1$ norm of its diagonal entries; (d)-(f) follow respectively from items~ \ref{item:dffInvToInv}-\ref{item:resolvNormUB} in the matrix equivalence list at the beginning of the appendix; finally, (c) follows from the definition of the Gramian matrix in \eqref{gramianFine_def} and the Neumann series formula such that 
 the limit below exists
\begin{align}\label{def:matFGramian}
    \nonumber\Cgram(\matFnom, \mbf{I}_{\dimMatF\times \dimMatF})&\triangleq\lim\limits_{T\to \infty}\Cgram_T(\matFnom, \mbf{I}_{\dimMatF\times \dimMatF}) 
    \\
    \nonumber& =
    \lim_{T\to \infty}\sum_{t=0}^{T-1} (\matFnom)^t(\matFnom^\transpose)^t \\
    & =\lim\limits_{T\to \infty}\sum_{t=0}^{T-1} (\matFnom)^{2t} = (\mbf{I}-\matFnom^2)^{-1},
    \end{align}
    because $\matF$ is a symmetric matrix and $\| \matFnom\|_2 =  \frac{\specRadius(\matF)}{\constFine+\specRadius(\matF)}<1$. Similarly, we have $\| \matFBarNom\|_2 =  \frac{\specRadius(\bar{\matF})}{\constFine+\specRadius(\bar{\matF})}<1$, and hence
    \begin{align*}
    \Cgram(\matFBarNom, \mbf{I}_{\dimMatF\times \dimMatF})\triangleq\lim\limits_{T\to \infty}\Cgram_T(\matFBarNom, \mbf{I}_{\dimMatF\times \dimMatF}) 
    =(\mbf{I}-\matFBarNom^2)^{-1}.
\end{align*}

In \eqref{alphanSteps0}, the two terms $\|\matFnom^2\|_2<1, \|\matFBarNom^2\|_2<1$ by definition.
In the following, we will find a bound on the remaining term $\trace(\lvert\matFnom^2-\matFBarNom^2\rvert)$.

Using the Hoeffding's inequality in Proposition~\ref{prop:hoeffdingIneq}, the following holds for $i,j\in[\dimMatF]$ and $i\neq j$, with probability at least $1-\hoeffProbConst$:
\begin{align}\label{probA2NonDiag}
        \nonumber \left\lvert [\matF^2]_{ij} - \mathbb{E}[{\matF}^2]_{ij}\right\rvert  
        & = \left\lvert \displaystyle\sum_{\ell \in [\dimMatF]} {\matF_{i\ell}\matF_{j\ell}} - {\mathbb{E}[\displaystyle\sum\matF_{i\ell}\matF_{j\ell}]}\right\rvert  \\
         \nonumber & \leq\! \sqrt{\!\frac{\dimMatF\log(\!\frac{1}{\hoeffProbConst}\!)}{2}}\frac{\max(\matF_{i\ell})\shortminus\min(\matF_{i\ell})}{\min([\bar{\matF}^2]_{ij})} [\bar{\matF}^2]_{ij}
         \\
         & \stackrel{(a)}{\leq} \sqrt{\frac{\log\left(\frac{1}{\hoeffProbConst}\right)}{2\dimMatF}}\frac{1}{ \left\|\matCom \diagComSize \matCom\right\|_{\mathrm{min}}} [\bar{\matF}^2]_{ij},
\end{align} 
where (a) follows by definitions of $\diagComSize,\comFine$ in Section \ref{subsec:SBMmodel} coupled with the fact that
\begin{align*}
         \bar{\matF}^2 = \comFine \matCom \comFine^\transpose \comFine \matCom \comFine^\transpose  = \dimMatF \comFine \matCom \diagComSize  \matCom \comFine^\transpose.
\end{align*}
Similarly, for the diagonal elements, the following inequality holds with probability at least $1-\hoeffProbConst$:
\begin{align}\label{probA2Diag1}
        \nonumber\left\lvert [\matF^2]_{ii} - \mathbb{E}[{\matF}^2]_{ii}\right\rvert 
        & = \!\left\lvert \displaystyle\sum_{\ell \in [\dimMatF]} {\matF_{i\ell}^2} \shortminus \mathbb{E}[{\matF_{i\ell}^2}]\right\rvert 
         \!=\!  \left\lvert \displaystyle\sum_{\ell \in [\dimMatF]} {\matF_{i\ell}} \shortminus \mathbb{E}[{\matF_{i\ell}}]\right\rvert
        \\ 
        \nonumber & \leq \! \sqrt{\!\frac{\dimMatF\log(\!\frac{1}{\hoeffProbConst}\!)}{2}}\frac{\max(\matF_{i\ell})\shortminus\min(\matF_{i\ell})}{\min({\mathbb{E}\left[\matF^2\right]_{ii}})} \mathbb{E}[{\matF}^2]_{ii}
        \\
        & \stackrel{(a)}{\leq} \sqrt{\frac{\log\left(\frac{1}{\hoeffProbConst}\right)}{2\dimMatF}}\frac{1}{ \|\matCom \comSizes\|_{\mathrm{min}}} \mathbb{E}[{\matF}^2]_{ii},
\end{align}
where (a) is due to the identity $\mathbb{E}\left[\matF^2\right]_{ii}=\comFine_i \matCom \comFine \mbf{1}_\dimMatF=\comFine_i \matCom \comSizes$.
By merging \eqref{probA2NonDiag}, \eqref{probA2Diag1} and putting them in a matrix form, and by defining $\hoeffCoeffConst$ as
\begin{align*}
    \frac{1}{\sqrt{2\dimMatF}}\!\max\!\left(\! \frac{\sqrt{\log(\frac{\dimMatF^2\!-\!\dimMatF}{\hoeffProbConst})}}{\|\matCom  \diagComSize \matCom \|_{\mathrm{min}}},
    \frac{\scalingSBM\sqrt{\log(\frac{\dimMatF}{\hoeffProbConst})}}{\|\matCom \! \comSizes\|_{\mathrm{min}}}\!\right) 
    \!\leq \!\sqrt{\!\frac{\log\left(\!\frac{\dimMatF^2}{\hoeffProbConst}\!\right)}{2\dimMatF}} \frac{1}{\|\matCom \!\diagComSize \matCom\|_{\mathrm{min}}} \!\triangleq \!\hoeffCoeffConst,
\end{align*}
the following element-wise inequality holds with probability at least $1-\hoeffProbConst$:
\begin{align}\label{AsquareConcentration_UB}
     \nonumber\matF^2  & \leq (1+\hoeffCoeffConst) \mathbb{E}[{\matF}^2]\\
     \nonumber & \leq (1+\hoeffCoeffConst) (\bar{\matF}^2 - \mathrm{Diag}[{\dimMatF \comFine^\transpose \matCom \diagComSize  \matCom\comFine}] + \dimMatF{\mathrm{Diag}[\comFine^\transpose \matCom \comSizes]}
     \\
     &  = (1\!+\!\hoeffCoeffConst) \dimMatF \left( \comFine^\transpose\matCom \diagComSize  \matCom\comFine \!+\!  \mathrm{Diag}[\comFine^\transpose(\matCom \odot(\mathbf{1}_{\dimMatCom,\dimMatCom}\!-\!\matCom) \comSizes)]\right).
\end{align}
Similarly, for the lower bound we get:
\begin{align}\label{AsquareConcentration_LB}
     \matF^2  \geq (1\!-\!\hoeffCoeffConst) \dimMatF\!\left(\! \comFine^\transpose\matCom \diagComSize  \matCom\comFine \!+ \! \mathrm{Diag}\left[\!\comFine^\transpose\!(\matCom \!\odot\!(\mathbf{1}_{\dimMatCom,\dimMatCom}\!-\!\matCom) \comSizes)\right]\right).
\end{align} 
From \eqref{AsquareConcentration_UB}, \eqref{AsquareConcentration_LB}, and Proposition~\ref{prop:matrixElemInequExtends2norms}, we obtain several matrix norm inequalities which we use later:
\begin{align}\label{specRadiusAsquared_UB}
        \nonumber \specRadius^2(\matF) &  = \specRadius(\matF^2) \\
         \nonumber &  \geq \! (1\shortminus\hoeffCoeffConst) \left(\specRadius(\bar{\matF})^2 \!+\! \dimMatF \specRadius(\mathrm{Diag}[\comFine^\transpose(\matCom \odot(\mathbf{1}_{\dimMatCom,\dimMatCom}\shortminus\matCom) \comSizes)])\right)\\
         & =  (1-\hoeffCoeffConst) \left(\dimMatF^2  \specRadius^2(\matCom \diagComSize) + \dimMatF  \| \matCom \odot(\mathbf{1}_{\dimMatCom,\dimMatCom}-\matCom) \comSizes\|_{\mathrm{max}}\right),
\end{align}
and:
\begin{align}\label{specRadiusAsquared_LB}
        \nonumber\specRadius^2(\matF)  & \leq (1+\!\hoeffCoeffConst) \dimMatF \left(\dimMatF \specRadius^2(\!\matCom \diagComSize)\! + \!\| \matCom \!\odot\!(\mathbf{1}_{\dimMatCom,\dimMatCom}\!-\!\matCom) \comSizes\|_{\mathrm{max}}\!\right)
        \\
        & \leq (1+\hoeffCoeffConst) \dimMatF^2 \scalingSBM^2\left( \specRadius^2(\matCom_{\circ} \diagComSize) + \frac{1}{\dimMatF \scalingSBM}\right),
\end{align}
by employing the triangle inequality on matrix norm.
In addition, \eqref{def_coarseMembershipMat} and the invariance of trace under cyclic permutation yield the following inequality
\begin{align}\label{trAbarSquare}
    \trace\left(\bar{\matF}^2\right) = \dimMatF^2\trace\left((\diagComSize\matCom)^2\right).
\end{align}
Combining \eqref{trAbarSquare} and \eqref{AsquareConcentration_UB}, gives
\begin{align}\label{trASquare}
        \nonumber  \trace\left(\matF^2\right) 
        & \leq  \trace\left[(1\!+\!\hoeffCoeffConst) \dimMatF \left( \comFine^\transpose\matCom \diagComSize  \matCom\comFine \!+\!  \mathrm{Diag}[\comFine^\transpose(\matCom \!\odot\!(\mathbf{1}_{\dimMatCom,\dimMatCom}\!\shortminus\!\matCom) \comSizes)]\right)\right]  \\
          & \leq (1+\hoeffCoeffConst)\dimMatF^2\left[\trace\left((\diagComSize\matCom)^2\right) + \| \matCom\comSizes\|_{\mathrm{max}}\right]
\end{align}
Moreover, Lemma~\ref{lemma:spectralAbarEquivalence} and \eqref{specRadiusAsquared_UB} give
\begin{align}\label{spectRadiusAdivAbar_UB}
    \nonumber \frac{\specRadius(\matF)}{\specRadius(\bar{\matF})} & \leq 
      \frac{\sqrt{(1+\hoeffCoeffConst) \dimMatF^2 \scalingSBM^2\left( \specRadius^2(\matCom_{\circ} \diagComSize) + \frac{1}{\dimMatF \scalingSBM}\right)}}{\dimMatF \scalingSBM \specRadius(\matCom_{\circ} \diagComSize)}  \\
       & \leq  \sqrt{1+\hoeffCoeffConst} \sqrt{1+ \frac{1}{\dimMatF \scalingSBM \specRadius^2(\matCom_{\circ} \diagComSize)}}.
\end{align}
Similarly, by combining Lemma~\ref{lemma:spectralAbarEquivalence} and \eqref{specRadiusAsquared_LB}, we have
\begin{align}\label{spectRadiusAdivAbar_LB}
    \frac{\specRadius(\matF)}{\specRadius(\bar{\matF})} \geq 
       \sqrt{(1-\hoeffCoeffConst) + (1-\hoeffCoeffConst) \frac{ \| \matCom_{\circ} \odot(\mathbf{1}-\matCom) \comSizes\|_{\mathrm{max}}}{\dimMatF \scalingSBM \specRadius^2(\matCom_{\circ} \diagComSize)}}
       \geq  \sqrt{1-\hoeffCoeffConst}.
\end{align}
Combining \eqref{AsquareConcentration_UB}-\eqref{spectRadiusAdivAbar_LB}, gives the following upper bound for the $\trace\left(\left\lvert\matFnom^2-\matFBarNom^2\right\rvert\right)$ term:
\begin{align}\label{diffAsquareNom}
     \nonumber\trace\left(\left\lvert\matFnom^2-\matFBarNom^2\right\rvert\right)
      & \stackrel{(a)}{=} \trace\left(\left\lvert\frac{\matF^2-\bar{\matF}^2}{(\constFine+\specRadius(\bar{\matF}))^2} + (\frac{{\matF}^2}{\constFine+\specRadius({\matF})}-\frac{{\matF}^2}{\constFine+\specRadius(\bar{\matF})})\right\rvert\right) 
     \\   
      \nonumber& \stackrel{(b)}{\leq} \! \trace\left(\left\lvert\frac{\hoeffCoeffConst\bar{\matF}^2+\hoeffCoeffConst \dimMatF \mathrm{Diag}[\comFine^\transpose\matCom \odot(\mathbf{1}-\matCom) \comSizes]}{(\constFine+\specRadius(\bar{\matF}))^2} \!+\! (\frac{{\matF}^2}{(\constFine+\specRadius(\matF))^2}\!-\!\frac{{\matF}^2}{(\constFine+\specRadius(\bar{\matF}))^2})\right\rvert\right) \\
      \nonumber& \stackrel{(c)}{\leq}  \hoeffCoeffConst \frac{\trace\left(\bar{\matF}^2\right)}{(\constFine+\specRadius(\bar{\matF}))^2} + \hoeffCoeffConst \dimMatF\frac{ \trace\left(\mathrm{Diag}[\comFine^\transpose\matCom \odot(\mathbf{1}-\matCom) \comSizes]\right)}{(\constFine+\specRadius(\bar{\matF}))^2} 
      + \left\lvert \frac{1}{(\constFine+\specRadius(\matF))^2}-\frac{1}{(\constFine+\specRadius(\bar{\matF}))^2}\right\rvert \trace\left({\matF}^2\right) \\
       \nonumber& \stackrel{(d)}{\leq}  \hoeffCoeffConst \frac{\dimMatF^2\trace\left((\diagComSize\matCom)^2\right)}{\left(\constFine+\dimMatF \specRadius(\matCom \diagComSize)\right)^2} + \hoeffCoeffConst \frac{\dimMatF^2\|\matCom \odot(\mathbf{1}-\matCom) \comSizes \|_{\mathrm{max}}}{(\constFine+\dimMatF \specRadius(\matCom \diagComSize))^2}
      \\  \nonumber& \qquad  
      + \left\lvert 1 - \frac{\specRadius(\matF)}{\specRadius(\bar{\matF})} \right\rvert \left( 1 + \frac{\specRadius(\matF)+2\constFine}{\specRadius(\bar{\matF})} \right)
      \frac{(1+\hoeffCoeffConst)\dimMatF^2\left[\trace\left((\diagComSize\matCom)^2\right) + \| \matCom\comSizes\|_{\mathrm{max}}\right]}{(1-\hoeffCoeffConst) \left(\dimMatF^2  \specRadius^2(\matCom \diagComSize) + \dimMatF  \| \matCom \odot(\mathbf{1}-\matCom) \comSizes\|_{\mathrm{max}}\right)}\\
       \nonumber& \stackrel{(e)}{\leq}  \hoeffCoeffConst \left[\frac{1/\scalingSBM+\trace\left((\diagComSize\matCom_{\circ})^2\right)}{ \specRadius^2(\matCom_{\circ} \diagComSize)}+\frac{1}{2} (1+4\hoeffCoeffConst)\frac{\left[\trace\left((\diagComSize\matCom_{\circ})^2\right) + \frac{1}{\scalingSBM}\right]}{  \specRadius^2(\matCom_{\circ} \diagComSize) }\right.
      \\  & 
      \qquad\qquad\qquad\qquad\qquad\qquad
      \quad 
      \times\! \!\left.\left( \!1 \!+ \!\sqrt{1\!+\!\hoeffCoeffConst} \sqrt{1+ \frac{ 1}{\dimMatF\scalingSBM  \specRadius^2(\matCom_{\circ} \diagComSize)}} \!+\! \frac{2\constFine}{\dimMatF\scalingSBM  \specRadius(\matCom_{\circ} \diagComSize)} \! \right)\!\right]\!,
\end{align}
where step (a) follows from the normalization factors defined in \eqref{generalLTI} and an additive-subtractive term; (b) is by replacing from \eqref{AsquareConcentration_UB}; 
(c) is due to triangle inequality and the linearity of trace; (d) is because of \eqref{trAbarSquare}, \eqref{trASquare}, \eqref{specRadiusAsquared_LB}, the definition of $\comFine$ in \eqref{def_scaled_comMemberMat}, and the following inequality
\begin{align}
        \nonumber\left\lvert \frac{1}{(\constFine\!+\!\specRadius(\matF))^2}\shortminus\frac{1}{(\constFine\!+\!\specRadius(\bar{\matF}))^2}\right\rvert  &  
        \!=\! \left\lvert \frac{(\constFine\!+\!\specRadius(\bar{\matF}))^2\shortminus(\constFine\!+\!\specRadius(\matF))^2}{(\constFine\!+\!\specRadius(\matF))^2(\constFine\!+\!\specRadius(\bar{\matF}))^2} \right\rvert  
        \\ \nonumber& = \left\lvert \frac{(\specRadius(\bar{\matF})+\specRadius(\matF)+2\constFine)(\specRadius(\bar{\matF})-\specRadius(\matF))}{(\constFine+\specRadius(\matF))^2(\constFine+\specRadius(\bar{\matF}))^2} \right\rvert 
        \\ & \leq \frac{\left[1+\frac{\specRadius(\matF)+2\constFine}{\specRadius(\bar{\matF})}\right]\left\lvert 1-\frac{\specRadius(\matF)}{\specRadius(\bar{\matF})}\right\rvert}{(\constFine+\specRadius(\matF))^2};
\end{align}
finally, (e) follows from \eqref{spectRadiusAdivAbar_LB}, \eqref{spectRadiusAdivAbar_UB}, and two easy-to-validate inequalities $\sqrt{1-\hoeffCoeffConst}\geq 1-\hoeffCoeffConst/2$ and  $\frac{1+\hoeffCoeffConst}{1-\hoeffCoeffConst}\leq 1+4\hoeffCoeffConst, \forall\hoeffCoeffConst<1/2$.
From \eqref{diffAsquareNom}, and using the normalization factor definition when $\matF$ and its expectation are the dynamics of the LTI system in \eqref{fine_LTI}, we can further simplify \eqref{alphanSteps0} as
\begin{align}\label{alphanSteps}
	 \alpha_\dimMatF
	  & \leq \frac{ \hoeffCoeffConst\left[\frac{1/\scalingSBM+{c}_1}{ {c}_\circ^2}+ (\frac{1}{2}+2\hoeffCoeffConst)\frac{\left[{c}_1 + \frac{1}{\scalingSBM}\right]}{{c}_\circ^2}
      \left( 1 \!+ \!\sqrt{1+\hoeffCoeffConst} \sqrt{1+ \frac{ 1/\scalingSBM}{\dimMatF  {c}_\circ^2}} + \frac{2\constFine/\scalingSBM}{\dimMatF  {c}_\circ} \right)\right]}{\left[1-\left(\frac{{c}_\circ}{\constFine/\scalingSBM+{c}_\circ}\right)^2\right] \left[1-\left(\frac{\specRadius(\matF)}{\constFine+\specRadius(\matF)}\right)^2\right]}
      \\
      &  \stackrel{(a)}{=}  \Oscale(\frac{\hoeffCoeffConst}{\scalingSBM}).
\end{align}
In the above ${c}_\circ= \specRadius(\matCom_{\circ} \diagComSize), {c}_1=\trace\left((\diagComSize\matCom_{\circ})^2\right)$, and step (a) holds in the asymptotic scenario when $\diagComSize,\matCom$ are constant and $\sqrt{\dimMatF}\scalingSBM$ is very large and hence, $\hoeffCoeffConst$ becomes very small.


elaborated in the following.

Following similar lines as those in \eqref{alphanSteps0}, the upper bound on $\tilde{\alpha}_\dimMatC$ \eqref{eq: Gramian differences alphatilde} can be simplified to
\begin{align}\label{alphaTildem0}
	 \tilde{\alpha}_\dimMatC 
	  \leq \frac{1}{1-\|\matCnom^2\|_2} \frac{1}{1-\|\matCBarNom^2\|_2}  \trace(\lvert\matCnom^2-\matCBarNom^2\rvert).
\end{align} 
To get \eqref{alphaTildem0}, one should note that \eqref{gramianFine_def} and the Neumann series yield
\begin{align}\label{def:matCGramian}
	\nonumber\Cgram(\matCnom, \mbf{I}_{m\times m}) & \triangleq\lim\limits_{T\to \infty}\Cgram_T(\matCnom, \mbf{I}_{m\times m}) 
	\\
	\nonumber&=
	\lim_{T\to \infty}\sum_{t=0}^{T-1} \matCnom^t\left(\matCnom^\transpose\right)^t \\
	&=\lim\limits_{T\to \infty}\sum_{t=0}^{T-1} \matCnom^{2t}=(\mbf{I}-\matCnom^2)^{-1},
\end{align}
since $\| \matCnom\|_2 = \frac{\specRadius(\matC)}{\constCoarse+\specRadius(\matC)} <1$. 
Next, we simplify the term $\trace(\lvert\matCnom^2-\matCBarNom^2\rvert)$ in \eqref{alphaTildem0} using similar concentration bounds as those in \eqref{diffAsquareNom}.
From Proposition~\ref{prop:hoeffdingIneq}, Lemma~\ref{lemma:spectralAtildebarEquivalence}, and the definition of $\bar{\matC}$ in \eqref{def_coarseMembershipMat}, the following holds for $i,j\in[\dimMatC]$ and $i\neq j$ with probability at least $1-\tilde{\hoeffProbConst}$:
\begin{align}\label{probAcoarse2NonDiag}
        \nonumber\lvert [\matC^2]_{ij} - \mathbb{E}[{\matC}^2]_{ij}\rvert 
        & = \lvert \displaystyle\sum_{\ell \in [\dimMatC]} {\matC_{i\ell}\matC_{j\ell}} - \underbrace{\mathbb{E}[\displaystyle\sum\matC_{i\ell}\matC_{j\ell}]}_{[\bar{\matC}^2]_{ij}}\rvert 
        \\
        \nonumber&  \leq \sqrt{\frac{\dimMatC\log\left(\frac{1}{\tilde{\hoeffProbConst}}\right)}{2}}\frac{\max(\matC_{i\ell}\matC_{j\ell})-\min(\matC_{i\ell}\matC_{j\ell})}{\min(\syncComAs^\transpose_i \matCom \syncComAs^\transpose\syncComAs  \matCom \syncComAs_j)} \mathbb{E}[{\matC}^2]_{ij}
        \\
        & {\leq} \sqrt{\frac{\log\left(\frac{1}{\tilde{\hoeffProbConst}}\right)}{2\dimMatC}}\frac{1}{\left\|\matCom \frac{\syncComAs^\transpose\syncComAs}{\dimMatC} \matCom\right\|_{\mathrm{min}}} \mathbb{E}[{\matC}^2]_{ij}.
\end{align}
Similarly for the diagonal elements, we have:
\begin{align}\label{probAcoarse2Diag1}
       \nonumber \left\lvert [\matC^2]_{ii} -  \mathbb{E}[{\matC}^2_{ii}]\right\rvert 
          &= \left\lvert \displaystyle\sum_{\ell \in [\dimMatC]} {\matC_{i\ell}^2} - \mathbb{E}[{\matC_{i\ell}^2}]\right\rvert 
        \\ 
        & \leq \sqrt{\frac{\dimMatC\log\left(\frac{1}{\tilde{\hoeffProbConst}}\right)}{2}}\frac{\max(\matC_{i\ell}^2)-\min(\matC_{i\ell}^2)}{\min(\mathbb{E}[{\displaystyle\sum_{\ell \in [\dimMatC]}\matC_{i\ell}^2}])} \mathbb{E}[\displaystyle\sum_{\ell \in [\dimMatC]} {\matC_{i\ell}^2}].
\end{align}
Next, we simplify the term $\mathbb{E}[\matC_{i\ell}^2]$ in \eqref{probAcoarse2Diag1}. 
From Section \ref{subsec:SBMmodel}, we can verify that
\begin{align}\label{poissonBinomialDist0}
\protoCovSize^2\matC_{ij} \sim  \text{Poisson-Binomial}(\lbrace\underbrace{\matCom_{kk'}\cdots \matCom_{kk'}}_{\protoCovSize^2\syncComAs_{ik} \syncComAs_{jk'}}: k,k'\in[\dimMatCom]\rbrace),
\end{align}
since $\protoCovSize^2\matC_{ij}$ is a sum of independent (but not necessarily identically) Bernoulli random variables, parameterized by $\matCom_{kk'}$s for $\protoCovSize^2\syncComAs_{ik} \syncComAs_{jk'}$ of the variables contributing to the sum. We can denote \eqref{poissonBinomialDist0} concisely by:
\begin{align}\label{poissonBinomialDist}
     \protoCovSize^2\matC_{ij} \sim \text{Poisson-Binomial}(\lbrace \mathrm{vec}(\matCom) \rbrace^{\odot \mathrm{vec}(\protoCovSize^2\syncComAs_i \syncComAs_j^\transpose)}),
\end{align}
where $\mathrm{vec}(\genMat)$ outputs a vector containing all the elements in its input matrix $\genMat$. We use $\lbrace \mathrm{vec}(\genMat)\rbrace^{\odot \mathrm{vec}(\genMatOth)}$ to denote a vector that contains all the elements of $\genMat$, where $\genMat_{ij}$ has $\genMatOth_{ij}$ occurrences. 
From \eqref{poissonBinomialDist}, and using the variance definition of Poisson-Binomial, the term $\mathbb{E}[\matC_{i\ell}^2]$ in \eqref{probAcoarse2Diag1} can be represented as
\begin{align}\label{expAtildeElemSquared}
        \nonumber\mathbb{E}[\matC_{i\ell}^2] & =  \mathbb{E}[\matC_{i\ell}]^2 + \mathrm{var}(\matC_{i\ell}) \\
         & = (\syncComAs^\transpose_i \matCom \syncComAs_{\ell})^2 
         + \frac{1}{\protoCovSize^4} (\mbf{1}\!-\!\lbrace \mathrm{vec}(\matCom) \rbrace^{\odot \mathrm{vec}(\protoCovSize^2\syncComAs_i \syncComAs_\ell^\transpose)})^\transpose \lbrace \mathrm{vec}(\matCom) \rbrace^{\odot \mathrm{vec}(\protoCovSize^2\syncComAs_i \syncComAs_\ell^\transpose)} \\
         \nonumber& = \!(\syncComAs^\transpose_i \!\matCom\! \syncComAs_{\ell})^2 
        \!+ \!\frac{1}{\protoCovSize^4} 
         \mathrm{vec}(\protoCovSize^2\syncComAs_i \syncComAs_\ell^\transpose)^\transpose [\mathrm{vec}(\matCom)\odot (\mbf{1}_{\dimMatCom^2}\!-\!\mathrm{vec}(\matCom))]\\
         & = \scalingSBM^2 (\syncComAs^\transpose_i \matCom_{\circ} \syncComAs_{\ell})^2 \!+\! \frac{\scalingSBM}{\protoCovSize^2} 
         \mathrm{vec}(\syncComAs_i \syncComAs_\ell^\transpose)^\transpose \mathrm{vec}(\matCom_{\circ} \odot (\mbf{1}_{\dimMatCom,\dimMatCom}\!-\!\matCom)).
\end{align}
We then sum over the $\ell$ index of \eqref{expAtildeElemSquared} to obtain the $\mathbb{E}[{\matC}^2_{ii}]$ term in \eqref{probAcoarse2Diag1}:  
\begin{align}\label{diagEqLB}
         \mathbb{E}[{\matC}^2_{ii}] & = \mathbb{E}[\displaystyle\sum_{\ell \in [\dimMatC]}\matC_{i\ell}^2] \\
         \nonumber&= \!\syncComAs_i^\transpose \! \matCom \! \syncComAs^\transpose\! \syncComAs \! \matCom \! \syncComAs_i \!+\! 
        \frac{1}{\protoCovSize^2} \!
        \mathrm{vec}(\!\syncComAs_i\overbrace{\displaystyle\sum_{\ell \in [\dimMatC]}\syncComAs_\ell^\transpose}^{\mbf{1}_\dimMatC^\transpose\syncComAs}\!)^\transpose \mathrm{vec}(\!\matCom  \odot (\mbf{1}_{\dimMatCom,\dimMatCom}\!-\!\matCom)\!) \\
         \nonumber& = \syncComAs_i^\transpose  \matCom  \syncComAs^\transpose \syncComAs  \!\matCom \! \syncComAs_i \!+\!
        \frac{1}{\protoCovSize^2} 
        \mathrm{vec}(\syncComAs_i  \mbf{1}_\dimMatC^\transpose\syncComAs)^\transpose \mathrm{vec}(\matCom  \odot (\mbf{1}_{\dimMatCom,\dimMatCom}\!-\!\matCom)) \\
         & {\geq} 
         \dimMatC\|\matCom \frac{\syncComAs^\transpose \syncComAs}{\dimMatC} \matCom\|_{\mathrm{min}}.
\end{align}
Similarly an upper bound for \eqref{diagEqLB} is obtained:
\begin{align}\label{diagEqUB}
        \nonumber\mathbb{E}[{\matC}^2_{ii}] & = \mathbb{E}[\displaystyle\sum_{\ell \in [\dimMatC]}\matC_{i\ell}^2]
        \\ \nonumber & 
        = \dimMatC\scalingSBM^2 \syncComAs_i^\transpose  \matCom_{\circ}  \frac{\syncComAs^\transpose \syncComAs}{\dimMatC} \matCom_{\circ} \syncComAs_i + 
        \frac{\scalingSBM}{4\protoCovSize^2} 
        \mathrm{vec}(\syncComAs_i  \mbf{1}^\transpose\syncComAs)^\transpose \mbf{1}\\
        & \leq  \dimMatC\scalingSBM^2 \|\matCom_{\circ}  \frac{\syncComAs^\transpose \syncComAs}{\dimMatC} \matCom_{\circ}\|_{\mathrm{max}} + 
        \frac{\scalingSBM}{4\protoCovSize^2} \underbrace{\sum_{\ell\in[\dimMatC],k,k'\in[\dimMatCom]}\syncComAs_{ik}\syncComAs_{\ell k'}}_{=\dimMatC}
\end{align}
From \eqref{diagEqLB}, and the fact $0\leq \matC_{i,\ell}\leq 1$ by definition in \eqref{linear_coarsening_model}, the upper bound on  \eqref{probAcoarse2Diag1} can be further simplified as
\begin{align}\label{probAcoarse2Diag}
        \lvert [\matC^2]_{ii} - \mathbb{E}[{\matC}^2_{ii}]\rvert 
        \leq \sqrt{\frac{\log\left(\frac{1}{\tilde{\hoeffProbConst}}\right)}{2\dimMatC}}\frac{1}{ \| \matCom  \frac{\syncComAs^\transpose \syncComAs}{\dimMatC} \matCom\|_{\mathrm{min}}}  \mathbb{E}[{\displaystyle\sum_{\ell \in [\dimMatC]}\matC_{i\ell}^2}].
\end{align}
By merging \eqref{probAcoarse2NonDiag} and \eqref{probAcoarse2Diag} and putting $[\matC^2]_{ij}$'s in a matrix form, we get
\begin{align}\label{AtildesquareConcentration_UB}
    \nonumber \matC^2 &  \leq (1+\tilde{\hoeffCoeffConst}) \mathbb{E}[{\matC}^2]\\
     \nonumber& \leq (1+\tilde{\hoeffCoeffConst})  \left(\bar{\matC}^2+  \frac{\dimMatC\scalingSBM}{4\protoCovSize^2} \mbf{I} \right)   \\
     \nonumber&  = (1+\tilde{\hoeffCoeffConst})   \left( (\syncComAs  \matCom \syncComAs^\transpose)^2 +  \frac{\dimMatC\scalingSBM}{4\protoCovSize^2} \mbf{I} \right) \\
     &  = (1+\tilde{\hoeffCoeffConst})  \dimMatC \left(\syncComAs  \matCom \frac{\syncComAs^\transpose \syncComAs}{\dimMatC} \matCom \syncComAs^\transpose +  \frac{\scalingSBM}{4\protoCovSize^2} \mbf{I}\right),
\end{align}
and
\begin{align}\label{AtildesquareConcentration_LB}
     \matC^2  \geq (1-\tilde{\hoeffCoeffConst}) \dimMatC  \syncComAs  \matCom \frac{\syncComAs^\transpose \syncComAs}{\dimMatC} \matCom \syncComAs^\transpose,
\end{align}
which both hold with probability at least $1-\tilde{\hoeffProbConst}$ for $\tilde{\hoeffCoeffConst}$  in \ref{hoeffCoeffConst_def}.
Similar to the $\alpha_\dimMatF$ derivations, further norm inequalities are obtained from \eqref{AtildesquareConcentration_UB} and \eqref{AtildesquareConcentration_LB}, which will be then used to simplify the upper bound  on $\tilde{\alpha}_\dimMatC$ in \eqref{alphaTildem0}. Using Proposition~\ref{prop:matrixElemInequExtends2norms}), $\specRadius(\matC)^2$ can be bounded from below by
\begin{align}\label{spectRadiusAcoarseSquared_UB}
        \nonumber\specRadius(\matC)^2 &  = \specRadius(\matC^2) \\
         & \geq (1-\tilde{\hoeffCoeffConst}) \dimMatC^2 \scalingSBM^2 \specRadius^2\left(\matCom_{\circ} \frac{\syncComAs^\transpose \syncComAs}{\dimMatC}\right),
\end{align}
and bounded from above by
\begin{align}\label{spectRadiusAcoarseSquared_LB}
    \specRadius(\matC)^2 
         \leq (1+\tilde{\hoeffCoeffConst}) \dimMatC[  \dimMatC \scalingSBM^2 \specRadius^2\left(\matCom_{\circ} \frac{\syncComAs^\transpose \syncComAs}{\dimMatC}\right) + \frac{\scalingSBM}{4\protoCovSize^2} ].
\end{align}
Equation \eqref{spectRadiusAcoarseSquared_UB}, coupled with Lemma~\ref{lemma:spectralAtildebarEquivalence}, gives
\begin{align}\label{spectralTildeFractionUB}
        \nonumber\frac{\specRadius(\matC)}{\specRadius(\bar{\matC})} & \leq   \frac{\sqrt{(1+\tilde{\hoeffCoeffConst}) \dimMatC[  \dimMatC \scalingSBM^2 \specRadius^2\left(\matCom_{\circ} \frac{\syncComAs^\transpose \syncComAs}{\dimMatC}\right) + \frac{\scalingSBM}{4\protoCovSize^2} ]}}{ \dimMatC  \specRadius\left(\matCom \frac{\syncComAs^\transpose \syncComAs}{\dimMatC}\right)} \\
        & \leq {\sqrt{1+\tilde{\hoeffCoeffConst}}}\sqrt{1 + \frac{1}{4\protoCovSize^2 \dimMatC \scalingSBM  \specRadius^2\left(\matCom_{\circ} \frac{\syncComAs^\transpose \syncComAs}{\dimMatC}\right)} }.
\end{align}
 Similarly, with the help of \eqref{spectRadiusAcoarseSquared_LB}, we obtain
 \begin{align}
        \nonumber\frac{\specRadius(\matC)}{\specRadius(\bar{\matC})} & \geq \sqrt{1-\tilde{\hoeffCoeffConst}}\sqrt{1 + \frac{1}{4\protoCovSize^2 \dimMatC \scalingSBM  \specRadius^2\left(\matCom_{\circ} \frac{\syncComAs^\transpose \syncComAs}{\dimMatC}\right)} }
        \\
        & \geq \sqrt{1-\tilde{\hoeffCoeffConst}}.
\end{align}
In addition, \eqref{def_coarseMembershipMat} and the invariance of trace under cyclic permutation, yield the following inequality 
\begin{align}\label{trAtildebarSquare}
    \trace\left(\bar{\matC}^2\right) = \dimMatC^2\trace\left((\frac{\syncComAs^\transpose \syncComAs}{\dimMatC}\matCom)^2\right).
\end{align}
Combining \eqref{trAtildebarSquare} and \eqref{AtildesquareConcentration_UB}, gives
\begin{align}\label{trAtildeSquare}
        \nonumber\trace\left(\matC^2\right) 
        & \leq  \trace\left[(1+\tilde{\hoeffCoeffConst})  \dimMatC \left(\syncComAs  \matCom \frac{\syncComAs^\transpose \syncComAs}{\dimMatC} \matCom \syncComAs^\transpose +  \frac{\scalingSBM}{4\protoCovSize^2} \mbf{I}\right)\right]  \\
         & \leq (1+\tilde{\hoeffCoeffConst})\dimMatC^2\left[\trace\left((\frac{\syncComAs^\transpose \syncComAs}{\dimMatC}\matCom)^2\right)+\frac{\scalingSBM}{4\protoCovSize^2}\right].
\end{align}
Similar to \eqref{alphanSteps}, we bound the term $\trace(\lvert\matCnom^2-\matCBarNom^2\rvert)$ to finalize the upper bound simplification in \eqref{alphaTildem0} (inner steps are removed due to redundancy):
\begin{align}\label{ATildenomSquareFrobNormUB}
        \tilde{\alpha}_\dimMatC
        & \leq  \frac{\tilde{\hoeffCoeffConst}\left[\frac{1/\scalingSBM+\tilde{c}_1}{ \tilde{c}_\circ^2}+ (\frac{1}{2}+2\tilde{\hoeffCoeffConst})\frac{\left[\tilde{c}_1 + \frac{1}{\scalingSBM}\right]}{\tilde{c}_\circ^2}
      \left( 1 \!+ \!\sqrt{1+\tilde{\hoeffCoeffConst}} \sqrt{1+ \frac{ 1/\scalingSBM}{\dimMatC  \tilde{c}_\circ^2}} + \frac{2\constFine/\scalingSBM}{\dimMatC  \tilde{c}_\circ} \right)\right]}
      {\left[1-\left(\frac{\tilde{c}_\circ}{\constCoarse/\scalingSBM+\tilde{c}_\circ}\right)^2\right] \left[1-\left(\frac{\specRadius(\matC)}{\constCoarse+\specRadius(\matC)}\right)^2\right]}
      \\
      & = \Oscale(\frac{\tilde{\hoeffCoeffConst}}{\scalingSBM}),
\end{align}
where $\tilde{c}_\circ= \specRadius\left(\matCom_{\circ} \frac{\syncComAs^\transpose \syncComAs}{\dimMatC}\right), \tilde{c}_1=\trace\left((\frac{\syncComAs^\transpose \syncComAs}{\dimMatC}\matCom_{\circ})^2\right)$.
The last scaling function is similarly found as that of \eqref{alphanSteps}, for the asymptotic scenario when $\sqrt{\dimMatC}\scalingSBM$ is very large and hence $\tilde{\hoeffCoeffConst}$ becomes very small.
The proof is now complete. 
\end{proof}

\vspace{-1em}
\subsection{Proof of Lemma \ref{lemma:approximate_gAC}}\label{subsec:proofLemmaApproximategAC}
We aim to find the fine controllability $\overline{\mathcal{S}}_\mathrm{fine}$, i.e. when \eqref{fine_LTI} has expected dynamics. To do so, we write down the definition of  $\bs{\theta}_{\mathrm{fine},\bar{\matF}}$ by substituting $\genMat$  in \eqref{cntrlblty_fine_def} with $\bar{\matF}$ (explanation for each step succeeds the equations): 
\begin{align}\label{thetaFineAbar}
	 \nonumber\bs{\theta}_{\mathrm{fine},\bar{\matF}}   
	 &= \mathrm{diag}\left(\displaystyle\sum_{\tau=0}^{\infty} \matFBarNom^{2\tau} \right)
	 \stackrel{(a)}{=} \mathrm{diag}\left(\displaystyle\sum_{\tau=0}^{\infty} (\frac{\overalConstTemp}{\dimMatF}\comFine^\transpose \matCom_{\circ} \comFine)^{2\tau}\right) 
	\\
	\nonumber& = \mathrm{diag}\left(I + (\frac{\overalConstTemp}{\dimMatF})^2\comFine^\transpose \matCom_{\circ} \comFine\comFine^\transpose \matCom_{\circ} \comFine \right.
	\\ \nonumber& \qquad \qquad
	\left.+ (\frac{\overalConstTemp}{\dimMatF})^4\comFine^\transpose \matCom_{\circ} \comFine \comFine^\transpose \matCom_{\circ} \comFine \comFine^\transpose \matCom_{\circ} \comFine \comFine^\transpose \matCom_{\circ} \comFine  + \cdots\right)\\
	\nonumber& = \mbf{1}_{\dimMatF} + \frac{1}{\dimMatF}  \mathrm{diag}\left(\comFine^\transpose {\diagComSize^{-\frac{1}{2}} \left[(\overalConstTemp\diagComSize^{\frac{1}{2}} \matCom_{\circ}\diagComSize^{\frac{1}{2}})^2 + \cdots\right]} \diagComSize^{-\frac{1}{2}} \comFine \right)\\
	& = \mbf{1}_{\dimMatF} + \frac{1}{\dimMatF}  \mathrm{diag}\left(\comFine^\transpose \inCtrlMat\comFine \right)
	\stackrel{\mathrm{(b)}}{=} \mbf{1}_{\dimMatF} + \frac{1}{\dimMatF}\comFine^\transpose{\mathrm{diag}(\inCtrlMat)},
\end{align}
where as defined in \eqref{upsilon_def}
\begin{align}\label{upsilon_def2}
        \nonumber\inCtrlMat & =   \diagComSize^{-\frac{1}{2}} \left(\left[\mbf{I}-(\overalConstTemp\diagComSize^{\frac{1}{2}} \matCom_{\circ}\diagComSize^{\frac{1}{2}})^2\right]^{-1} -\mbf{I}\right) \diagComSize^{-\frac{1}{2}}
        \\
        & =\overalConstTemp^2{\matCom_{\circ} \diagComSize \matCom_{\circ}} (\mbf{I}  -  ( \overalConstTemp{\matCom_{\circ} \diagComSize \matCom_{\circ}})^2)^{-1},
\end{align}
and $\overalConstTemp\triangleq \frac{1}{\frac{\constFine}{\scalingSBM\dimMatF} +\specRadius(\matCom_{\circ} \diagComSize)}=\Oscale(1) $ since $\frac{\constFine}{\dimMatF\scalingSBM}\xrightarrow{n\rightarrow \infty}0$. Step (a) follows from the definition of $\bar{\matF}$ in \eqref{def_coarseMembershipMat} and Lemma~\ref{lemma:spectralAbarEquivalence}; (b) is due to the special structure of $\comFine$ in \eqref{def_scaled_comMemberMat}, where for an arbitrary matrix $\genMat$ of appropriate size the following equality holds:
\begin{align*}
\nonumber [\mathrm{diag}( \comFine^\transpose\genMat \comFine)]_i &=\! \displaystyle\sum_{k,k'\in[\dimMatCom]} \comFine_{ki} \genMat_{k,k'} \comFine_{k'i} 
\!=\! \displaystyle\sum_{k\in[\dimMatCom]} \comFine_{ki} \genMat_{k,k} \comFine_{ki} \\
& = \displaystyle\sum_{k\in[\dimMatCom]} \comFine_{ki}^2 \genMat_{k,k}  \stackrel{\mathrm{(c)}}{=} \comFine_{i} \mathrm{diag}(\genMat).
\end{align*}
Step (c) is true because $\comFine$ is binary.
Combining Lemma~\ref{lemma:groupTheta} and \eqref{thetaFineAbar} with the definition of $\syncComAs$ in \eqref{def_coarseMembershipMat}, yields:
\begin{align}\label{groupAbar}
	\nonumber\bs{\theta}_{\mathrm{group},\bar{\matF}} & = \frac{1}{\protoCovSize }\matLinCoarse\left(\mbf{1}_{\dimMatF} + \frac{1}{\dimMatF}\comFine^\transpose {\mathrm{diag}(\inCtrlMat)} )\right)
	 \\ & 
	 = \frac{1}{\protoCovSize}\left(\mbf{1}_{\dimMatC} +  \frac{1}{\dimMatF}\syncComAs {\mathrm{diag}(\inCtrlMat)} \right).
\end{align}\qed

\vspace{-1em}
\subsection{Theorem
and Proof: 
$\ell_1$ error bound between the group controllability of the true and expected $\matF$}
Define the error ${\Delta(\matF,\bar{\matF})}$ similar to \eqref{eq: error metric MOR}:
\begin{align}\label{eq: error metric AAbar}
	{\Delta(\matF,\bar{\matF})}\triangleq \left\|\frac{\protoCovSize\bs{\theta}_{\mathrm{group},{\matF}}-\mbf{1}_\dimMatC}{\|\protoCovSize\bs{\theta}_{\mathrm{group},{\matF}}-\mbf{1}_\dimMatC]\|_1}-\frac{\protoCovSize\bs{\theta}_{\mathrm{group},\bar{\matF}}-\mbf{1}_\dimMatC}{\|\protoCovSize\bs{\theta}_{\mathrm{group},\bar{\matF}}-\mbf{1}_\dimMatC\|_1}\right\|_1.
\end{align}
The following theorem bounds the difference between the group controllability of the true and expected $\matF$ in \eqref{eq: error metric AAbar}. As hinted in Section~\ref{sec:est_solution}, this difference, when small, will allow the introduction of the candidate estimator in \eqref{eq: candidate estimator} only as a function of $\syncComAs,\matCom$, and $\diagComSize$, which are retrievable from the coarse-scale matrix $\matC$.
\begin{thm}\label{thm:est_cntrl_error}
(\textbf{$\ell_1$ error bound between} $\bs{\theta}_{\mathrm{group},\mathbf{A}}$ and $\bs{\theta}_{\mathrm{group},\bar{\matF}}$): 
 Let ${\Delta(\matF,{\bar{\matF}})}$ be defined as above for $\matF\in \lbrace0,1\rbrace^{\dimMatF\times\dimMatF}$ and $\matC\in \mathbb{R}^{\dimMatC\times\dimMatC}$, and $0<\probUnifSampl,\hoeffProbConst<1$ be two constants. Then, under Assumption~\ref{assump: graph sparsity},  $\exists \dimMatF_0\in \mathbb{N}$ such that for $\dimMatF>\dimMatF_0$
\begin{align}\label{eq:fullUppBoundAAbar}
\nonumber{\Delta(\matF,{\bar{\matF}})} \!&\!\leq \!
 {2\frac{\left(\frac{\constFine}{\scalingSBM\dimMatF} \!+\! \specRadius(\matCom_{\circ} \diagComSize)\right)^2}{\| {\matCom_{\circ} \diagComSize \matCom_{\circ}}\|_{\mathrm{min}}} \!\left[1\!+\!\sqrt{\frac{\dimMatF}{\dimMatC\protoCovSize}}\sqrt{\frac{\log\left(\frac{1}{\probUnifSampl}\right)}{2}}\right]\!\alpha_\dimMatF}
 \\
 & \!=\!{\Oscale\left(\frac{\hoeffCoeffConst}{\scalingSBM}\!+\!\sqrt{\frac{\dimMatF}{\dimMatC\protoCovSize}}\sqrt{\frac{\log\left(\frac{1}{\probUnifSampl}\right)}{2}}\right)},
\end{align}
holds with probability at least $(1-\probUnifSampl)(1-\hoeffProbConst)$, where $\hoeffCoeffConst = \frac{\sqrt{{\log(\frac{\dimMatF^2}{\hoeffProbConst})}/{(2\dimMatF)}}}{\left\|\matCom \diagComSize \matCom\right\|_{\mathrm{min}}}$; $\dimMatF$ and $\dimMatC$ are the dimensions of $\matF$ and $\matC$; $\matCom$ is defined in Def.~\ref{def:generalSBM}; and $\diagComSize=(1/\dimMatF)\comFine\comFine^\transpose$.	 \qed
\end{thm}

\begin{proof}
We start by providing an upper bound on $\|\bs{\theta}_{\mathrm{fine},\matF}-\bs{\theta}_{\mathrm{fine},\bar{\matF}} \|_{\mathrm{max}}$ which will be shortly used to bound \eqref{eq:fullUppBoundAAbar}. Using the triangle inequality on matrix norm, we have:
\begin{align}\label{maxDiffCtrl_UB1}
        \nonumber \|\bs{\theta}_{\mathrm{fine},\matF}\!-\!\bs{\theta}_{\mathrm{fine},\bar{\matF}} \|_{\mathrm{max}} 
          & \leq \|\bs{\theta}_{\mathrm{fine},\matF} \!-\!\mbf{1}_\dimMatC\|_{\mathrm{max}} \!+\!\|\bs{\theta}_{\mathrm{fine},\bar{\matF}} \!-\!\mbf{1}_\dimMatC\|_{\mathrm{max}}  \\
          \nonumber& \stackrel{(a)}{\leq}  \|(\mbf{I}-\matFnom^2)^{-1}-\mbf{I}\|_{\mathrm{max}} 
          +
          \|\frac{1}{\dimMatF} \comFine^\transpose  {\mathrm{diag}(\inCtrlMat)} \|_{\mathrm{max}}  \\
          \nonumber& \stackrel{(b)}{\leq}  \displaystyle\sum_{\ell=1}^\infty \frac{\|(\matF^2)^{\ell}\|_{\mathrm{max}}}{(\constFine+\specRadius(\matF))^{2\ell}}  +
          \frac{1}{\dimMatF} {\| {\mathrm{diag}(\inCtrlMat)} \|_{\mathrm{max}}} \\
          & \stackrel{(c)}{\leq}  \displaystyle\sum_{\ell=1}^\infty (1+\hoeffCoeffConst)^{\ell} \dimMatF^{\ell} \frac{\maxFunctionInn(\ell)}{\specRadius^{2\ell}(\matF)} 
          + \frac{1}{\dimMatF} {\|\inCtrlMat\|_{\mathrm{max}}},
 \end{align}
where (a) follows from \eqref{theta_fine} and \eqref{thetaFineAbar}; (b) is because $\comFine$ is binary; and finally (c) is due to \eqref{AsquareConcentration_UB} and defining $\maxFunctionInn(\ell)\triangleq\|\left( \comFine^\transpose\matCom \diagComSize  \matCom\comFine +  \mathrm{Diag}[\comFine^\transpose(\matCom \odot(\mathbf{1}_{\dimMatCom,\dimMatCom}\!-\!\matCom) \comSizes)]\right)^{\ell}\|_{\mathrm{max}}$.
We, now, have $\maxFunctionInn(1) = \|\matCom \diagComSize  \matCom \|_{\mathrm{max}} + \|\matCom\comSizes\|_{\mathrm{max}}$, and for all $\ell>1$
\begin{align}\label{maxfunctionInnEll}
         \nonumber \maxFunctionInn(\ell) 
          &\stackrel{(a)}{\leq} \|\comFine^\transpose\matCom \diagComSize  \matCom\comFine +  \mathrm{Diag}[\comFine^\transpose(\matCom \odot(\mathbf{1}_{\dimMatCom,\dimMatCom}\!-\!\matCom) \comSizes)]\|_\infty 
          \|\!\left( \comFine^\transpose\matCom \diagComSize  \matCom\comFine \!+\!  \mathrm{Diag}[\comFine^\transpose(\matCom \odot(\mathbf{1}_{\dimMatCom,\dimMatCom}\!-\!\matCom) \comSizes)]\right)^{\ell\!-\!1}\! \|_{\mathrm{max}} \\
          \nonumber&\stackrel{(b)}{\leq} \!(\| \comFine^\transpose\matCom \diagComSize  \matCom\comFine \|_{\infty}\! +\! \|\matCom \odot(\mathbf{1}_{\dimMatCom,\dimMatCom}\!-\!\matCom) \comSizes)\|_{\mathrm{max}}) \maxFunctionInn(\ell\!-\!1)\\
          &\stackrel{(c)}{\leq} ( {\dimMatF} \|\matCom \diagComSize  \matCom \|_{\mathrm{max}} \!+\! \|\matCom\comSizes\|_{\mathrm{max}}) \maxFunctionInn(\ell-1),
\end{align}
where (a) is due to item~\ref{item:maxInftyNorms} in the matrix norm equivalence list at the beginning of the appendix; (b) follows from the triangle inequality on matrix norms matrix, $\comFine$ being binary, and the definition of $\maxFunctionInn(\ell)$ prior to \eqref{maxfunctionInnEll}; and finally (c) is because of item~\ref{item:InftyNormsqrtMax} of the list at the start of the appendix.
Using \eqref{maxfunctionInnEll}, we upper bound the summation term in \eqref{maxDiffCtrl_UB1}:
\begin{align}\label{summaxfunctionInnEll}
     \nonumber\displaystyle\sum_{\ell=1}^\infty (1+\hoeffCoeffConst)^{\ell} \dimMatF^{\ell} \frac{\maxFunctionInn(\ell)}{\specRadius^{2\ell}(\matF)} 
          &\stackrel{(a)}{\leq} \! \displaystyle\sum_{\ell=1}^\infty (\frac{1\!+\!\hoeffCoeffConst}{1\!-\!\hoeffCoeffConst})^{\ell} \dimMatF^{\ell} \frac{( {\dimMatF} \|\matCom \diagComSize  \matCom \|_{\mathrm{max}} \!+\! \|\matCom\comSizes\|_{\mathrm{max}})^{\ell}}{\!\left(\!\dimMatF^2 \!\specRadius^2\!(\!\matCom \! \diagComSize) \!+\! \dimMatF \!\|\matCom  \odot\!(\!\mathbf{1}_{\dimMatCom,\dimMatCom}\!-\!\matCom) \comSizes\!\|_{\mathrm{max}}\right)^\ell\!}  
          \\
          & =\! \displaystyle\sum_{\ell=1}^\infty \!\left(\! \frac{1\!+\!\hoeffCoeffConst}{1\!-\!\hoeffCoeffConst} \frac{\|\matCom_{\circ} \diagComSize  \matCom_{\circ} \|_{\mathrm{max}} + \frac{1}{{\dimMatF}\scalingSBM} \|\matCom_{\circ}\comSizes\|_{\mathrm{max}}}{ \specRadius^2(\matCom_{\circ} \diagComSize) \!+\! \frac{1}{\!\dimMatF \!\scalingSBM\!} \!\|\!\matCom_{\circ} \!\odot\!(\!\mathbf{1}_{\dimMatCom,\dimMatCom}\!-\!\matCom\!) \comSizes\|_{\mathrm{max}}}\!\right)^\ell,
 \end{align}
Using \eqref{summaxfunctionInnEll}, we can further simplify \eqref{maxDiffCtrl_UB1} as
\begin{align}\label{maxDiffCtrl_UB}
         \nonumber\|\bs{\theta}_{\mathrm{fine},\matF}-\bs{\theta}_{\mathrm{fine},\bar{\matF}} \|_{\mathrm{max}}  
        & \leq
          \frac{\frac{1+\hoeffCoeffConst}{1-\hoeffCoeffConst} \frac{\|\matCom_{\circ} \diagComSize  \matCom_{\circ} \|_{\mathrm{max}} + \frac{1}{{\dimMatF}\scalingSBM} \|\matCom_{\circ}\comSizes\|_{\mathrm{max}}}{ \specRadius^2(\matCom_{\circ} \diagComSize) + \frac{1}{\dimMatF \scalingSBM} \|\matCom_{\circ} \odot(\mathbf{1}_{\dimMatCom,\dimMatCom}\!-\!\matCom) \comSizes\|_{\mathrm{max}}}}{1-\frac{1+\hoeffCoeffConst}{1-\hoeffCoeffConst} \frac{\|\matCom_{\circ} \diagComSize  \matCom_{\circ} \|_{\mathrm{max}} + \frac{1}{{\dimMatF}\scalingSBM} \|\matCom_{\circ}\comSizes\|_{\mathrm{max}}}{ \specRadius^2(\matCom_{\circ} \diagComSize) + \frac{1}{\dimMatF \scalingSBM} \|\matCom_{\circ} \odot(\mathbf{1}_{\dimMatCom,\dimMatCom}\!-\!\matCom) \comSizes\|_{\mathrm{max}}}}
        +
          \frac{1}{\dimMatF} {\|\inCtrlMat\|_{\mathrm{max}}}\\
          & \stackrel{(a)}{=} \Oscale(1).
 \end{align}
(a) in \eqref{maxDiffCtrl_UB} is because $\hoeffCoeffConst$ is very small when $\dimMatF$ is very large, and $\| \inCtrlMat\|_{\mathrm{max}}$ is a constant (c.f. \eqref{upsilon_def2}):
\begin{align}
\nonumber\left\| \inCtrlMat\right\|_{\mathrm{max}} & =  
\left\|\diagComSize^{-\frac{1}{2}} \left(\left[\mbf{I}-(\overalConstTemp\diagComSize^{\frac{1}{2}} \matCom_{\circ}\diagComSize^{\frac{1}{2}})^2\right]^{-1} -\mbf{I}\right) \diagComSize^{-\frac{1}{2}} \right\|_{\mathrm{max}}
\\&
 = \Oscale(1).
\end{align}
We, now, return to bounding $\Delta(\matF,{\bar{\matF}})$. Using \eqref{eq: error metric AAbar} and Lemma~\ref{lemma:groupTheta}, we obtain:
\begin{align}\label{relDiffAAbar}
	\nonumber{\Delta(\matF,{\bar{\matF}})} 
	& = \displaystyle\sum_{i=1}^\dimMatC\vert \frac{\protoCovSize\bs{\theta}_{\mathrm{group},{\matF}}^{(i)}-1}{\|\protoCovSize\bs{\theta}_{\mathrm{group},{\matF}}-\mbf{1}_\dimMatC\|_1}-\frac{\protoCovSize\bs{\theta}_{\mathrm{group},\bar{\matF}}^{(i)}-1}{\|\protoCovSize\bs{\theta}_{\mathrm{group},\bar{\matF}}-\mbf{1}_\dimMatC\|_1} \vert  
	\\
	 \nonumber& = \displaystyle\sum_{i=1}^\dimMatC\left\lvert \frac{\vecLinCoarse_i \bs{\theta}_{\mathrm{fine},\matF}-1}{\|\matLinCoarse \bs{\theta}_{\mathrm{fine},\matF}-\mbf{1}_\dimMatC\|}-\frac{\vecLinCoarse_i \bs{\theta}_{\mathrm{fine},\bar{\matF}}-1}{\|\matLinCoarse \bs{\theta}_{\mathrm{fine},\bar{\matF}}-\mbf{1}_\dimMatC\|_1}\right\rvert 
    \\
	\nonumber& \stackrel{(a)}{\leq}  2 \frac{\|\matLinCoarse \bs{\theta}_{\mathrm{fine},\matF}-\matLinCoarse\bs{\theta}_{\mathrm{fine},\bar{\matF}} \|_1}{\max\left({\|\matLinCoarse \bs{\theta}_{\mathrm{fine},\matF}-1\|_1, \|\matLinCoarse\bs{\theta}_{\mathrm{fine},\bar{\matF}}-1\|_1}\right)}
	\\
	\nonumber& \stackrel{(b)}{\leq}  2 \frac{\frac{1}{\protoCovSize}\| \bs{\theta}_{\mathrm{fine},\matF}-\bs{\theta}_{\mathrm{fine},\bar{\matF}} \|_{1, (\protoCovSize\dimMatC)}}{\max\left({\|\matLinCoarse \bs{\theta}_{\mathrm{fine},\matF}-1\|_{1}, \|\matLinCoarse\bs{\theta}_{\mathrm{fine},\bar{\matF}}-1\|_1}\right)}
	\\ 
	\nonumber& \stackrel{(c)}{\leq}  2 \frac{\frac{1}{\protoCovSize} [\frac{\dimMatC\protoCovSize}{\dimMatF}\| \bs{\theta}_{\mathrm{fine},\matF}-\bs{\theta}_{\mathrm{fine},\bar{\matF}} \|_1 \!+\! \sqrt{\frac{\dimMatC\protoCovSize\log\left(\frac{1}{\probUnifSampl}\right)}{2}}{\| \bs{\theta}_{\mathrm{fine},\matF}\!-\!\bs{\theta}_{\mathrm{fine},\bar{\matF}} \|_{\mathrm{max}}}]}{\max\left({\|\matLinCoarse \bs{\theta}_{\mathrm{fine},\matF}\!-\!1\|_1, \|\matLinCoarse\bs{\theta}_{\mathrm{fine},\bar{\matF}}-1\|_1}\right)}
	\\
	\nonumber& =  2 \frac{ [\frac{\dimMatC}{\dimMatF}\alpha_\dimMatF+\sqrt{\frac{\dimMatC\log\left(\frac{1}{\probUnifSampl}\right)}{2\protoCovSize}} {\| \bs{\theta}_{\mathrm{fine},\matF}-\bs{\theta}_{\mathrm{fine},\bar{\matF}} \|_{\mathrm{max}}}] }{\max\left({\|\matLinCoarse \bs{\theta}_{\mathrm{fine},\matF}-1\|_1, \|\matLinCoarse\bs{\theta}_{\mathrm{fine},\bar{\matF}}-1\|_1}\right)}
	\\
	\nonumber& \stackrel{(d)}{\leq}  
	 2 \frac{ [\frac{\dimMatC}{\dimMatF}\alpha_\dimMatF+\sqrt{\frac{\dimMatC\log\left(\frac{1}{\probUnifSampl}\right)}{2\protoCovSize}} {\| \bs{\theta}_{\mathrm{fine},\matF}-\bs{\theta}_{\mathrm{fine},\bar{\matF}} \|_{\mathrm{max}}}] }{\frac{\dimMatC}{\dimMatF} \frac{1}{\frac{\constFine}{\scalingSBM\dimMatF} + \specRadius(\matCom_{\circ} \diagComSize)} {\|\mathrm{diag}(\inCtrlMat)\|_{\textrm{min}}}}
	\\
	\nonumber&  \stackrel{(e)}{\leq}   
	 2 \frac{(\frac{\constFine}{\scalingSBM\dimMatF} + \specRadius(\matCom_{\circ} \diagComSize))^2}{\| {\matCom_{\circ} \diagComSize \matCom_{\circ}}\|_{\mathrm{min}}} 
	 [\alpha_\dimMatF+\frac{1+\hoeffCoeffConst}{1-\hoeffCoeffConst}\sqrt{\frac{\dimMatF}{\dimMatC\protoCovSize}}\sqrt{\frac{\log\left(\frac{1}{\probUnifSampl}\right)}{2}}{\| \bs{\theta}_{\mathrm{fine},\matF}-\bs{\theta}_{\mathrm{fine},\bar{\matF}} \|_{\mathrm{max}}}] 
	\\
    & \stackrel{(f)}{=}  \Oscale\left(\frac{\hoeffCoeffConst}{\scalingSBM}+\sqrt{\frac{\dimMatF}{\dimMatC\protoCovSize}}\sqrt{\frac{\log\left(\frac{1}{\probUnifSampl}\right)}{2}}\right)
\end{align} 
that holds with probability at least $(1-\hoeffProbConst)(1-\probUnifSampl) \geq  1-\probUnifSampl-\hoeffProbConst$ (c.f. probability of joint events), for sufficiently large $\dimMatF$, i.e. $\exists\dimMatF_0: \dimMatF\geq \dimMatF_0$.
 (a) follows from Proposition~\ref{prop:uppBound_relDiff} where $\genVec$ is replaced with $\matLinCoarse \bs{\theta}_{\mathrm{fine},\matF}$ and $\genVecOth$ with $\matLinCoarse\bs{\theta}_{\mathrm{fine},\bar{\matF}}$; (c) uses \eqref{maxDiffCtrl_UB}; (d) is from \eqref{groupAbar} which proves:
\begin{align}\label{LB_thetaGroupMinus1}
        \nonumber\| \protoCovSize\bs{\theta}_{\mathrm{group},\bar{\matF}}^{(i)}-1 \|_1 & = \left\|\frac{1}{\dimMatF}  \syncComAs {\mathrm{diag}(\inCtrlMat)} \right\|_1
        \\
         \nonumber & \geq  \frac{1}{\dimMatF}  \left\| {\syncComAs \mbf{1}_\dimMatCom} \right\|_1 {\|\inCtrlMat\|_{\mathrm{min}}}
         \geq  \frac{\dimMatC}{\dimMatF}  \|\inCtrlMat\|_{\mathrm{min}} 
        \\ &
         \stackrel{(g)}{=} \Omega\left( \frac{\dimMatC}{\dimMatF} \right),
\end{align} 
where (g) in \eqref{LB_thetaGroupMinus1} follows from 
\begin{align}\label{upsilonnmin}
\nonumber\| \inCtrlMat\|_{\mathrm{min}} & =  \left\| \diagComSize^{-\frac{1}{2}} \left(\left[\mbf{I}\!-\!(\overalConstTemp\diagComSize^{\frac{1}{2}} \matCom_{\circ}\diagComSize^{\frac{1}{2}})^2\right]^{-1} \!-\!\mbf{I}\right) \diagComSize^{-\frac{1}{2}}\right\|_{\mathrm{min}} 
\\&
 = \Omega(1).
\end{align} 
(e) uses the definition of $\inCtrlMat$ in \eqref{upsilon_def}; (f) due to the definition of $\alpha_\dimMatF$ in \eqref{eq: Gramian differences alpha} and an easy-to-verify inequality $\frac{1+\hoeffCoeffConst}{1-\hoeffCoeffConst}\leq 1+4\hoeffCoeffConst$ for $\hoeffCoeffConst<1/2$ that is satisfied for sufficiently large $\dimMatF$; and finally, (b) is the result of the properties of the coarse-measurement matrix in \eqref{linear_coarsening_model}, the following notation definition 
 \begin{align*}
 \left\|\bs{\theta}_{\mathrm{fine},\matF}-\bs{\theta}_{\mathrm{fine},\bar{\matF}} \right\|_{1, (\protoCovSize\dimMatC)} \triangleq \displaystyle\sum_{i \in \displaystyle\cup_{j\in[\dimMatC]}\mathcal{K}_j} \left\lvert\bs{\theta}^{(i)}_{\mathrm{fine},\matF}-\bs{\theta}^{(i)}_{\mathrm{fine},\bar{\matF}}  \right\rvert  
 \end{align*}
for $\mathcal{K}_i=\mathrm{supp}(\vecLinCoarse_i)$ defined in Section~\ref{subsec:probstatement}, and invoking \eqref{subsetRandomSampleSum} which is explained below.

\noindent \textbf{Sum of Random Samples from a set of Data Points:} Using a Hoeffding's bound equivalence for the \textit{sum of random samples of a set of data points} $\lbrace\genVec_i\rbrace_{i\in \indexRandomSet}$ indexed by $\indexRandomSet$, we have \cite[Proposition 1.2]{bardenet2015concentration}
\begin{align}
      \nonumber \mathbb{P}\lbrace \lvert \frac{1}{ \lvert \indexRandomSet \rvert} \displaystyle\sum_{i\in \indexRandomSet} \lvert\genVec_i\rvert- \frac{1}{\lvert \genVec \rvert } \displaystyle\sum_{i\in [\lvert \genVec \rvert]} \lvert\genVec_i\rvert \rvert \leq \constUnifSampl\rbrace 
        & \geq 1 - \exp\left(-\frac{2\lvert \indexRandomSet \rvert \constUnifSampl^2}{(\max(\lvert\genVec\rvert)-\min(\lvert \genVec\rvert))^2}\right) \\
        & \geq 1 - \exp\left(-\frac{2\lvert\indexRandomSet\rvert\constUnifSampl^2 }{\max^2\lvert\genVec\rvert}\right) 
\end{align}
which means with probability at least $1-\probUnifSampl$
\begin{align}\label{subsetRandomSampleSum}
         \nonumber \|\genVec \|_{1,(\protoCovSize\dimMatC)} &\leq \frac{ \lvert \indexRandomSet \rvert}{ \lvert \genVec \rvert}  \|\genVec \|_1 + \sqrt{\frac{\lvert\indexRandomSet\rvert\log\left(\frac{1}{\probUnifSampl}\right)}{2}}{\max\lvert\genVec\rvert}   \\
        & \leq [\frac{ \lvert \indexRandomSet \rvert}{ \lvert \genVec \rvert}  + \sqrt{\frac{\lvert\indexRandomSet\rvert\log\left(\frac{1}{\probUnifSampl}\right)}{2}} \frac{\|\genVec \|_{\mathrm{max}}}{\|\genVec \|_1}]  \|\genVec \|_1.
\end{align}
To simplify the notations, we set $\probUnifSampl=\hoeffProbConst$. This concludes the proof.
\end{proof}

\vspace{-1em}
\subsection{Proof of Theorem \ref{thm: group-coarse-error2}}
The proof of the theorem makes use of Lemma \ref{lma: Gramian differences}. Let $i \in [\dimMatC]$, and recall that $\mbf{B}_i=\mathrm{diag}(\vecLinCoarse_i^\transpose)$ and $\wt{\mbf{B}}_i=\protoCovSize\matLinCoarse\mathrm{diag}(\vecLinCoarse_i^\transpose)$. Consider the following bound using the triangle inequality for matrix norms:
\begin{align}\label{eq: Theorem 1 proof delta bound}
\nonumber\Delta(\matF,\matC)  & \triangleq \left\|\frac{\protoCovSize\bs{\theta}_{\mathrm{group},{\matF}}\!-\!\mbf{1}_{\dimMatC}}{\|\protoCovSize\bs{\theta}_{\mathrm{group},{\matF}}\!-\!\mbf{1}_{\dimMatC}\|_1}\!-\!\frac{\bs{\theta}_{\mathrm{coarse},\matC}\!-\!\mbf{1}_{\dimMatC}}{\|\bs{\theta}_{\mathrm{coarse},\matC}-\mbf{1}_{\dimMatC}\|_1}\right\|_1 \\
\nonumber&  \leq \underbrace{\left\|\frac{\protoCovSize\bs{\theta}_{\mathrm{group},{\matF}}\!-\!\mbf{1}_{\dimMatC}}{\|\protoCovSize\bs{\theta}_{\mathrm{group},{\matF}}-\mbf{1}_{\dimMatC}\|_1}\!-\!\frac{\protoCovSize\bs{\theta}_{\mathrm{group},\bar{\matF}}\!-\!\mbf{1}_{\dimMatC}}{\|\protoCovSize\bs{\theta}_{\mathrm{group},\bar{\matF}}-\mbf{1}_{\dimMatC}\|_1}\right\|_1}_{=\Delta(\matF,\bar{\matF})}
\\
\nonumber & \quad + \underbrace{\left\| \frac{\bs{\theta}_{\mathrm{coarse},\matC}-\mbf{1}_{\dimMatC}}{\|\bs{\theta}_{\mathrm{coarse},\matC}\!-\!\mbf{1}_{\dimMatC}\|_1}\! -\! \frac{\bs{\theta}_{\mathrm{coarse},\bar{\matC}}-\mbf{1}_{\dimMatC}}{\|\bs{\theta}_{\mathrm{coarse},\bar{\matC}}\!-\!\mbf{1}_{\dimMatC}\|_1} \right\|_1}_{=\Delta(\matC,\bar{\matC})} \\
 & \qquad 
+  \underbrace{\left\| \frac{\protoCovSize\bs{\theta}_{\mathrm{group},{\bar{\matF}}}-\mbf{1}_{\dimMatC}}{\|\protoCovSize\bs{\theta}_{\mathrm{group},\bar{\matF}}\!-\!\mbf{1}_{\dimMatC}\|_1}\!-\!\frac{\bs{\theta}_{\mathrm{coarse},\bar{\matC}}-\mbf{1}_{\dimMatC}}{\|\bs{\theta}_{\mathrm{coarse},\bar{\matC}}\!-\!\mbf{1}_{\dimMatC}\|_1} \right\|_1}_{=\Delta({\bar{\matF}},\bar{\matC})}
 \!.
\end{align}
where $\bs{\theta}_{\mathrm{group},\matF}$ and $\bs{\theta}_{\mathrm{coarse},\matC}$ are defined respectively in \eqref{eq: group avg-controllability} and \eqref{eq: coarse avg-controllability}. Similarly, $\bs{\theta}_{\mathrm{group},\bar{\matF}}$ and $\bs{\theta}_{\mathrm{coarse},\bar{\matC}}$ are defined by replacing $\matF$ with $\bar{\matF}$ in \eqref{eq: group avg-controllability}, and $\matC$ with $\bar{\matC}$ in \eqref{eq: coarse avg-controllability}.
An upper bound on $\Delta(\matF,\bar{\matF})$ has been provided in Theorem~\ref{thm:est_cntrl_error}.
We now derive an upper bound on $\Delta(\matC,\bar{\matC})$ with the help of Proposition~\ref{prop:uppBound_relDiff}. 
Replacing $\genMat$ with $\bar{\matC}$ in \eqref{theta_fine} yields $
\bs{\theta}_{\mathrm{coarse},\bar{\matC}} = \mathrm{diag}\left((\mbf{I} - \matCBarNom^2)^{-1} \right)
$. Hence
\begin{align}\label{coarseThetaAtildeBarSum}
\left\| \bs{\theta}_{\mathrm{coarse},\bar{\matC}}-\mbf{1} \right\|_1
\nonumber& = \trace\left[(\mbf{I} - \matCBarNom^2)^{-1}-\mbf{I}\right] \\
\nonumber& = \trace\left[\matCBarNom^2(\mbf{I} - \matCBarNom^2)^{-1}\right] \\
\nonumber& \geq \trace\left[\matCBarNom^2\right] 
 \geq \frac{\trace\left[\left(\syncComAs \matCom \syncComAs^\transpose\right)^2\right]}{(\constCoarse+\specRadius(\syncComAs \matCom \syncComAs^\transpose))^2}\\
& \geq \frac{ \trace\left[(\frac{\syncComAs^\transpose \syncComAs}{\dimMatC}\matCom_{\circ})^2\right] }{\left(\frac{\constCoarse}{\dimMatC\scalingSBM}+   \specRadius\left(\matCom_{\circ} \frac{\syncComAs^\transpose \syncComAs}{\dimMatC}\right)\right)^2},
\end{align}
for the proof of which \eqref{def_coarseMembershipMat} and the normalization factor in \eqref{eq: coarse system} are used.
Replacing $\genVec$ and $\genVecOth$ in Proposition~\ref{prop:uppBound_relDiff} with $\bs{\theta}_{\mathrm{coarse},\matC}$ and  $\bs{\theta}_{\mathrm{coarse},\bar{\matC}}$ yields
\begin{align}\label{tautildefinal}
\nonumber\Delta(\matC,\bar{\matC}) & = \!  
\left\| \frac{\bs{\theta}_{\mathrm{coarse},{\matC}}\!-\!\mbf{1}_{\dimMatC}}{\|\bs{\theta}_{\mathrm{coarse},\matC}\!-\!\mbf{1}_{\dimMatC}\|_1}\!-\!\frac{\bs{\theta}_{\mathrm{coarse},\bar{\matC}}\!-\!\mbf{1}_{\dimMatC}}{\|\bs{\theta}_{\mathrm{coarse},\bar{\matC}}\!-\!\mbf{1}_{\dimMatC}\|_1} \right\|_1  
\\
\nonumber& \leq  2 \frac{\|\bs{\theta}_{\mathrm{coarse},{\matC}}-\bs{\theta}_{\mathrm{coarse},\bar{\matC}} \|_1}{\max\left({\|\bs{\theta}_{\mathrm{coarse},{\matC}}-\mbf{1}_{\dimMatC}\|_1, \|\bs{\theta}_{\mathrm{coarse},\bar{\matC}}-\mbf{1}_{\dimMatC}\|_1}\right)}
\\
\nonumber& \stackrel{(a)}{\leq} 2\frac{\tilde{\alpha}_\dimMatC}{\| \bs{\theta}_{\mathrm{coarse},\bar{\matC}}-\mbf{1}_\dimMatC \|_1}
\\
& \stackrel{(b)}{\leq} \!\underbrace{2 \frac{\!\left(\!\frac{\constCoarse}{\dimMatC\scalingSBM}\!+\!\specRadius\!\left(\! \matCom_{\circ} \frac{\syncComAs^\transpose \syncComAs}{\dimMatC}\!\right)\!\right)\!^2}{  \trace((\frac{\syncComAs^\transpose \syncComAs}{\dimMatC}\matCom_{\circ})^2) }}_{\Oscale(1)} \tilde{\alpha}_\dimMatC
 \!\stackrel{(c)}{=}\! \Oscale\!\left(\!\frac{\tilde{\hoeffCoeffConst}}{\scalingSBM}\!\right),
\end{align}
with probability at least $1-\tilde{\hoeffProbConst}$. Steps (a) and (b) are by substitution respectively from \eqref{eq: Gramian differences alphatilde} and \eqref{coarseThetaAtildeBarSum}; and (c) follows from \eqref{eq: Gramian differences alpha} and the assumption that $\dimMatF,\dimMatC$ are sufficiently large.

Similar to \eqref{eq:fullUppBoundAAbar} and \eqref{tautildefinal}, we  find an upper bound for the remaining error term $\Delta({\bar{\matF}},\bar{\matC})$ in \eqref{eq: Theorem 1 proof delta bound}.
We first derive a closed-form formula for $\bs{\theta}_{\mathrm{coarse},\bar{\matC}}$ similar to that of $\bs{\theta}_{\mathrm{group},{\bar{\matF}}}$ obtained in Lemma~\ref{lemma:groupTheta} (middle steps are removed due to redundancy): 
\begin{align}\label{thetaCoarseAtildeBar}
 \nonumber\bs{\theta}_{\mathrm{coarse},\bar{\matC}}  
	 & \!=\! \mathrm{diag}\left(\displaystyle\sum_{\tau=0}^{\infty} \matCBarNom^{2\tau} \right)
	 \!=\! \mathrm{diag}\left(\displaystyle\sum_{\tau=0}^{\infty} (\frac{\overalConstTempCoarse}{\dimMatC} \syncComAs \matCom_{\circ} \syncComAs^\transpose)^{2\tau}\right) \\
	& \!=\! \mbf{1}_{\dimMatC} + \frac{1}{\dimMatC}\mathrm{diag}\left(\syncComAs \tilde{\inCtrlMat}\syncComAs^\transpose\right),
\end{align}
where 
\begin{align*}
         \tilde{\inCtrlMat} \!\triangleq \!\overalConstTempCoarse^2 \matCom_{\circ} \frac{\syncComAs^\transpose \syncComAs}{\dimMatC} \matCom_{\circ} ( \mbf{I} \! - \! ( \overalConstTempCoarse \frac{\syncComAs^\transpose \syncComAs}{\dimMatC} \matCom_{\circ} )^2 )^{-1},
\end{align*}
and $\overalConstTempCoarse\triangleq 1/\!\left(\frac{\constCoarse}{\scalingSBM\dimMatC} \!+\!\specRadius\!\left(\! \matCom \frac{\syncComAs^\transpose \syncComAs}{\dimMatC}\!\right)\!\right)$.
By substituting \eqref{groupAbar} and \eqref{thetaCoarseAtildeBar} into $\Delta({\bar{\matF}},\bar{\matC})$ in \eqref{eq: Theorem 1 proof delta bound}, we get
\begin{align}\label{biasUB}
        \nonumber\Delta(\bar{\matF},\bar{\matC}) 
        & =\left\|\frac{\protoCovSize\bs{\theta}_{\mathrm{group},{\bar{\matF}}}-\mbf{1}}{\|\protoCovSize\bs{\theta}_{\mathrm{group},{\bar{\matF}}}-\mbf{1}_\dimMatC\|_1}-\frac{\bs{\theta}_{\mathrm{coarse},\bar{\matC}}-\mbf{1}}{\|\bs{\theta}_{\mathrm{coarse},\bar{\matC}}-\mbf{1}_\dimMatC\|_1} \right\|_1 \\
        \nonumber& = \left\|\frac{\frac{1}{\dimMatF}\syncComAs {\mathrm{diag}(\inCtrlMat)}}{\|\frac{1}{\dimMatF}\syncComAs {\mathrm{diag}(\inCtrlMat)}\|_1}-\frac{\frac{1}{\dimMatC} \mathrm{diag}(\syncComAs \tilde{\inCtrlMat}\syncComAs^\transpose)}{\|\frac{1}{\dimMatC} \mathrm{diag}(\syncComAs \tilde{\inCtrlMat}\syncComAs^\transpose)\|_1} \right\|_1  \\
         \nonumber& \stackrel{(a)}{\leq} 2 \frac{\|\syncComAs {\mathrm{diag}(\inCtrlMat)}-\mathrm{diag}(\syncComAs \tilde{\inCtrlMat}\syncComAs^\transpose)\|_1}{\max\left({\|\syncComAs {\mathrm{diag}(\inCtrlMat)}\|_1, \|\mathrm{diag}(\syncComAs \tilde{\inCtrlMat}\syncComAs^\transpose)\|_1}\right)} \\
         \nonumber& \leq 2 \frac{\|\syncComAs {\mathrm{diag}(\inCtrlMat)}-\mathrm{diag}(\syncComAs (\tilde{\inCtrlMat}-\inCtrlMat+\inCtrlMat)\syncComAs^\transpose)\|_1}{\dimMatC\|\inCtrlMat\|_{\mathrm{min}}}\\
         \nonumber& \leq 2 \frac{\|\syncComAs {\mathrm{diag}(\inCtrlMat)}-\mathrm{diag}(\syncComAs\inCtrlMat\syncComAs^\transpose)\|_1+\dimMatC\|\tilde{\inCtrlMat}-\inCtrlMat\|_{\mathrm{max}}}{\dimMatC\|\inCtrlMat\|_{\mathrm{min}}}
         \\
         & \stackrel{(b)}{=} \!\Oscale\!\left(\!\frac{\|\syncComAs {\mathrm{diag}(\inCtrlMat)}\!-\!\mathrm{diag}(\syncComAs\inCtrlMat\syncComAs^\transpose)\|_1}{\dimMatC}\!+\!\|\diagComSize\!-\!\frac{\syncComAs^\transpose\syncComAs}{\dimMatC}\|_{\mathrm{max}}\!\right)\!,
\end{align}
where step (a) follows from similar logic as in Proposition~\ref{prop:uppBound_relDiff}; (b) is proved similar to a later elaborate derivation in Section~\ref{subsec:proof_thm_full_est_cntrl_error} which results in \eqref{Ebarmax}, except that $\hat{\matCom}$ in \eqref{Ebarmax} is replaced with $\matCom$ (i.e. $\errorMat_{\matCom}=0$) and $\hat{\diagComSize}$ in \eqref{Ebarmax} is substituted by $\frac{\syncComAs^\transpose\syncComAs}{\dimMatC}$ (i.e. $\errorMat_{\diagComSize}=\diagComSize\!-\!\frac{\syncComAs^\transpose\syncComAs}{\dimMatC}$).
The statement of the theorem follows with probability at least $(1-\probUnifSampl)(1-\hoeffProbConst)(1-\tilde{\hoeffProbConst})$ (the probability of joint events) that satisfies
\begin{align}
    (1-\probUnifSampl)(1-\hoeffProbConst)(1-\tilde{\hoeffProbConst})\geq 1-\probUnifSampl-\hoeffProbConst-\tilde{\hoeffProbConst},
\end{align}
and by substituting \eqref{eq:fullUppBoundAAbar}, \eqref{tautildefinal}, and \eqref{biasUB} into \eqref{eq: Theorem 1 proof delta bound}. To simplify the notations, we set $\probUnifSampl=\tilde{\hoeffProbConst}=\hoeffProbConst$.
The proof is now complete.\qed


\vspace{-1em}
\subsection{Proof of Theorem \ref{thm:full_est_cntrl_error}}\label{subsec:proof_thm_full_est_cntrl_error}
We start by using the triangle inequality for matrix norms in order to find an upper bound for $\widehat{\Delta}(\matC)$:
\begin{align}\label{full_est_Err}
	\nonumber\widehat{\Delta}(\matC) &= \left\|\frac{\hat{\bs{\theta}}_{\mathrm{group}}-\mbf{1}_{\dimMatC}}{\|\hat{\bs{\theta}}_{\mathrm{group}}-\mbf{1}_{\dimMatC}\|_1}-\frac{\protoCovSize\bs{\theta}_{\mathrm{group},{\matF}}-\mbf{1}_{\dimMatC}}{\|\protoCovSize\bs{\theta}_{\mathrm{group},{\matF}}-\mbf{1}_{\dimMatC}\|_1}\right\|_1 \\
	\nonumber& \leq \underbrace{\left\|\frac{\protoCovSize\bs{\theta}_{\mathrm{group},\bar{\matF}}-\mbf{1}_{\dimMatC}}{\|\protoCovSize\bs{\theta}_{\mathrm{group},\bar{\matF}}-\mbf{1}_{\dimMatC}\|_1}-\frac{\protoCovSize\bs{\theta}_{\mathrm{group},{\matF}}-\mbf{1}_{\dimMatC}}{\|\protoCovSize\bs{\theta}_{\mathrm{group},{\matF}}-\mbf{1}_{\dimMatC}\|_1}\right\|_1}_{={\Delta(\matF,{\bar{\matF}})}} 
	\\ & \qquad  
	+ \underbrace{\left\|\frac{\hat{\bs{\theta}}_{\mathrm{group}}-\mbf{1}_{\dimMatC}}{\|\hat{\bs{\theta}}_{\mathrm{group}}-\mbf{1}_{\dimMatC}\|_1}-\frac{\protoCovSize\bs{\theta}_{\mathrm{group},\bar{\matF}}-\mbf{1}_{\dimMatC}}{\|\protoCovSize\bs{\theta}_{\mathrm{group},\bar{\matF}}-\mbf{1}_{\dimMatC}\|_1}\right\|_1}_{=\widehat{\Delta}(\tilde{\matF}, \bar{\matF})}.
\end{align}
The first term on the RHS of \eqref{full_est_Err} has already been bounded in \eqref{eq:fullUppBoundAAbar}. In the following, we bound the second term in \eqref{full_est_Err} and then combine the two bounds. 
We substitute $\hat{\bs{\theta}}_{\mathrm{group}}^{(i)}$ from the output of Algorithm ~\ref{alg:modular_coarse_ctrl}, and $\bs{\theta}_{\mathrm{group},\bar{\matF}}^{(i)}$ from \eqref{groupAbar}, and $\inCtrlMat$ defined in \eqref{upsilon_def}:
\begin{align}\label{thetaHatAbardiff1}
	\widehat{\Delta}(\matC, \bar{\matF}) 
	= \left\|
\frac{\hat{\syncComAs} {\mathrm{diag}(\hat{\inCtrlMat})}}{\|\hat{\syncComAs} {\mathrm{diag}(\hat{\inCtrlMat})}\|_1 }  - \frac{\syncComAs {\mathrm{diag}(\inCtrlMat)} }{\|\syncComAs {\mathrm{diag}(\inCtrlMat)} \|_1}	
	\right\|_1 .
\end{align}
We set $\hat{\syncComAs}= \syncComAs + \errorMat_{\syncComAs}$, $\hat{\matCom}= \matCom+ \errorMat_{\matCom}$, and  $\hat{\diagComSize}=\diagComSize+\errorMat_{\diagComSize}$ as the error matrices of appropriate sizes. Substitution of these error matrices into \eqref{thetaHatAbardiff1}, successive application of the triangle inequality, and the use of Proposition~\ref{prop:uppBound_relDiff} gives
\begin{align}\label{thetaHatAbardiff2}
	\nonumber\widehat{\Delta}(\matC, \bar{\matF}) &
	 \leq 2 \frac{\|\vert \hat{\syncComAs} {\mathrm{diag}(\hat{\inCtrlMat})}-\syncComAs\mathrm{diag}(\inCtrlMat)\|_1}{\max\left({\|\hat{\syncComAs} {\mathrm{diag}(\hat{\inCtrlMat})}\|_1, \|\syncComAs\mathrm{diag}(\inCtrlMat)\|_1}\right)}
	 \\
	 \nonumber&\leq 2 \frac{\|\vert (\syncComAs + \errorMat_{\syncComAs}) {\mathrm{diag}(\hat{\inCtrlMat})}-\syncComAs\mathrm{diag}(\inCtrlMat)\|_1}{ \|\syncComAs\mathrm{diag}(\inCtrlMat)\|_1}
	 \\
	 \nonumber&\leq 2 \frac{ \| \syncComAs\mathrm{diag}(\hat{\inCtrlMat}-\inCtrlMat)\|_1 + \|\vert \errorMat_{\syncComAs} {\mathrm{diag}(\hat{\inCtrlMat})}\|_1}{ \dimMatC \|\mathrm{diag}(\inCtrlMat)\|_{\mathrm{min}}}
	 \\
	 & \leq 2 \frac{  \|\hat{\inCtrlMat}-\inCtrlMat\|_{\mathrm{max}} + \|\hat{\inCtrlMat}\|_{\mathrm{max}} \|\errorMat_{\syncComAs}\|_{1,1}/\dimMatC }{ \|\inCtrlMat\|_{\mathrm{min}}}
	.
\end{align}
To continue the simplification of $\|\hat{\inCtrlMat}-\inCtrlMat\|_{\mathrm{max}}$ in \eqref{thetaHatAbardiff2}, we bring the two error matrices $\errorMat_{\matCom}$ and $\errorMat_{\diagComSize}$ introduced in the statement of the theorem into play:
\begin{align}\label{thetaHatAbardiff6}
\nonumber(\matCom_{\circ}+ \errorMat_{\matCom})(\diagComSize+\errorMat_{\diagComSize})(\matCom_{\circ} + \errorMat_{\matCom}) 
& 
= \matCom_{\circ} \diagComSize\matCom_{\circ}
+ \mbf{\Gamma}_1, \\
\nonumber(\diagComSize+\errorMat_{\diagComSize}) (\matCom_{\circ}+ \errorMat_{\matCom}) & = \diagComSize\matCom_{\circ} +\mbf{\Gamma}_2,
\\  
\nonumber (\diagComSize^{\frac{1}{2}}+\errorMat_{\diagComSize}^{\frac{1}{2}})(\matCom+ \errorMat_{\matCom})  (\diagComSize^{\frac{1}{2}}+\errorMat_{\diagComSize}^{\frac{1}{2}}) 
    & = \diagComSize^{\frac{1}{2}}\matCom \diagComSize^{\frac{1}{2}}+{\mbf{\Gamma}_3}, \\
    \left( \diagComSize\matCom_{\circ} + \mbf{\Gamma}_2\right)^{2\ell}  & = (\diagComSize\matCom_{\circ})^{2\ell} +  {\mbf{\Gamma}_4},
\end{align} 
where we define
\begin{align*}
\nonumber\mbf{\Gamma}_1 & \triangleq \matCom_{\circ}\diagComSize\errorMat_{\matCom}
+ \matCom_{\circ}\errorMat_{\diagComSize}\matCom_{\circ}
+ \matCom_{\circ}\errorMat_{\diagComSize} \errorMat_{\matCom} 
+\errorMat_{\matCom} \diagComSize\matCom_{\circ}
\\
\nonumber& \qquad +\errorMat_{\matCom} \diagComSize\errorMat_{\matCom}
+\errorMat_{\matCom} \errorMat_{\diagComSize}\matCom_{\circ}
+\errorMat_{\matCom} \errorMat_{\diagComSize} \errorMat_{\matCom}, \\
\mbf{\Gamma}_2 & \triangleq \diagComSize\errorMat_{\matCom} + \errorMat_{\diagComSize}\matCom_{\circ} + \errorMat_{\diagComSize}\errorMat_{\matCom},
\end{align*}
The exact definitions of $\mbf{\Gamma}_3$ and $\mbf{\Gamma}_4$ can be  retrieved by simple expansion, where $\mbf{\Gamma}_3=\Oscale(\errorMat_{\matCom}+\errorMat_{\diagComSize}), \mbf{\Gamma}_4=\Oscale(\mbf{\Gamma}_2)$ when $\errorMat_{\diagComSize},\errorMat_{\matCom}$ are small.
From \eqref{thetaHatAbardiff6}, we have
\begin{align}\label{thetaHatAbardiff8}
    \lvert\specRadius(\matCom \diagComSize+\mbf{\Gamma}_3)-\specRadius(\matCom \diagComSize) \rvert \stackrel{(a)}{\leq}   \|\mbf{\Gamma}_3\|_2 \stackrel{(b)}{\leq} \dimMatCom \|\mbf{\Gamma}_3\|_{\mathrm{max}},
\end{align}
where (a) is due to Proposition~\ref{prop:matrixElemInequExtends2norms}, and (b) uses item~\ref{item:twoNormMaxIneq} of the list in beginning of the appendix.
Substituting \eqref{thetaHatAbardiff6}, \eqref{thetaHatAbardiff8}, \eqref{upsilon_def2}, and the equivalence of \eqref{hat_defs} (similar to \eqref{upsilon_def2}) into \eqref{thetaHatAbardiff2} yields:
\begin{align}\label{Ebarmax2}
\nonumber\| \hat{\inCtrlMat}-\inCtrlMat\|_{\mathrm{max}} 
 & \!=\! \| \hat{\overalConstTemp}^2{\hat{\matCom}\hat{\diagComSize}\hat{\matCom}} (\mbf{I} \! - \!  ( \hat{\overalConstTemp}{\hat{\diagComSize} \hat{\matCom}})^2\!)^{-1} \!\! - \!\overalConstTemp^2{\matCom_{\circ} \diagComSize \matCom_{\circ}} (1\!-\!(\overalConstTemp\!{\diagComSize \!\matCom_{\circ}}\!)^2)^{-1} \|_{\mathrm{max}}\\
 \nonumber &\!=\! \hat{\overalConstTemp}^2\| (\matCom_{\circ} \diagComSize\matCom_{\circ}
+ \mbf{\Gamma}_1) (\mbf{I}  -  ( \hat{\overalConstTemp}{\hat{\diagComSize} \hat{\matCom}})^2)^{-1}
- (\frac{\overalConstTemp}{\hat{\overalConstTemp}})^2 {\matCom_{\circ} \diagComSize \matCom_{\circ}} (1-(\overalConstTemp{\diagComSize \matCom_{\circ}})^2)^{-1} \|_{\mathrm{max}}
\\
\nonumber & = \hat{\overalConstTemp}^2 \| (\matCom_{\circ} \diagComSize\matCom_{\circ}
) {[(\mbf{I}  -  ( \hat{\overalConstTemp}{\hat{\diagComSize} \hat{\matCom}})^2)^{-1} -(\mbf{I}-(\overalConstTemp{\diagComSize \matCom_{\circ}})^2)^{-1}]}
\\ \nonumber & \qquad
 + \mbf{\Gamma}_1 (\!\mbf{I} 
\!-\!  ( \hat{\overalConstTemp}{\hat{\diagComSize} \hat{\matCom}})^2\!)^{-1}
\!-\! \frac{\overalConstTemp^2\!-\!\hat{\overalConstTemp}^2}{\hat{\overalConstTemp}^2} 
{\matCom_{\circ} \diagComSize \matCom_{\circ}} (\mbf{I}\!-\!(\overalConstTemp{\diagComSize \matCom_{\circ}})^2)^{-1} \!\|_{\mathrm{max}}
\\
\nonumber & \stackrel{(c)}{\leq}   \!\left[\!\| \matCom_{\circ} \diagComSize \matCom_{\circ}\!\|_{\mathrm{max}} \| \mbf{\Gamma}_5\!\|_{\mathrm{max}} \!+\! \| \mbf{\Gamma}_1\!\|_{\mathrm{max}} \| (\!1\!-\!(\overalConstTemp{\diagComSize \matCom_{\circ}}\!)^2)^{-1}\|_{\mathrm{max}} \! \right.
\\\nonumber & 
\qquad \quad  
+\left(2\specRadius(\matCom \diagComSize)+ \specRadius(\mbf{\Gamma}_3) +\frac{2\constFine}{\dimMatF}\right) \frac{\|\mbf{\Gamma}_3\|_2}{\left(\frac{\constFine}{\dimMatF} +\specRadius(\matCom \diagComSize)\right)^2}
\left. \| \matCom_{\circ} \diagComSize \matCom_{\circ}\|_{\mathrm{max}} \| (\mbf{I}  -  (\overalConstTemp{\diagComSize \matCom_{\circ}})^2)^{-1}\|_{\mathrm{max}}\right] \dimMatCom \hat{\overalConstTemp}^2\\
 & \stackrel{(d)}{=} \Oscale(\| \mbf{\Gamma}_5\|_{\mathrm{max}} + \| \mbf{\Gamma}_1\|_{\mathrm{max}} +\|\mbf{\Gamma}_3\|_{\mathrm{max}}),
\end{align}
where 
    $\mbf{\Gamma}_5 \triangleq (\mbf{I}  -  ( \hat{\overalConstTemp}{\hat{\diagComSize} \hat{\matCom}})^2)^{-1} -(\mbf{I}-(\overalConstTemp{\diagComSize \matCom_{\circ}})^2)^{-1}$;
(c) follows from the triangle inequality of matrix norms, multiple use of item~\ref{item:maxProdMax} in the list at the beginning of the appendix, \eqref{upsilon_def}, \eqref{hat_defs}, and inequality (a) of \eqref{thetaHatAbardiff8}; and (d) is due to inequality (b) of \eqref{thetaHatAbardiff8}, $\matCom_{\circ},\diagComSize$ being constant values, and for when $\mbf{\Gamma}_1, \mbf{\Gamma}_3,$ and $\mbf{\Gamma}_5$ are small. Next, we further simplify \eqref{Ebarmax2} by finding the max norms of $\mbf{\Gamma}_1, \mbf{\Gamma}_3,$ and $\mbf{\Gamma}_5$. First, we upper bound $\|\mbf{\Gamma}_5\|_{\mathrm{max}}$ in \eqref{Ebarmax2} with the help of Neumann series, \eqref{thetaHatAbardiff6}, and triangle inequality on matrix norms:
\begin{align}\label{ep1max}
\|\mbf{\Gamma}_5\|_{\mathrm{max}}\!\!\!\! \nonumber&=\left\|\displaystyle\sum_{\ell\geq 1} [( \hat{\overalConstTemp}{\hat{\diagComSize} \hat{\matCom}}
)^{2\ell} - (\overalConstTemp{\diagComSize \matCom_{\circ}})^{2\ell}]\right\|_{\mathrm{max}} \\
\nonumber& =\left\|\displaystyle\sum_{\ell\geq 1} \hat{\overalConstTemp}^{2\ell} \mbf{\Gamma}_4- (\overalConstTemp^{2\ell}-\hat{\overalConstTemp}^{2\ell})({\diagComSize \matCom_{\circ}})^{2\ell}\right\|_{\mathrm{max}} \\
\nonumber& \leq \left\|\hat{\overalConstTemp}^2(1-\hat{\overalConstTemp}^2)^{-1} \mbf{\Gamma}_4 \right\|_{\mathrm{max}} 
\!+\! \left\|({\diagComSize \matCom_{\circ}})^2 [1\!-\!({\diagComSize \matCom_{\circ}})^2]^{-1}\right\|_{\mathrm{max}} \displaystyle\max_{\ell\geq 1}\lvert \overalConstTemp^{2\ell}\!-\!\hat{\overalConstTemp}^{2\ell} \rvert \\
& 
\stackrel{(a)}{=} \!\Oscale\left(\| \mbf{\Gamma}_4\|_{\mathrm{max}}\!+\!\left\lvert \overalConstTemp\!\shortminus\!\hat{\overalConstTemp}\right\rvert\right)
\!\stackrel{(b)}{=}\! \Oscale(\| \mbf{\Gamma}_2\|_{\mathrm{max}}\!+\!\|\mbf{\Gamma}_3\|_{\mathrm{max}}\!),
\end{align}
where (a) holds when $\mbf{\Gamma}_2,\mbf{\Gamma}_3$ are small, and because $\overalConstTemp, \hat{\overalConstTemp}$ are constant.
Similar to \eqref{ep1max}, and by using triangle inequality on matrix norms as well as item~\ref{item:maxProdMax} in the list at the beginning of the appendix, $\| \mbf{\Gamma}_1\|_{\mathrm{max}}, \| \mbf{\Gamma}_2\|_{\mathrm{max}}$ are bounded from above by
\begin{align}\label{gammaMax}
 \| \mbf{\Gamma}_1\|_{\mathrm{max}}  
\nonumber &  \leq\!
\dimMatCom^2\left(\| \matCom_{\circ}\diagComSize\|_{\mathrm{max}}\| \errorMat_{\matCom}\|_{\mathrm{max}} \!+\!  \| \matCom_{\circ}\|_{\mathrm{max}}\| \errorMat_{\diagComSize}\|_{\mathrm{max}}\| \matCom\|_{\mathrm{max}} \right.
\\  \nonumber& \qquad \quad
\!+\! \| \matCom_{\circ}\|_{\mathrm{max}} \| \errorMat_{\diagComSize}\|_{\mathrm{max}} \| \errorMat_{\matCom} \|_{\mathrm{max}}\!+\!\| \errorMat_{\matCom}\|_{\mathrm{max}} \| \diagComSize\matCom_{\circ}\|_{\mathrm{max}}
\\  \nonumber&  \qquad \qquad
 \!+\!\| \errorMat_{\matCom}\!\|_{\mathrm{max}} \|  \diagComSize\|_{\mathrm{max}}\| \errorMat_{\matCom}\!\|_{\mathrm{max}}
 \!+\!\| \errorMat_{\matCom}\!\!\|_{\mathrm{max}} \| \errorMat_{\diagComSize}\!\!\|_{\mathrm{max}} \| \matCom_{\circ}\!\|_{\mathrm{max}}   
 +\left.\| \errorMat_{\matCom}\|_{\mathrm{max}} \| \errorMat_{\diagComSize}\|_{\mathrm{max}} \| \errorMat_{\matCom}\|_{\mathrm{max}}\right)\\
 \nonumber& \leq
\!\dimMatCom^2\!\left(\! \| \errorMat_{\matCom}\!\|_{\mathrm{max}}\!+\!  \| \errorMat_{\diagComSize}\!\|_{\mathrm{max}}\! +\! \| \errorMat_{\diagComSize}\!\|_{\mathrm{max}} \| \errorMat_{\matCom}\! \|_{\mathrm{max}}\!\!+\!\| \errorMat_{\matCom}\!\|_{\mathrm{max}} \right.
\\ \nonumber& \qquad \qquad
\left.+\| \errorMat_{\matCom}\!\|_{\mathrm{max}}^2
+  \| \errorMat_{\matCom}\!\|_{\mathrm{max}} \| \errorMat_{\diagComSize}\!\|_{\mathrm{max}}+\| \errorMat_{\matCom}\!\|_{\mathrm{max}}^2 \| \errorMat_{\diagComSize}\!\|_{\mathrm{max}}\right) \\
 & = \Oscale(\| \errorMat_{\matCom}\|_{\mathrm{max}} + \| \errorMat_{\diagComSize}\|_{\mathrm{max}}),
\end{align} 
and
\begin{align}\label{frackEmax}
\nonumber\| \mbf{\Gamma}_2\|_{\mathrm{max}} & = \|  \diagComSize\errorMat_{\matCom} + \errorMat_{\diagComSize}\matCom_{\circ} + \errorMat_{\diagComSize}\errorMat_{\matCom}\|_{\mathrm{max}}\\
\nonumber& \leq \dimMatCom( \|  \diagComSize\|_{\mathrm{max}}\| \errorMat_{\matCom}\|_{\mathrm{max}} + \| \errorMat_{\diagComSize}\|_{\mathrm{max}}\| \matCom_{\circ}\|_{\mathrm{max}} 
+ \| \errorMat_{\diagComSize}\|_{\mathrm{max}}\| \errorMat_{\matCom}\|_{\mathrm{max}})\\
\nonumber& \stackrel{(b)}{\leq} \!\dimMatCom \!( \| \errorMat_{\matCom}\!\|_{\mathrm{max}} \!+ \!\| \errorMat_{\diagComSize}\!\|_{\mathrm{max}} \!+\! \| \errorMat_{\diagComSize}\!\|_{\mathrm{max}}\| \errorMat_{\matCom}\|_{\mathrm{max}}\!) \\
& = \Oscale(\| \errorMat_{\matCom}\|_{\mathrm{max}} + \| \errorMat_{\diagComSize}\|_{\mathrm{max}}),
\end{align}
where (b) in \eqref{frackEmax} is because of the definitions of $\diagComSize$ and $\matCom_{\circ}$ respectively in Assumption~\ref{assump: graph sparsity} and the paragraph following \eqref{eq: SBM}; and the scaling behaviour both in \eqref{gammaMax} and \eqref{frackEmax} is for when $\errorMat_{\matCom}, \errorMat_{\diagComSize}$ are small.

By substituting \eqref{ep1max}, \eqref{gammaMax}, and \eqref{frackEmax} into \eqref{Ebarmax2},
we get:
\begin{align}\label{Ebarmax}
\| \hat{\inCtrlMat}-\inCtrlMat\|_{\mathrm{max}} = \Oscale(\| \errorMat_{\matCom}\|_{\mathrm{max}} + \| \errorMat_{\diagComSize}\|_{\mathrm{max}}),
\end{align}
Substituting \eqref{Ebarmax} into the original error in \eqref{thetaHatAbardiff2}, yields:
\begin{align}\label{finalErrorHatAtildeAbar}
	\widehat{\Delta}(\matC, \bar{\matF})
	= \Oscale\left(\frac{\|\errorMat_{\syncComAs}\|_{1,1}}{\dimMatC}+\| \errorMat_{\matCom}\|_{\mathrm{max}} + \| \errorMat_{\diagComSize}\|_{\mathrm{max}}\right),
\end{align}
because $\|\inCtrlMat\|_{\mathrm{min}}=\Omega(1)$ and $\|\hat{\inCtrlMat}\|_{\mathrm{max}}=\Oscale(1)$, per their definitions in \eqref{upsilon_def} and \eqref{hat_defs}.
The proof concludes by merging \eqref{finalErrorHatAtildeAbar} and Theorem~\ref{thm:est_cntrl_error} and substituting into \eqref{full_est_Err}.\qed
}

{\small
	\bibliography{main}

\begin{thebibliography}{10}

\bibitem{betzel2017multi}
R.~F. Betzel and D.~S. Bassett, ``Multi-scale brain networks,'' {\em
  Neuroimage}, vol.~160, pp.~73--83, 2017.

\bibitem{gu2015controllability}
S.~Gu, F.~Pasqualetti, M.~Cieslak, Q.~K. Telesford, B.~Y. Alfred, A.~E. Kahn,
  J.~D. Medaglia, J.~M. Vettel, M.~B. Miller, S.~T. Grafton, {\em et~al.},
  ``Controllability of structural brain networks,'' {\em Nature
  Communications}, vol.~6, no.~1, pp.~1--10, 2015.

\bibitem{karrer2020practical}
T.~M. Karrer, J.~Z. Kim, J.~Stiso, A.~E. Kahn, F.~Pasqualetti, U.~Habel, and
  D.~S. Bassett, ``A practical guide to methodological considerations in the
  controllability of structural brain networks,'' {\em Journal of Neural Eng.},
  vol.~17, no.~2, p.~026031, 2020.

\bibitem{heck2014two}
C.~N. Heck, D.~King-Stephens, A.~D. Massey, D.~R. Nair, B.~C. Jobst, G.~L.
  Barkley, V.~Salanova, A.~J. Cole, M.~C. Smith, R.~P. Gwinn, {\em et~al.},
  ``Two-year seizure reduction in adults with medically intractable partial
  onset epilepsy treated with responsive neurostimulation: Final results of the
  rns system pivotal trial,'' {\em Epilepsia}, vol.~55, no.~3, pp.~432--441,
  2014.

\bibitem{benner15}
P.~Benner, S.~Gugercin, and K.~Willcox, ``A survey of projection-based model
  reduction methods for parametric dynamical systems,'' {\em SIAM Review},
  vol.~57, no.~4, pp.~483--531, 2015.

\bibitem{pasqualetti2014controllability}
F.~Pasqualetti, S.~Zampieri, and F.~Bullo, ``Controllability metrics,
  limitations and algorithms for complex networks,'' {\em IEEE Trans. on
  Control of Network Syst.}, vol.~1, no.~1, pp.~40--52, 2014.

\bibitem{yuan2013exact}
Z.~Yuan, C.~Zhao, Z.~Di, W.-X. Wang, and Y.-C. Lai, ``Exact controllability of
  complex networks,'' {\em Nature Communications}, vol.~4, no.~1, pp.~1--9,
  2013.

\bibitem{baggio2022}
G.~Baggio, F.~Pasqualetti, and S.~Zampieri, ``Energy-aware controllability of
  complex networks,'' {\em Annual Review of Control, Robotics, and Autonomous
  Systems}, vol.~5, no.~1, 2022.

\bibitem{lindmark2020}
G.~Lindmark, {\em Controllability of Complex Networks at Minimum Cost}.
\newblock PhD thesis, Link\"{o}ping University Electronic Press, 2020.
\newblock (visited on 04/24/2022).

\bibitem{o2018nurip}
G.~O’Leary, D.~M. Groppe, T.~A. Valiante, N.~Verma, and R.~Genov, ``{NURIP}:
  Neural interface processor for brain-state classification and
  programmable-waveform neurostimulation,'' {\em IEEE Journal of Solid-State
  Circuits}, vol.~53, no.~11, pp.~3150--3162, 2018.

\bibitem{kassiri2017closed}
H.~Kassiri, S.~Tonekaboni, M.~T. Salam, N.~Soltani, K.~Abdelhalim, J.~L.~P.
  Velazquez, and R.~Genov, ``Closed-loop neurostimulators: A survey and a
  seizure-predicting design example for intractable epilepsy treatment,'' {\em
  IEEE Trans. on Biomedical Circuits and Syst.}, vol.~11, no.~5,
  pp.~1026--1040, 2017.

\bibitem{schaub2020blind}
M.~T. Schaub, S.~Segarra, and J.~N. Tsitsiklis, ``Blind identification of
  stochastic block models from dynamical observations,'' {\em SIAM Journal on
  Mathematics of Data Science}, vol.~2, no.~2, pp.~335--367, 2020.

\bibitem{xing2020community}
Y.~Xing, X.~He, H.~Fang, and K.~H. Johansson, ``Community detection for gossip
  dynamics with stubborn agents,'' in {\em 2020 59th IEEE Conference on
  Decision and Control}, pp.~4915--4920, 2020.

\bibitem{abbe2017community}
E.~Abbe, ``Community detection and stochastic block models: Recent
  developments,'' {\em The Journal of Machine Learning Research}, vol.~18,
  no.~1, pp.~6446--6531, 2017.

\bibitem{betzel2018diversity}
R.~F. Betzel, J.~D. Medaglia, and D.~S. Bassett, ``Diversity of meso-scale
  architecture in human and non-human connectomes,'' {\em Nature
  Communications}, vol.~9, no.~1, pp.~1--14, 2018.

\bibitem{farina2000positive}
L.~Farina and S.~Rinaldi, {\em Positive linear systems: Theory and
  applications}, vol.~50.
\newblock John Wiley \& Sons, 2000.

\bibitem{mao2020estimating}
X.~Mao, P.~Sarkar, and D.~Chakrabarti, ``Estimating mixed memberships with
  sharp eigenvector deviations,'' {\em Journal of the American Statistical
  Association}, no.~just-accepted, pp.~1--24, 2020.

\bibitem{kalman1960b}
R.~Kalman, ``On the general theory of control systems,'' {\em IFAC Proceedings
  Volumes}, vol.~1, no.~1, pp.~491--502, 1960.
\newblock 1st Int. IFAC Congress on Automatic and Remote Control, Moscow, USSR,
  1960.

\bibitem{klamka2018controllability}
J.~Klamka, {\em Controllability and Minimum Energy Control}.
\newblock Studies in Systems, Decision and Control, Springer Int. Publishing,
  2018.

\bibitem{pequito2016minimum}
S.~Pequito, S.~Kar, and A.~P. Aguiar, ``Minimum cost input/output design for
  large-scale linear structural systems,'' {\em Automatica}, vol.~68,
  pp.~384--391, 2016.

\bibitem{olshevsky2014minimal}
A.~Olshevsky, ``Minimal controllability problems,'' {\em IEEE Trans. on Control
  of Network Syst.}, vol.~1, no.~3, pp.~249--258, 2014.

\bibitem{tzoumas2015minimal}
V.~Tzoumas, M.~A. Rahimian, G.~J. Pappas, and A.~Jadbabaie, ``Minimal actuator
  placement with bounds on control effort,'' {\em IEEE Trans. on Control of
  Network Syst.}, vol.~3, no.~1, pp.~67--78, 2015.

\bibitem{dilip2019}
A.~S.~A. Dilip, ``The controllability gramian, the {H}adamard product, and the
  optimal actuator/leader and sensor selection problem,'' {\em IEEE Control
  Syst. Lett.}, vol.~3, no.~4, pp.~883--888, 2019.

\bibitem{chanekar2021}
P.~V. Chanekar, E.~Nozari, and J.~Cortés, ``Energy-transfer edge centrality
  and its role in enhancing network controllability,'' {\em IEEE Trans. on
  Network Science and Eng.}, vol.~8, no.~1, pp.~331--346, 2021.

\bibitem{vosughi2019target}
A.~Vosughi, C.~Johnson, M.~Xue, S.~Roy, and S.~Warnick, ``Target control and
  source estimation metrics for dynamical networks,'' {\em Automatica},
  vol.~100, pp.~412--416, 2019.

\bibitem{klickstein2018control}
I.~Klickstein and F.~Sorrentino, ``Control distance and energy scaling of
  complex networks,'' {\em IEEE Trans. on Network Science and Eng.}, vol.~7,
  no.~2, pp.~726--736, 2018.

\bibitem{cortesi2014submodularity}
F.~L. Cortesi, T.~H. Summers, and J.~Lygeros, ``Submodularity of energy related
  controllability metrics,'' in {\em $53^{rd}$ IEEE Conference on Decision and
  Control}, pp.~2883--288, 2014.

\bibitem{srighakollapu2021optimizing}
M.~V. Srighakollapu, R.~K. Kalaimani, and R.~Pasumarthy, ``Optimizing network
  topology for average controllability,'' {\em Syst. \& Control Lett.},
  vol.~158, p.~105061, 2021.

\bibitem{cortesi14submodular}
F.~L. Cortesi, T.~H. Summers, and J.~Lygeros, ``Submodularity of energy related
  controllability metrics,'' in {\em $53^{rd}$ IEEE Conference on Decision and
  Control}, pp.~2883--2888, 2014.

\bibitem{preciado2016controllability}
V.~M. Preciado and M.~A. Rahimian, ``Controllability gramian spectra of random
  networks,'' in {\em 2016 American Control Conference}, pp.~3874--3879, 2016.

\bibitem{bopardikar2021randomized}
S.~D. Bopardikar, ``A randomized approach to sensor placement with
  observability assurance,'' {\em Automatica}, vol.~123, p.~109340, 2021.

\bibitem{antoulas2005}
A.~C. Antoulas, {\em Approximation of Large-Scale Dynamical Systems}.
\newblock Society for Industrial and Applied Mathematics, 2005.

\bibitem{moore1981principal}
B.~Moore, ``Principal component analysis in linear systems: Controllability,
  observability, and model reduction,'' {\em IEEE Trans. on Automatic Control},
  vol.~26, no.~1, pp.~17--32, 1981.

\bibitem{cheng2021model}
X.~Cheng and J.~Scherpen, ``Model reduction methods for complex network
  systems,'' {\em Annual Review of Control, Robotics, and Autonomous Syst.},
  vol.~4, pp.~425--453, 2021.

\bibitem{chengrobust2018}
X.~Cheng and J.~M. Scherpen, ``Robust synchronization preserving model
  reduction of {L}ur’e networks,'' in {\em European Control Conference},
  pp.~2254--2259, 2018.

\bibitem{MONSHIZADEH20131}
N.~Monshizadeh, H.~L. Trentelman, and M.~{Kanat Camlibel}, ``Stability and
  synchronization preserving model reduction of multi-agent systems,'' {\em
  Syst. \& Control Lett.}, vol.~62, no.~1, pp.~1--10, 2013.

\bibitem{jongsmamodel2018}
H.-J. Jongsma, P.~Mlinari\'{c}, S.~Grundel, P.~Benner, and H.~L. Trentelman,
  ``Model reduction of linear multi-agent systems by clustering with ${H}_2$
  and ${H}_\infty$ error bounds,'' {\em Mathematics of Control, Signals, and
  Syst.}, no.~6, 2018.

\bibitem{Benner2021}
P.~Benner, S.~Grundel, and P.~Mlinari{\'{c}}, {\em Clustering-Based Model Order
  Reduction for Nonlinear Network Systems}, pp.~75--96.
\newblock Cham: Springer Int. Publishing, 2021.

\bibitem{chengreduction2017}
X.~Cheng, Y.~Kawano, and J.~M.~A. Scherpen, ``Reduction of second-order network
  systems with structure preservation,'' {\em IEEE Trans. on Automatic
  Control}, vol.~62, no.~10, pp.~5026--5038, 2017.

\bibitem{Monshizadeh2014}
N.~Monshizadeh, H.~L. Trentelman, and M.~K. Camlibel, ``Projection-based model
  reduction of multi-agent systems using graph partitions,'' {\em IEEE Trans.
  on Control of Network Systems}, vol.~1, no.~2, pp.~145--154, 2014.

\bibitem{dulac2020mixed}
A.~Dulac, E.~Gaussier, and C.~Largeron, ``Mixed-membership stochastic block
  models for weighted networks,'' in {\em Conference on Uncertainty in
  Artificial Intelligence}, pp.~679--688, 2020.

\bibitem{funke2019stochastic}
T.~Funke and T.~Becker, ``Stochastic block models: A comparison of variants and
  inference methods,'' {\em PloS one}, vol.~14, no.~4, p.~e0215296, 2019.

\bibitem{ghoroghchian2018hierarchical}
N.~Ghoroghchian, S.~C. Draper, and R.~Genov, ``A hierarchical graph signal
  processing approach to inference from spatiotemporal signals,'' in {\em 2018
  29th Biennial Symp. on Commun.}, pp.~1--5, 2018.

\bibitem{ghoroghchian2020node}
N.~Ghoroghchian, D.~M. Groppe, R.~Genov, T.~A. Valiante, and S.~C. Draper,
  ``Node-centric graph learning from data for brain state identification,''
  {\em IEEE Trans. on Signal and Inf. Process. over Networks}, vol.~6,
  pp.~120--132, 2020.

\bibitem{gilbert2004compressing}
A.~C. Gilbert and K.~Levchenko, ``Compressing network graphs,'' in {\em Proc.
  of the LinkKDD workshop at the 10th ACM Conference on KDD}, vol.~124, 2004.

\bibitem{ahn2012graph}
K.~J. Ahn, S.~Guha, and A.~McGregor, ``Graph sketches: Sparsification,
  spanners, and subgraphs,'' in {\em Proc. of the 31st ACM SIGMOD-SIGACT-SIGAI
  symposium on Principles of Database Syst.}, pp.~5--14, 2012.

\bibitem{dasarathy2015sketching}
G.~Dasarathy, P.~Shah, B.~N. Bhaskar, and R.~D. Nowak, ``Sketching sparse
  matrices, covariances, and graphs via tensor products,'' {\em IEEE Trans. on
  Inf. Theory}, vol.~61, no.~3, pp.~1373--1388, 2015.

\bibitem{ghoroghchian2021graph}
N.~Ghoroghchian, G.~Dasarathy, and S.~Draper, ``Graph community detection from
  coarse measurements: Recovery conditions for the coarsened weighted
  stochastic block model,'' in {\em Int. Conference on Artificial Intelligence
  and Statistics}, pp.~3619--3627, 2021.

\bibitem{muller1972analysis}
P.~M{\"u}ller and H.~Weber, ``Analysis and optimization of certain qualities of
  controllability and observability for linear dynamical systems,'' {\em
  Automatica}, vol.~8, no.~3, pp.~237--246, 1972.

\bibitem{yan2012controlling}
G.~Yan, J.~Ren, Y.-C. Lai, C.-H. Lai, and B.~Li, ``Controlling complex
  networks: How much energy is needed?,'' {\em Physical Review Lett.},
  vol.~108, no.~21, p.~218703, 2012.

\bibitem{summers2015submodularity}
T.~H. Summers, F.~L. Cortesi, and J.~Lygeros, ``On submodularity and
  controllability in complex dynamical networks,'' {\em IEEE Trans. on Control
  of Network Syst.}, vol.~3, no.~1, pp.~91--101, 2015.

\bibitem{ikeda2018sparsity}
T.~Ikeda and K.~Kashima, ``Sparsity-constrained controllability maximization
  with application to time-varying control node selection,'' {\em IEEE Control
  Syst. Lett.}, vol.~2, no.~3, pp.~321--326, 2018.

\bibitem{pirani2020controllability}
M.~Pirani and J.~A. Taylor, ``Controllability of ac power networks with dc
  lines,'' {\em IEEE Trans. on Power Syst.}, vol.~36, no.~2, pp.~1649--1651,
  2020.

\bibitem{mcgowan2021controllability}
A.~L. McGowan, L.~Parkes, X.~He, O.~Stanoi, Y.~Kang, S.~Lomax, M.~Jovanova,
  P.~J. Mucha, K.~N. Ochsner, E.~B. Falk, {\em et~al.}, ``Controllability of
  structural brain networks and the waxing and waning of negative affect in
  daily life,'' {\em Biological Psychiatry Global Open Science}, 2021.

\bibitem{fang2022effects}
F.~Fang, B.~Godlewska, R.~Y. Cho, S.~I. Savitz, S.~Selvaraj, and Y.~Zhang,
  ``Effects of escitalopram therapy on functional brain controllability in
  major depressive disorder,'' {\em Journal of Affective Disorders}, 2022.

\bibitem{deng2020controllability}
S.~Deng and S.~Gu, ``Controllability analysis of functional brain networks,''
  {\em arXiv preprint arXiv:2003.08278}, 2020.

\bibitem{recht2010guaranteed}
B.~Recht, M.~Fazel, and P.~A. Parrilo, ``Guaranteed minimum-rank solutions of
  linear matrix equations via nuclear norm minimization,'' {\em SIAM review},
  vol.~52, no.~3, pp.~471--501, 2010.

\bibitem{chi2017convex}
C.-Y. Chi, W.-C. Li, and C.-H. Lin, {\em Convex optimization for signal
  processing and communications: From fundamentals to applications}.
\newblock CRC press, 2017.

\bibitem{dunford1988linear}
N.~Dunford and J.~T. Schwartz, {\em Linear operators, part 1: General theory},
  vol.~10.
\newblock John Wiley \& Sons, 1988.

\bibitem{sporns2016modular}
O.~Sporns and R.~F. Betzel, ``Modular brain networks,'' {\em Annual review of
  psychology}, vol.~67, pp.~613--640, 2016.

\bibitem{galhotra2018geometric}
S.~Galhotra, A.~Mazumdar, S.~Pal, and B.~Saha, ``The geometric block model,''
  in {\em Proc. of the AAAI Conference on Artificial Intelligence}, vol.~32,
  2018.

\bibitem{rossetti2019cdlib}
G.~Rossetti, L.~Milli, and R.~Cazabet, ``Cdlib: A python library to extract,
  compare and evaluate communities from complex networks,'' {\em Applied
  Network Science}, vol.~4, no.~1, p.~52, 2019.

\bibitem{aicher2015learning}
C.~Aicher, A.~Z. Jacobs, and A.~Clauset, ``Learning latent block structure in
  weighted networks,'' {\em Journal of Complex Networks}, vol.~3, no.~2,
  pp.~221--248, 2015.

\bibitem{lin2016dirichlet}
J.~Lin, ``On the {D}irichlet distribution,'' {\em Department of Mathematics and
  Statistics, Queens University}, 2016.

\bibitem{liu2009trace}
J.~Liu, J.~Zhang, and Y.~Liu, ``Trace inequalities for matrix products and
  trace bounds for the solution of the algebraic riccati equations,'' {\em
  Journal of Inequalities and Applications}, vol.~2009, pp.~1--17, 2009.

\bibitem{horn2012matrix}
R.~A. Horn and C.~R. Johnson, {\em Matrix analysis}.
\newblock Cambridge university press, 2012.

\bibitem{chapter5April26folder}
F.~Keinert, ``Course notes: Applied linear algebra.''
\newblock
  \url{https://orion.math.iastate.edu/keinert/math507/notes/chapter5.pdf}.

\bibitem{boucheron2013concentration}
S.~Boucheron, G.~Lugosi, and P.~Massart, {\em Concentration inequalities: A
  nonasymptotic theory of independence}.
\newblock Oxford university press, 2013.

\bibitem{tyrtyshnikov1997brief}
E.~E. Tyrtyshnikov, {\em A brief introduction to numerical analysis}.
\newblock Springer Science $\&$ Business Media, 1997.

\bibitem{bardenet2015concentration}
R.~Bardenet and O.-A. Maillard, ``Concentration inequalities for sampling
  without replacement,'' {\em Bernoulli}, vol.~21, no.~3, pp.~1361--1385, 2015.

\end{thebibliography}
}
	

\end{document}